\documentclass{cocv}

\usepackage[utf8]{inputenc}
\usepackage{amsmath,amsfonts,amssymb,amsthm}
\usepackage[svgnames]{xcolor}
\usepackage{enumitem}
\usepackage{float}
\usepackage{caption}
\usepackage{algorithm}
\usepackage{algpseudocode}
\usepackage{nicematrix}
\usepackage{graphicx}%
\usepackage{multirow}%
\usepackage{mathrsfs}%
\usepackage[title]{appendix}%
\usepackage{xcolor}
\usepackage{textcomp}%
\usepackage{manyfoot}%
\usepackage{booktabs}%
\usepackage{listings}%
\usepackage{bm}  
\usepackage{mathtools}
\usepackage{tabularx}

\theoremstyle{thmstyleone}%
\newtheorem{theorem}{Theorem}
\newtheorem{proposition}[theorem]{Proposition}%
\newtheorem{corollary}[theorem]{Corollary}
\newtheorem{lemma}[theorem]{Lemma}

\theoremstyle{thmstyletwo}%
\newtheorem{example}{Example}%
\newtheorem{remark}{Remark}%

\theoremstyle{thmstylethree}%
\newtheorem{definition}{Definition}%

\numberwithin{equation}{section}

\newcommand{\N}{{\mathbb{N}}}
\newcommand{\R}{{\mathbb{R}}}
\newcommand{\C}{{\mathbb{C}}}
\newcommand{\Q}{{\mathbb{Q}}}
\newcommand{\Z}{{\mathbb{Z}}}

\newcommand{\Bc}{\mathcal{B}}
\newcommand{\Dc}{\mathcal{D}}
\newcommand{\Lc}{\mathcal{L}}

\newcommand{\Sc}{\mathcal{S}}
\newcommand{\Uc}{\mathcal{U}}

\DeclareMathOperator{\im}{im}
\DeclareMathOperator{\rk}{rk}
\newcommand{\ve}[1]{\bm{#1}}
\begin{document}

\title{Data-Based Analysis of Relative Degree and Zero Dynamics in Linear Systems}
\author{Janina Schaa}\address{Institut für Mathematik, Martin-Luther-Universit\"{a}t Halle-Wittenberg, Theodor-Lieser-Straße 5, 06120 Halle (Saale), Germany; \email{janina.schaa@mathematik.uni-halle.de, thomas.berger@mathematik.uni-halle.de}}
\author{Thomas Berger}\sameaddress{1} 

\begin{abstract}
Data-driven control offers a powerful alternative to traditional model-based methods, particularly when accurate system models are unavailable or prohibitively complex. While existing data-driven control methods primarily aim to construct controllers directly from measured data, our approach uses the available data to assess fundamental system-theoretic properties. This allows the informed selection of suitable control strategies without explicit model identification. We provide data-based conditions characterizing the (vector) relative degree and the stability of the zero dynamics, which are critical for ensuring proper performance of modern controllers. Our results cover both single- and multi-input/output settings of discrete-time linear systems.  We further show how a continuous-time system can be reconstructed from three sampling discretizations obtained via Zero-order Hold at suitable sampling times, thus allowing the extension of the results to the combined data collected from these discretizations.  All results can be applied directly to observed data sets using the proposed algorithms.
\end{abstract}



\maketitle

\section{Introduction}

Data-driven control methods provide a practical alternative to classical model-based approaches, particularly when accurate system models are unavailable or too complex to allow effective controller design \cite{Persis.2023, J.Coulson.2019, Hou.2013}. Using an inaccurate model can lead to poor performance, complete controller failure, or the inability to bound errors \cite{Hou.2013}, motivating the development of controllers designed directly from measured data. Early data-driven control methods concentrated on tuning the parameters of PID controllers~\cite{Ziegler.1942, Desborough.2002}, but poor performance \cite{Ender.1993} motivates modern data-driven strategies. Recent work in data-driven control has largely focused on constructing controllers from either input–state \cite{Bisoffi.2022} or input–output trajectories \cite{J.Berberich.2021}. 

In contrast, in the present work we derive conditions under which a measured data set of input–output sequences contains sufficient information to establish certain system-theoretic properties of the underlying system. Knowledge of properties such as the (vector) relative degree or the stability of zero dynamics is critical for the application of modern control strategies like funnel control \cite{Ilchmann.2002,Berger.2025} or sliding mode control \cite{ShteEdwa14,SiraRamirez.2015}. Hence, based on the data-driven investigation of these properties, we can determine a suitable control law for the system in question.

Our analysis builds on the Fundamental Lemma by Willems et al.~\cite{J.C.Willems.2004}, which provides conditions under which a given set of input-output data suffices to uniquely identify a linear time-invariant system. Using this result, one can define data-dependent matrices which are equivalent to the true system model \cite{C.DePersis.2020}, enabling controller design. We utilize the behavioral approach \cite{PoldWill98} and the concept of most powerful unfalsified models \cite{Heij.1990, Willems.1997} and extend the data informativity framework \cite{Waarde.2020b} to the (vector) relative degree and the stability of the zero dynamics of linear systems. We restrict ourselves to the case of noise-free data.

Specifically, we establish conditions on sampled data that are equivalent to informativity for the system’s relative degree and informativity for the stability of the zero dynamics, and present algorithms to verify these conditions for a given data set.  The structure of the paper is as follows. Section~\ref{secLinearSystems} introduces the central concepts. Section~\ref{secRelativeDegreeandZD} presents the main results on the relative degree, while Section~\ref{secZD} presents the main results on the stability of the zero dynamics and a numerical example of the proposed algorithms. Section~\ref{secContinuousTime} discusses the extension of the results to continuous-time systems, reconstructed from sampling discretizations obtained via Zero-order Hold. Section~\ref{secConclusion} concludes the paper.

\textbf{Notation.}\ 
In the following, let $\N$ denote the natural numbers, $\N_0 = \N \cup\{0\}$, $\Z$ the integers, $\Q$ the rational numbers and $\R$ the real numbers. For a matrix $A\in\R^{n\times m}$, let $A^\top$, $\im A$, $\ker A$ denote its transpose, image and kernel, resp., as well as $A_i$ its  $i$-th row for $i=1,\ldots,n$. For indices $i \leq j$, the notation $A_{i,\hdots,j}$ denotes the submatrix consisting of rows $i$ through $j$, and $A_{i,\hdots,j}^{\top}$ denotes its transpose. The symbol $A^\dagger$ denotes the Moore-Penrose pseudoinverse of $A$. Moreover, for easy reference, the following list collects some notation, which is frequently used throughout the paper, but defined at later points.\\

\begin{tabularx}{\linewidth}{ >{$}r<{$} >{\raggedright\arraybackslash}X } \toprule 
\mathcal{B} & behavior \\ 
\mathcal{L}^w & class of complete discrete-time LTI systems \\ 
\mathcal{C}^{m,p} & class of continuous-time LTI systems \\ 
(\mathbf{A},\mathbf{B},\mathbf{C},\mathbf{D}) & state-space realization \\ 
l(\mathcal{B}) & lag of the behavior $\mathcal{B}$ \\
m(\mathcal{B}) & input dimension of the behavior $\mathcal{B}$ \\ 
p(\mathcal{B}) & output dimension of the behavior $\mathcal{B}$ \\ 
n(\mathcal{B}) & McMillan degree of the behavior $\mathcal{B}$ \\ 
r(\mathcal{B}) & vector relative degree of the behavior $\mathcal{B}$ \\ 
\mathcal{B}_{\mathrm{MPUM}} & most powerful unfalsified model \\ 
\mathcal{B}_{\mathrm{ZD}} & zero dynamics \\ 
\mathcal{O}_L(\mathbf{C},\mathbf{A}) & truncated observability matrix of depth $L$ \\
\mathbf{H} & impulse response \\
\mathcal{H}_k(\ve{f}) & Hankel matrix generated by $\ve{f}$ \\
\mathcal{T}_k(\ve{f}) & Toeplitz matrix generated by $\ve{f}$\\ 
\bottomrule \end{tabularx}
 
\section{Preliminaries: Linear systems} \label{secLinearSystems}

\subsection{Behavioral approach and integer invariants}\label{Ssec:behaviors}

We recall some basic concepts of the behavioral approach from~\cite{Willems.1986a}. 

A dynamical system $\mathcal{S} = (T,W,\mathcal{B})$ is determined by the set of possible trajectories $\mathcal{B} \subseteq W^T = \{ f: T \rightarrow W \}$, called the behavior of the system, where $T \subseteq \R$ is a time set and $W$ is the signal alphabet. Notice that this approach does not differ between input and output variables. Initially, we will consider the case of discrete-time systems, hence $T= \N_0$, and a finite dimensional signal space $W = \R^w$. Denote the restriction of the behavior to a finite interval $[a,b]$ as 
$$\mathcal{B}_{[a,b]} \coloneqq \{ \ve{\tilde{w}_r} \in W^{b-a+1} :\ \exists\, \ve{\tilde{w}} \in \mathcal{B} \text{ with } \ve{\tilde{w}}_{\vert [a,b]} = \ve{\tilde{w}_r} \}.$$

Let $\sigma$ denote the backwards shift operator on the set $W^{\N_0}$, defined for each $w \in W^{\N_0}$ by $(\sigma w)(t) = w(t+1)$ for every $t \in \N_0$. A dynamical system $\mathcal{S} = (\N_0, \R^w, \mathcal{B})$ is called time-invariant, if the behavior $\mathcal{B}$ is invariant under application of the shift operator, $\sigma \mathcal{B} \subseteq \mathcal{B}$, and called linear, if the behavior $\mathcal{B}$ is a vector space. $\mathcal{S}$ is called complete, if the behavior $\mathcal{B}$ is closed in the topology of pointwise convergence. Denote the set of linear, time-invariant, complete systems with a $w$-dimensional signal space by $\mathcal{L}^w$. Since the system is determined by its behavior, we also write $\mathcal{B} \in \mathcal{L}^w$ instead of $\mathcal{S} = (\N_0, \R^w, \mathcal{B}) \in \mathcal{L}^w$, and the terms system and behavior will be used interchangeably.

Recall the concept of free subspaces from~\cite{J.C.Willems.2004}. 
The coordinates in a free subspace of a behavior are not restricted by the past of the trajectory, since any future sequence is a valid continuation of the trajectory. This property does also describe input variables, as they cannot be predicted by the past behavior of a causal system. The maximal dimension of a free subspace of the behavior $\mathcal{B}$ is denoted by $m(\mathcal{B})$ and is called the input cardinality of the system. The output cardinality of the system $\mathcal{B} \in \mathcal{L}^w$ is denoted as $p(\mathcal{B}) = w - m(\mathcal{B})$.

As another concept, recall that linear, time invariant and complete systems admit a state space representation as shown in \cite{Willems.1986a}. That is, every system $\mathcal{S} \in \mathcal{L}^w$ can be represented as an input/state/output system parametrized by the matrices $\ve{A} \in \R^{n \times n}, \ve{B} \in \R^{n \times m}, \ve{C} \in \R^{p \times n}$ and $\ve{D} \in \R^{p \times m}$, and denoted by $(\ve{A},\ve{B},\ve{C},\ve{D})$, where $W = \R^{m+p}, m = m(\mathcal{B})$ and $p = p(\mathcal{B})$, by defining the behavior $\mathcal{B}$ as the set of trajectories $(\ve{u},\ve{y}) \in (\R^{m+p})^{\N_0}$ satisfying the system equations
\begin{equation} \label{isosystem}
\begin{aligned}
\ve{x}(t+1) &= \ve{A}\ve{x}(t) + \ve{B}\ve{u}(t)\\
\ve{y}(t) &=  \ve{C}\ve{x}(t) +  \ve{D}\ve{u}(t)
\end{aligned}
\end{equation}
for some state trajectory $\ve{x} \in (\R^n)^{\N_0}$. Such a representation is called minimal, if the dimension $n$ of the state space is minimal in the set of possible state space representations of the behavior. This minimal state space dimension of the behavior $\mathcal{B}$ is denoted by $n(\mathcal{B})$ and is called the McMillan degree of the system.

\begin{definition} 
A system $\mathcal{B} \in \mathcal{L}^w$ is called controllable, if for any finite past trajectory $\ve{w_{p}} \in \mathcal{B}_{[0,T]}$ for some $T \in \N$, and any future trajectory $\ve{w_f} \in \mathcal{B}$, there is a time distance $\tilde{T} \in \N$ such that there exists a common trajectory $\ve{v} \in \mathcal{B}$ satisfying $\ve{v}_{\vert [0,T]} = \ve{w_p}$ and $\ve{v}(t) = \ve{w_f}(t-T-\tilde{T})$ for any $t \geq T + \tilde{T}$.
\end{definition}

Note that not every state space representation of a controllable behavior is state controllable (that is, $\rk  \begin{bmatrix}  \ve{B} &  \ve{AB} & \hdots &  \ve{A}^{n-1}  \ve{B} \end{bmatrix} = n$), since it is possible that there are additional uncontrollable and unobservable states that do not change the input/output behavior of the system. Since these can be canceled without changing the input/output behavior of the system, every minimal state space representation of a controllable behavior is state controllable, as discussed in~\cite{Markovsky.2021}. Furthermore, the following connection between the concepts of minimality and observability is shown in~\cite{PoldWill98}.

\begin{lemma} \label{minimalObservable}
Every minimal state space representation $(\ve{A},\ve{B},\ve{C},\ve{D})$ of a system $\mathcal{B} \in \mathcal{L}^w$ is observable, that is $\rk \begin{bmatrix}  \ve{C}^{\top} & ( \ve{CA})^{\top} & \hdots & ( \ve{CA}^{n(\mathcal{B})-1})^{\top} \end{bmatrix}^{\top} = n(\mathcal{B})$.
\end{lemma}

Another important integer invariant of a dynamical system $\mathcal{B} \in \mathcal{L}^w$ is its lag $l(\mathcal{B})$, defined as the smallest integer $L \in \N_0$ such that the truncated observability matrix $\mathcal{O}_L( \ve{C}, \ve{A}) = \begin{bmatrix}  \ve{C}^{\top} & ( \ve{CA})^{\top} & \hdots & ( \ve{CA}^{L-1 })^{\top} \end{bmatrix}^{\top}$ of a minimal state space representation $(\ve{A},\ve{B},\ve{C},\ve{D})$ of suitable sizes has full rank $n(\mathcal{B})$, see~\cite{IvanMarkovsky.2008}. Note that this notion is well-defined, as it is straightforward to see that two differing minimal representations yield the same lag.

The following lemma from \cite{IvanMarkovsky.2008} shows that the initial state of a fixed minimal representation of a system is uniquely determined by the observation of the input/output behavior of the system for at least $l(\mathcal{B})$ time steps. This enables us to characterize the state space by using input and output sequences of length $l(\mathcal{B})$ and will therefore play a crucial role for identifying systems from their input/output behavior.

\begin{lemma} \label{InitialCondition} 
Let $\mathcal{B} \in \mathcal{L}^w$ be a system, fix a minimal state space representation and an initializing sequence $(\ve{u},\ve{y}) \in \mathcal{B}_{[0,T-1]}$ of length $T \geq l(\mathcal{B})$. Then the resulting initial state $\ve{x}(T) $ of the state space representation is uniquely determined by the initializing sequence.
\end{lemma}

Next we characterize under which conditions the behaviors defined by different state space representations are equivalent.

\begin{definition}
The impulse response $\ve{H}$ of a state space representation $(\ve{A},\ve{B},\ve{C},\ve{D})$ is defined as the sequence $\big(\ve{D}, (\ve{CA}^{k}\ve{B})_{k \in \N_0}\big)$. 

Two state space representations $(\ve{A},\ve{B},\ve{C},\ve{D})$ and $(\ve{A'},\ve{B'}, \ve{C'},\ve{D'})$ are called behaviorally equivalent, if they induce the same behavior $\mathcal{B}$, meaning that there are integers $n,n' \in \N_0$ such that for each trajectory $(\ve{u},\ve{y}):\N_0 \rightarrow \R^{m+p}$ that is contained in $\mathcal{B}$, there are two state trajectories $\ve{x}: \N_0 \rightarrow \R^n$ and $\ve{x'}: \N_0 \rightarrow \R^{n'}$ such that $(\ve{u},\ve{x},\ve{y})$ satisfies (\ref{isosystem}) for the matrices $\ve{A},\ve{B},\ve{C},\ve{D}$ and $(\ve{u},\ve{x'},\ve{y})$ satisfies (\ref{isosystem}) for the matrices $\ve{A'},\ve{B'},\ve{C'},\ve{D'}$. 
\end{definition}

Observe that it is possible for representations with differing state space dimensions to be behaviorally equivalent. This stresses the importance of the concept of minimal representations. The following characterization of behavioral equivalence is a corollary of the decomposition of linear behaviors into controllable and autonomous parts developed in~\cite[Ch.5.2]{PoldWill98}.

\begin{lemma} \label{behaviorallyEquivalent}
Two behaviors $\mathcal{B}, \mathcal{B}' \in \mathcal{L}^w$ are equal if, and only if, for any two state space representations of the respective systems with state space dimensions $n, n' \in \N_0$, the impulse responses are equal, and for each initial state $\ve{x}(0) \in \R^{n}$ there exists an initial state $\ve{x'}(0) \in \R^{n'}$ such that the zero input responses of the systems, starting at the respective initial states, are equal.
\end{lemma}

Another equivalent way of characterizing behaviors $\mathcal{B} \in \mathcal{L}^w$ connected to autoregressive models of time series are kernel representations given by a polynomial matrix $\ve{R}(s) \in \R^{g \times q}[s]$ for some $g,q \in \N$ via $$ \mathcal{B} = \ker \ve{R}(\ve{\tilde{\sigma}}) $$ where $\ve{\tilde{\sigma}}$ denotes the shift operator, $\ve{\tilde{\sigma}}(\ve{f})(t) = \ve{f}(t-1)$ for all $t \in \N$ and all $\ve{f}:\N_0 \rightarrow \R^d$, as shown in \cite{Willems.1986a}.

All $\mathcal{B} \in \mathcal{L}^w$ have a kernel representation, and every behavior defined by a kernel representation is contained in $\mathcal{L}^w$, see \cite{Willems.1986a}. It is straightforward to show that the definition of the lag, which we adopted from~\cite{IvanMarkovsky.2008}, is equivalent to other possible characterizations from~\cite{Maupong.2017} and~\cite{J.C.Willems.2004}, respectively. 

\begin{proposition} \label{CharakterisierungenLag}
Let $\mathcal{B} \in \mathcal{L}^w$. 
The following are equivalent characterizations of the lag $l(\mathcal{B})$.
\begin{itemize}
\item[(i)] $l(\mathcal{B})$ is the minimal integer $L \geq 0$ such that for any sequence $\ve{v} \in (\R^w)^{\N_0}$, if $\ve{v}_{\vert [t, t+L]} \in \mathcal{B}_{[0,L]}$ for all $t \in \N_0$, then $\ve{v} \in \mathcal{B}$.
\item[(ii)] $l(\mathcal{B})$ is the minimal integer $L \geq 0$ such that there exists a kernel representation $\mathcal{B} = \ker(\ve{R}(\ve{\tilde{\sigma}}))$ and all entries of $\ve{R}(s)$ have degree $\leq L$. 
\end{itemize}
\end{proposition}

\subsection{Identifiability}

Our main objective is to identify certain relevant properties of a system based on a finite sequence of sampled data. The famous Fundamental Lemma by Willems et al.~\cite{J.C.Willems.2004} treats the question, when it is possible to uniquely identify a behavior $\mathcal{B} \in \mathcal{L}^w$ in the case of known input and output data. Recall that for a signal $\ve{u}:[0,T] \rightarrow \R^m$ of finite length, the associated Hankel matrix of order $L$ for some $L \leq T + 1$ is defined as 

$$ \mathcal{H}_L(\ve{u})  = \begin{bmatrix}  \ve{u}(0) & \ve{u}(1) & \hdots & \ve{u}(T-L+1) \\ \ve{u}(1) & \ve{u}(2) & \hdots & \ve{u}(T-L+2) \\ \vdots & \vdots & \ddots & \vdots \\ \ve{u}(L-1) & \ve{u}(L) & \hdots & \ve{u}(T) \end{bmatrix}. $$

The signal is called persistently exciting of order $L$, if $\rk \mathcal{H}_L(\ve{u}) = Lm$. 

\begin{theorem}[Fundamental Lemma~\cite{J.C.Willems.2004}] \label{Fundamentallemma}
Let $\mathcal{B} \in \mathcal{L}^w$ be controllable, and $(\ve{u_d}, \ve{y_d}) = \ve{w_d} : [0,T] \rightarrow \R^w$ be a sampled trajectory of the system, $\ve{w_d} \in \mathcal{B}_{[0,T]}$, such that $\ve{u_d}:[0,T] \rightarrow \R^{m({\mathcal{B}})}$ is persistently exciting of order $L+n(\mathcal{B})$. Then
$$ \mathcal{B}_{[0,L-1]} = \im \mathcal{H}_L(\ve{w_d}) .$$
\end{theorem}

This result shows that all possible trajectories of length $L$ can be obtained as linear combinations of the windows of length $L$ from the sampled data sequence. Notice that for any $L \geq l(\mathcal{B}) +1$, the behavior is entirely characterized by these sequences, as shown in Proposition \ref{CharakterisierungenLag}. Hence, Theorem \ref{Fundamentallemma} shows that the entire system behavior can be uniquely identified from the information gained from a single data sequence, as long as the system is controllable and the input sequence is persistently exciting of order at least $n(\mathcal{B} ) + l(\mathcal{B}) + 1$.

A combination of Theorem~\ref{Fundamentallemma} with Lemma \ref{InitialCondition} leads to conditions under which we can uniquely determine the future output trajectory from a past initializing sequence and the future input sequence, as shown in~\cite[Prop.~6]{IvanMarkovsky.2008}. 

\begin{corollary} \label{uniqueOutputInitialized} 
Let $\mathcal{B} \in \mathcal{L}^w$ be controllable and $(\ve{u_d},\ve{ y_d}) = \ve{w_d} : [0,T] \rightarrow \R^w$ be a sampled trajectory of the system, $\ve{w_d} \in \mathcal{B}_{[0,T]}$, where $\ve{u_d}:[0,T] \rightarrow \R^{m({\mathcal{B}})}$ is persistently exciting of order $T_{p}+T_{f}+n(\mathcal{B})$. Then for any initializing sequence $(\ve{u_{p}}, \ve{y_{p}}) \in \mathcal{B}_{[0,T_p - 1]}$ and any future input sequence $\ve{u_f} \in (\R^{m(\mathcal{B})})^{T_f}$, there exist $\ve{g} \in \R^{T-T_f-T_p + 2} $ and $\ve{y} \in (\R^{p(\mathcal{B})})^{T_f}$ such that 
$$ \begin{bmatrix} \mathcal{H}_{T_p + T_f}(\ve{u_d})  \\ \mathcal{H}_{T_p + T_f}(\ve{y_d}) \end{bmatrix} \ve{g}  =  \begin{pmatrix} \ve{u_p} \\ \ve{u_f} \\ \ve{y_p} \\ \ve{y} \end{pmatrix},$$
and if $T_p \geq l(\mathcal{B})$, the future output vector $\ve{y}$ is uniquely determined by $\ve{u_p}, \ve{y_p}$, and $\ve{u_f}$.
\end{corollary}

\subsection{Data informativity}

While the Fundamental Lemma is a powerful tool, it is often possible to identify certain properties from data when the behavior is not controllable or the input sequence is not persistently exciting, and hence the behavior cannot be fully identified. In order to treat such cases, the framework of data informativity has been introduced in~\cite{Waarde.2020b} and we recall the basic concepts here. In the following, we assume that the partition into input and output variables $w = m+p$ is known.

\begin{definition}\label{Def:informativity}
Let $\Sigma\subseteq \mathcal{L}^w$ be a model class containing the true system $\mathcal{B}$, and let $\mathcal{D}\in\R^{kw}$ be a given data set. Then $\Sigma_{\mathcal{D}}$ denotes the set of all systems from $\Sigma$ that can explain the observed data $\mathcal{D}$, that is, there exists some trajectory $f \in \mathcal{B}$ such that $\mathcal{D} = (f(0)^{\top}, \hdots, f(k-1)^{\top})^{\top}$, and are of the assumed input dimension $m(\mathcal{B}) = m$.\\
Let $\mathcal{P}$ be a system-theoretic property, and $\Sigma_{\mathcal{P}}$ be the set of all systems from $\Sigma$ that have the property $\mathcal{P}$. The data set $\mathcal{D}$ is called informative for the property $\mathcal{P}$, if $ \Sigma_{\mathcal{D}} \subseteq \Sigma_{\mathcal{P}}$.
\end{definition}

Another important tool is the most powerful unfalsified model, which goes back to the work of Willems~\cite{Willems.1986a,Willems.1986b,Willems.1987}.

\begin{definition}
In the model class $\mathcal{L}^{m+p}$, in which we assume that the partition of variables into input and output variables is known, we can define the most powerful unfalsified model (MPUM) for a given data set $\mathcal{D}\in\R^{k(m+p)}$ as the smallest behavior that explains the data, $\mathcal{B}_{MPUM} = \bigcap_{\mathcal{B} \in \Sigma_{\mathcal{D}}} \mathcal{B} $.
\end{definition}

\begin{remark}
Observe that the maximal dimension of a free subspace of the Most Powerful Unfalsified Model may be strictly smaller than the maximal dimension of a free subspace of the true system, yielding $m(\mathcal{B}_{MPUM}) \leq m$ and $\mathcal{B}_{MPUM} \notin \Sigma_{\mathcal{D}}$, as illustrated by the following example. Consider the data sequence $$ \ve{u_d} = \begin{pmatrix} 0 & 1 & 0 & 0 & 0 & 0 \\0 & 0 & 0 & 0 & 0 & 0 \end{pmatrix},\quad  \ve{y_d} = \begin{pmatrix} 0 & 0 & 1 & 1 & 1 & 1 \\0 & 0 & 0 & 0 & 0 & 0 \end{pmatrix}.$$ The two behaviors $\mathcal{B}_1,\mathcal{B}_2$ determined by the state space representations $(\ve{A_i},\ve{B_i},\ve{C_i},\ve{D_i})$ where $\ve{A_i} = \begin{bmatrix} 1 & 0 \\ 0 & 1 \end{bmatrix} = \ve{C_i} = \ve{B_1}$, $\ve{D_i} = 0$, and $\ve{B_2} = \begin{bmatrix} 1 & 0 \\ 0 & 2 \end{bmatrix}$ can explain the observed data, and are both of input dimension $m(\mathcal{B}_i) = 2$ for $i=1,2$. However, we can observe that the intersection of the two behaviors does not contain a trajectory where the second input coordinate has nontrivial entries. Hence, this coordinate is not contained in a free subspace, and we can conclude that $m(\mathcal{B}_1 \cap \mathcal{B}_2) = 1$.
\end{remark}

For a given data sequence $(\ve{u_d},\ve{y_d}):[0,T-1]\to\R^{m+p}$, the set of all systems explaining the data $\Dc = (\ve{u_d},\ve{y_d}) =(\ve{u_d}(0),\ldots, \ve{u_d}(T-1),\ve{y_d}(0),\ldots, \ve{y_d}(T-1))\in\R^{2Tw}$ is 
\[
    \{ \mathcal{B} \in \Sigma: (\ve{u_d}, \ve{y_d}) \in \mathcal{B}_{[0,T-1]} \}.
\]
If the data sequence is sampled from a system with known lag $l(\mathcal{B})$, we can determine the MPUM.

\begin{lemma} \label{DetermineMPUM}
Consider $\Bc\in \Sigma = \Lc^w$ with known lag $l(\mathcal{B})$ and let $\mathcal{D} = (\ve{u_d},\ve{y_d})$ be a sampled data sequence of length $T \geq l(\mathcal{B}) + 1$. Then the most powerful unfalsified model $\mathcal{B}_{MPUM}$ for this data sequence is uniquely determined by
$$ \mathcal{B}_{MPUM, [0,l(\mathcal{B})]} = \im \begin{bmatrix} \mathcal{H}_{l(\mathcal{B}) + 1} (\ve{u_d}) \\ \mathcal{H}_{l(\mathcal{B}) + 1} (\ve{y_d}) \end{bmatrix}. $$
\end{lemma}

\begin{proof}
We use Proposition \ref{CharakterisierungenLag} in order to notice that $l(\mathcal{B}_{MPUM}) $ is bounded by the maximal degree of the polynomials in any kernel representation of this behavior. We can construct a kernel representation of $\mathcal{B}_{MPUM}$ by using rows of kernel representations with minimal degrees of possible behaviors $\tilde{\mathcal{B}} \in \Sigma_{\mathcal{D}}$, hence the maximal degree of the rows of the kernel representation of $\mathcal{B}_{MPUM}$ will not exceed the maximal degree $l(\tilde{\mathcal{B}})$ of the kernel representations of any possible behavior that explains the data. This shows that $l(\mathcal{B}_{MPUM}) \leq l(\mathcal{B})$. 

By Proposition \ref{CharakterisierungenLag} we further note that characterizing all sequences of length $l(\mathcal{B}) + 1$ of $\mathcal{B}_{MPUM} $ is sufficient in order to characterize the entire behavior uniquely. It is then obvious that 
$$ \im \begin{bmatrix} \mathcal{H}_{l(\mathcal{B}) + 1} (\ve{u_d}) \\ \mathcal{H}_{l(\mathcal{B}) + 1} (\ve{y_d}) \end{bmatrix} \subseteq \bigcap_{\mathcal{B}' \in \Sigma_{\mathcal{D}}} \mathcal{B}'_{[0,l(\mathcal{B})]} = \mathcal{B}_{MPUM,[0,l(\mathcal{B})]},$$ since all sequences that are obtained as a linear combination of the sampled data sequences must be contained in every possible behavior $\mathcal{B}' \in \mathcal{L}^w$ that explains the data. On the other hand, the reverse inclusion is true, because otherwise the behavior $\mathcal{B}'$ defined by $\mathcal{B}'_{[0,l(\mathcal{B})]} = \im \begin{bmatrix} \mathcal{H}_{l(\mathcal{B}) + 1} (\ve{u_d}) \\ \mathcal{H}_{l(\mathcal{B}) + 1} (\ve{y_d}) \end{bmatrix} \subsetneq \mathcal{B}_{MPUM,[0,l(\mathcal{B})]}$ would characterize a behavior $\mathcal{B}'\in \Sigma_\Dc$ that is strictly contained in $\mathcal{B}_{MPUM}$ , which contradicts its definition.
\end{proof}

Since the behavior $\mathcal{B}_{MPUM}$ is uniquely determined by the data $\mathcal{D}$, we can use this data to explicitly characterize sequences of length $l(\mathcal{B}) + k$. 

\begin{lemma} \label{ExtendedSequencesForm}
Consider a true model $\Bc\in\Sigma = \Lc^w$ with known lag $l(\mathcal{B})$ and let $\mathcal{D} = (\ve{u_d},\ve{y_d})$ be a sampled data sequence of length $T \geq l(\mathcal{B}) + 1$. Further let 
\begin{align*}
\ve{U_{-}} &= \mathcal{H}_{l(\mathcal{B}) + 1}(\ve{u_d})_{1, \hdots, l(\mathcal{B})}, &\ve{U_{+}} = \mathcal{H}_{l(\mathcal{B}) + 1} (\ve{u_d})_{2, \hdots , l(\mathcal{B}) + 1}, \ \ve{U_{l(\mathcal{B}) + 1}} &= \mathcal{H}_{l(\mathcal{B}) + 1} (\ve{u_d})_{l(\mathcal{B}) + 1}, \\
\ve{Y_{-}} &= \mathcal{H}_{l(\mathcal{B}) + 1} (\ve{y_d})_{1, \hdots, l(\mathcal{B})}, &\ve{Y_{+}} = \mathcal{H}_{l(\mathcal{B}) + 1} (\ve{y_d})_{2, \hdots, l(\mathcal{B}) + 1}, \ \ve{Y_{l(\mathcal{B}) + 1}} &= \mathcal{H}_{l(\mathcal{B}) + 1} (\ve{y_d})_{l(\mathcal{B}) + 1}.
\end{align*}
Then for each $k \leq T - l(\mathcal{B}) - 1$, the sequences of length $l(\mathcal{B}) +  k +1$ can be determined via 
$$ \mathcal{B}_{MPUM, [0, l(\mathcal{B}) + k]} =  \ve{X_k} \ker \ve{Z_k},$$
where $\ve{X_k} \in \R^{w(l(\mathcal{B}) + k +1) \times (k+1) (T - l(\mathcal{B}))} \text{ and } \ve{Z_k} \in \R^{w((k+1)l(\mathcal{B})) \times (k+1)(T - l(\mathcal{B}))}$
are given by
$$ \ve{X_k} = \begin{bmatrix} \ve{U_-} & \ve{0} & \hdots & \ve{0} \\ \ve{U_{l(\mathcal{B}) + 1}} & \ve{0} & \hdots & \ve{0} \\ \ve{0} & \ve{U_{l(\mathcal{B}) + 1}} & \ve{0} & \vdots \\ \vdots & \ddots & \ddots & \vdots \\ \ve{0} & \hdots & \ve{0} & \ve{U_{l(\mathcal{B}) + 1}} \\ \ve{Y_-} & \ve{0} & \hdots & \ve{0} \\ \ve{Y_{l(\mathcal{B}) + 1}} & \ve{0} & \hdots & \ve{0} \\ \ve{0} & \ve{Y_{l(\mathcal{B}) + 1}} &\ve{ 0} & \vdots \\ \vdots & \ddots & \ddots & \vdots \\ \ve{0} & \hdots & \ve{0} & \ve{Y_{l(\mathcal{B}) + 1}} \end{bmatrix} \
\text{ and } \ \ve{Z_k} = \begin{bmatrix}  \ve{U_+} & -\ve{U_-} & \ve{0} & \hdots & \ve{0} \\ \ve{0} & \ve{U_+} & -\ve{U_-} & \ve{0} & \vdots \\ \vdots & \vdots & \ddots & \ddots & \vdots \\ \ve{0} & \hdots & \ve{0} & \ve{U_+} & -\ve{U_-} \\ \ve{Y_+} & -\ve{Y_-} & \ve{0} & \hdots & \ve{0} \\ \ve{0} & \ve{Y_+} & -\ve{Y_-} & \ve{0} & \vdots \\ \vdots & \vdots & \ddots & \ddots & \vdots \\ \ve{0} & \hdots & \ve{0} & \ve{Y_+} & -\ve{Y_-}  \end{bmatrix} .$$
\end{lemma}

\begin{proof}
Let $(\ve{g_1}^{\top}, \hdots, \ve{g_{k+1}}^{\top})^{\top} = \ve{g} \in \ker \ve{Z_k}$, where $\ve{g_i} \in \R^{T - l(\mathcal{B})}$. Observe that by the structure of $\ve{Z_k}$, we get $\ve{U_+ g_i} = \ve{U_- g_{i+1}}$ and $\ve{Y_+ g_i} =\ve{ Y_- g_{i+1}}$ for each $1 \leq i \leq k$. We want to show that all windows of the sequence of length $l(\mathcal{B}) + 1$ are contained in $\mathcal{B}_{MPUM,[0,l(\mathcal{B})]}$, which equals $\im \begin{bmatrix} \ve{U_-} \\ \ve{U_{l(\mathcal{B}) + 1}} \\ \ve{Y_-} \\ \ve{Y_{l(\mathcal{B}) + 1}} \end{bmatrix}$ by Lemma \ref{DetermineMPUM}.  Consider $\ve{U_1} = \mathcal{H}_{l(\mathcal{B}) + 1}(\ve{u_d})_1$ and $\ve{Y_1} = \mathcal{H}_{l(\mathcal{B}) + 1} (\ve{y_d})_1$ and define $$\ve{X_{k,i}} =
\scalebox{0.85}{$  
\begin{bNiceMatrix}
\ve{U_1}^{\top}	& \ve{0} 		& \hdots	& \hdots	& \ve{0} 					& \ve{0}		& \ve{0}					& \ve{Y_1}^{\top}	& \ve{0} 		& \hdots 	& \hdots 	& \ve{0} 					& \ve{0} 		& \ve{0} \\
\ve{0} 		& \ddots 	& \ve{0} 		& \ve{0} 		& \hdots				&  \ddots 	& \hdots				& \ve{0} 		& \ddots 	& \ve{0} 		& \ve{0} 		& \hdots 				&  \ddots 	& \hdots \\ 
\vdots 	& \ve{0} 		& \ve{U_1}^{\top} 	& \ve{U_+}^{\top} 	& \ve{0} 					& \ddots 	& \hdots 				& \vdots 	& \ve{0} 		& \ve{Y_1}^{\top} 	& \ve{Y_+}^{\top} 	& \ve{0} 					& \ddots 	& \hdots  \\
\vdots 	& \vdots 	& \ve{0} 		& \ve{0} 		& \ve{U_{l(\mathcal{B}) + 1}}^{\top} 	& \ve{0} 		& \hdots 				& \vdots 	& \vdots 	& \ve{0} 		& \ve{0} 		& \ve{Y_{l(\mathcal{B}) + 1}}^{\top} 	& \ve{0} 		& \hdots  \\
\vdots 	& \vdots 	& \vdots 	& \ddots 	& \ddots 				& \ddots 	& \ve{0} 					& \vdots 	& \vdots 	& \vdots 	& \ddots 	& \ddots 				& \ddots 	& \ve{0} \\
\ve{0} 		& \ve{0} 		& \ve{0} 		& \ve{0} 		& \ve{0} 					& \ve{0} 		& \ve{U_{l(\mathcal{B}) + 1}}^{\top} 	& \ve{0} 		& \ve{0} 		& \ve{0} 		& \ve{0} 		& \ve{0} 					& \ve{0} 		& \ve{Y_{l(\mathcal{B}) + 1}}^{\top}
\end{bNiceMatrix}^{\top} $} $$
for each $1 \leq i \leq l(\mathcal{B}) + 1$, where the rows containing $\ve{U_+}^{\top}$ and $\ve{Y_+}^{\top}$ are those with the indices $(i-1) (T-  l(\mathcal{B})) +1, \hdots ,i(T-l(\mathcal{B}))$. Observe that for each $\ve{g} \in \ker \ve{Z_k}$, the structure of $\ve{Z_k}$ shows that $\ve{X_k} \ve{g} = \ve{X_{k,i}} \ve{g} $ for each $1 \leq i \leq k+1$, by inductively using $\ve{U_+ g_i} = \ve{U_- g_{i+1}}$ and $\ve{Y_+ g_i }= \ve{Y_- g_{i+1}}$ for each $1 \leq i \leq k$.

Further, we notice that any window $(\ve{u_i}, \hdots \ve{u_{i + l(\mathcal{B})}})$ of length $l(\mathcal{B}) + 1$ of an arbitrary element $\ve{u} \in \ve{X_k } \ker \ve{Z_k}$ can be determined as 
\begin{align*}
 \begin{bmatrix} \ve{X_{k; i, \hdots, i+l(\mathcal{B})}} \\ \ve{X_{k; l(\mathcal{B})+k+1+i, \hdots , 2l(\mathcal{B}) + k + 1 + i}} \end{bmatrix} \ve{g} &= \begin{bmatrix} \ve{X_{k,i; i, \hdots, i + l(\mathcal{B})}} \\ \ve{X_{k,i; l(\mathcal{B}) + k + 1 + i, \hdots , 2l(\mathcal{B}) + k + 1 + i}} \end{bmatrix} \ve{g} \\
 &= \begin{bmatrix} \ve{U_1} \\ \ve{U_+} \\ \ve{Y_1} \\ \ve{Y_+} \end{bmatrix} \ve{g_i} \in \mathcal{B}_{MPUM,[0,l(\mathcal{B})]}.
\end{align*}
This shows that all windows of length $l(\mathcal{B}) + 1$ of any element of $\ve{X_k} \ker \ve{Z_k}$ are contained in $\mathcal{B}_{MPUM,[0,l(\mathcal{B})]}$, hence $\ve{X_k} \ker \ve{Z_k} \subseteq \mathcal{B}_{MPUM,[0,l(\mathcal{B}) + k]}$. In order to obtain the other inclusion, choose an arbitrary sequence $(\ve{u},\ve{y}) \in \mathcal{B}_{[0,l(\mathcal{B}) + k]}$. By Lemma \ref{DetermineMPUM}, there exist $\ve{g_i}   \in \R^{T - l(\mathcal{B})}$ such that $(\ve{u_i}, \hdots, \ve{u_{i+l(\mathcal{B})}}, \ve{y_i}, \hdots, \ve{y_{i+l(\mathcal{B})}})= \begin{bmatrix} \ve{U_-} \\ \ve{U_{l(\mathcal{B}) + 1}} \\ \ve{Y_-} \\ \ve{Y_{l(\mathcal{B}) + 1}} \end{bmatrix} \ve{g_i}$ for each $1 \leq i \leq k+1$. Define $\ve{g} = (\ve{g_1}^{\top}, \hdots , \ve{g_{k+1}}^{\top})^{\top}$. Then $(\ve{u},\ve{y}) = \ve{X_k g}$ by construction, and $\ve{g} \in \ker \ve{Z_k}$, since for each $1 \leq i \leq k$, we have that $\ve{U_{+} g_i} = (\ve{u_{i+1}} , \hdots , \ve{u_{i+l(\mathcal{B})}}) = \ve{U_- g_{i+1}}, $ and $\ve{Y_+ g_i} = (\ve{y_{i+1}}, \hdots , \ve{y_{i + l(\mathcal{B})}}) = \ve{Y_{i + l(\mathcal{B})}}$. This yields $(\ve{u},\ve{y}) \in \ve{X_k} \ker \ve{Z_k}$ and completes the proof. 
\end{proof}

\section{Relative degree of linear systems} \label{secRelativeDegreeandZD}

In this section, the objective is to identify the (vector) relative degree (cf.~\cite{Isidori.1995}) of a linear system from a sampled data sequence. Roughly speaking, the relative degree is an integer invariant of a system that characterizes the delay between the application of an input and the observation of its effect at the output of the system. The knowledge of the system's relative degree is a common prerequisite for the design of adaptive controllers, such as in sliding mode control~\cite{Levant.2003} or funnel control~\cite{Berger.2018}. Moreover, a certain relative degree might be required for the application of specific control strategies, see~\cite{Ling.2012}. Accordingly, determining the relative degree of an unknown system is of some interest.

\subsection{Relative degree using persistently exciting input data}

We first discuss the case of single-input single-output (SISO) systems $\mathcal{B} \in \mathcal{L}^2$, i.e., $m(\mathcal{B}) = p(\mathcal{B}) = 1$. Recall that the relative degree $r\in\N_0$ of a SISO system defined by~\eqref{isosystem} is either $r = 0$, if $\ve{D} \neq 0$, or such that 
$ \ve{CA}^k\ve{B} = 0$, for all $k < r-1$, and $\ve{CA}^{r-1}\ve{B} \neq 0$, if $\ve{D} = 0$. If $\ve{CA}^{k-1}\ve{B} = 0 = \ve{D}$ for all $k \in \N$, the relative degree of the system $r=\infty$. Therefore, the relative degree of a state space representation is determined by its impulse response, which is the same for behaviorally equivalent representations of a system, as shown in Lemma~\ref{behaviorallyEquivalent}. Hence, we can define the relative degree $r(\Bc)$ of a behavior $\mathcal{B} \in \mathcal{L}^2$ as the relative degree of an arbitrary state space representation of the behavior.

\begin{lemma} \label{rLeqL}
Consider $\mathcal{B} \in \mathcal{L}^2$ with $m(\mathcal{B}) = p(\mathcal{B}) = 1$. The relative degree $r(\Bc)\in\N_0$ fulfills either $r(\Bc) \leq l(\mathcal{B}) $ or $r(\Bc) = \infty$. If $\mathcal{B}$ is controllable and $n(\mathcal{B})>0$, or if $n(\mathcal{B}) = 0$ and there exists $(u,y) \in \mathcal{B}$ with $y \neq 0$, then the relative degree is finite.
\end{lemma}
\begin{proof}
Let $\ve{u} = (0_{l(\mathcal{B})}, \ 1, \ 0_{l(\mathcal{B})})^\top \in \R^{2l(\mathcal{B})+ 1}$ and $ \ve{y} = \ve{0} \in \R^{2l(\mathcal{B}) + 1}$. Assuming $r(\Bc) > l(\mathcal{B})$, we seek to show $r(\Bc) = \infty$. Observe that $(\ve{u},\ve{y}) \in \mathcal{B}_{[0,2l(\mathcal{B})]}$, since the first $l(\mathcal{B})+1$ entries of the impulse response are 0. Choose any minimal representation $(\ve{A},\ve{B},\ve{C},\ve{D})$ of $\mathcal{B}$. By Lemma \ref{InitialCondition}, the first window of length $l(\mathcal{B})$ will initialize the state to be zero, that is $\ve{x}(l(\mathcal{B})) = 0$. Using an identity from the proof of~\cite[Lem.~1]{IvanMarkovsky.2008}, we notice that 
\begin{align*}
\ve{0} &= \ve{y}_{l(\mathcal{B})+1, \hdots, 2l(\mathcal{B})} = \mathcal{O}_{l(\mathcal{B})}(\ve{C},\ve{A}) \ve{x}(l(\mathcal{B})+1) + \ve{\mathcal{T}}_{l(\mathcal{B})}(\ve{H}) \ve{u}_{l(\mathcal{B})+1, \hdots, 2l(\mathcal{B})} \\
&= \mathcal{O}_{l(\mathcal{B})}(\ve{C},\ve{A}) \ve{x}(l(\mathcal{B})+1),
\end{align*}
where $\ve{H}$ is the impulse response of $(\ve{A},\ve{B},\ve{C},\ve{D})$ and
\[
    \ve{\mathcal{T}}_{l(\mathcal{B})}(\ve{H}) = \begin{bmatrix} \ve{H}(0) & 0 & \cdots & 0\\ \ve{H}(1) & \ve{H}(0) & \ddots &  0\\ \vdots & \ddots & \ddots & 0 \\ \ve{H}(l(\mathcal{B})-1) & \cdots & \ve{H}(1) & \ve{H}(0)\end{bmatrix},
\]
and due to left-invertibility of $\mathcal{O}_{l(\mathcal{B})}(\ve{C},\ve{A})$ this shows that $\ve{x}(l(\mathcal{B})+1) = 0$. Hence, $(\ve{u},0,\hdots; \ve{y},0,\hdots) \in \mathcal{B}$. The impulse response of the system is therefore the constant zero sequence $0 \in \R^{\N_0}$, thus $r(\Bc) = \infty$.

Now assume that $\mathcal{B}$ is controllable. In case of $n(\mathcal{B}) > 0$, let $(\ve{A},\ve{B},\ve{C},\ve{D})$ be a minimal representation of $\mathcal{B}$. We can conclude that $\ve{B} \neq 0$, and since this representation is observable by Lemma~\ref{minimalObservable}, $\mathcal{O}_{n(\Bc)}(\ve{C},\ve{A})$ has full column rank. Hence $\mathcal{O}_{n(\Bc)}(\ve{C},\ve{A}) \ve{B} \neq 0$, thus there exists $k \in \N$ such that $\ve{CA}^{k-1}\ve{B} \neq 0$, i.e., the relative degree is finite. 

In case of $n(\mathcal{B}) = 0$, there is a representation $\ve{y}(t) = \ve{Du}(t)$, $t\in\N_0$, of the system. Hence, it follows from the assumption that $\ve{D} \neq 0$, and therefore $r(\Bc) = 0 = l(\mathcal{B})$.
\end{proof}

We have seen in Theorem~\ref{Fundamentallemma} that observing a system under a persistently exciting input sequence  allows us to determine the entire system behavior. In the following theorem, we will use this fact to identify the relative degree directly from the observed data.

\begin{theorem} \label{SISORelativgrad}
Let $\mathcal{B} \in \mathcal{L}^2$ be a controllable SISO system, and $(\ve{u_d},\ve{y_d}) \in \mathcal{B}_{[0,T]}$ with $\ve{u_d}$ persistently exciting of order $L+n(\mathcal{B})$, for some $L \geq l(\mathcal{B}) + n(\mathcal{B}) + 1$. Then the relative degree of the system is $r(\Bc) = j^* - 1$, where $j^*$ is the minimum of all integers $1 \leq j \leq L-l(\mathcal{B})$ that satisfy 
$$ \mathcal{H}_L(\ve{y_d})_{l(\mathcal{B})+j}^{\top} \notin \im \begin{bmatrix} \mathcal{H}_L(\ve{u_d})_{1,\ldots,l(\Bc)}\\ \mathcal{H}_L(\ve{y_d})_{1,\ldots,l(\Bc)}\end{bmatrix}^\top.
$$
If there is no $1 \leq j \leq n(\mathcal{B}) + 1$ that fulfills this condition, then $r(\Bc) = \infty$. 
\end{theorem}

\begin{remark}
The assumption of controllability of the system ensures that the relative degree of the system is finite by Lemma~\ref{rLeqL}, except for the case that all $(\ve{u},\ve{y}) \in \mathcal{B}$ satisfy $\ve{y} = \ve{0}$. In this case, the relative degree is infinite, and $\mathcal{H}_L(\ve{y_d})^{\top}_{l(\mathcal{B}) + j} = \ve{0}$ for all $1 \leq j \leq L - l(\mathcal{B})$, hence the condition is not fulfilled for any $j$.
\end{remark}

\begin{proof}[Proof of Theorem \ref{SISORelativgrad}]
We determine the index of the first non-trivial element of the impulse response of $\Bc$. We use the characterization of the first $L - l(\mathcal{B})$ elements of the impulse response given in \cite[Alg.~4]{IvanMarkovsky.2008} with $T_{ini} = l(\mathcal{B})$. All that is left to show is that the assumption is sufficient to ensure that there exists $\ve{g} \in \R^{T-L+2}$ such that 
$$\mathcal{H}_L(\ve{u_d}) \ve{g} = (0_{l(\mathcal{B})} \ 1 \ 0_{L-1-l(\mathcal{B})})^{\top}, \ \ve{g} \in \ker(\mathcal{H}_L(\ve{y_d})_{1,\hdots,l(\mathcal{B}) + j^* -1} )\text{, and }\mathcal{H}_L(\ve{y_d})_{l(\mathcal{B}) + j^*} \ve{g} \neq \ve{0}.$$
First, note that $(\ve{u},\ve{y}) = (\ve{0},\ve{0}) \in \mathcal{B}_{[0,l(\mathcal{B}) - 1]}$ is a valid initializing sequence, since $\mathcal{B}_{[0,l(\mathcal{B}) - 1]}$ is a vector space. This sequence will fix the initial state $\ve{x}(l(\mathcal{B})) = \ve{0}$ of any minimal representation of the system by Lemma \ref{InitialCondition}. Since the input sequence can be chosen freely for any state of the system, and the entire behavior is characterized by the image of the Hankel matrix by Theorem~\ref{Fundamentallemma}, there must be a solution $\ve{g} \in \R^{T-L+2}$ such that both conditions $\mathcal{H}_L(\ve{u_d}) \ve{g} = (0_{l(\mathcal{B})} \ 1 \ 0_{L-1-l(\mathcal{B})})^{\top} \text{ and } \ve{g} \in \ker(\mathcal{H}_L(\ve{y_d})_{1,\hdots,l(\mathcal{B})} )$ hold. The obtained output vector $\ve{y} \coloneqq \mathcal{H}_L(\ve{y_d})_{l(\mathcal{B}) + 1, \hdots, L} \ve{g}$ is determined uniquely by these conditions, as we discussed in Corollary \ref{uniqueOutputInitialized}. Due to the initialization $\ve{x}(l(\mathcal{B})) = \ve{0}$, the output vector $\ve{y}$ contains the first $L-l(\mathcal{B})$ elements of the impulse response.

If no $1 \leq j \leq n(\mathcal{B}) + 1$ fulfills the required condition, then $\ve{y} = \ve{0} \in \R^{n(\mathcal{B}) +1}$, hence $r(\mathcal{B}) \geq n(\mathcal{B}) + 1 > l(\mathcal{B})$. It then follows from Lemma~\ref{rLeqL} that $r(\mathcal{B}) = \infty$.

In case $j^*$ exist, it follows from its minimality that $y_j = 0$ for all $1 \leq j < j^*$. This shows a lower bound on the relative degree of the system, $r(\mathcal{B}) \geq j^*-1$, since the first $j^*-2$ elements of the impulse response are zero. It remains to show that $r(\mathcal{B}) \leq j^*-1$.  

By assumption, we have that there exists $\ve{g} \in \ker(\mathcal{H}_L(\ve{u_d})_{1,\hdots,l(\mathcal{B})}) \  \cap \  \ker(\mathcal{H}_L(\ve{y_d})_{1,\hdots,l(\mathcal{B})}) $ such that $\mathcal{H}_L(\ve{y_d})_{l(\mathcal{B}) + j^*} \ve{g} \neq \ve{0}$. Therefore, the input sequence $\ve{u} \coloneqq \mathcal{H}_L(\ve{u_d})_{l(\mathcal{B} ) + 1 , \hdots , L} \ve{g}$ yields an output $\ve{y}\coloneqq\mathcal{H}_{L}(\ve{y_d})_{l(\mathcal{B}) +1, \hdots, L}\ve{g}$ with a non-zero $j^*$th component, $y_{j^*} \neq 0$. Assume that $r(\Bc) \geq j^*$, then the $j^*$th component is not affected by the input $\ve{u}$. The non-zero output has to be caused by the zero input response of the state $\ve{x}(l(\mathcal{B}))$, which is initialized to be zero, a contradiction. This shows that $r(\mathcal{B}) = j^*-1$.
\end{proof}

Next, we extend the previous result to multiple-input, multiple-output (MIMO) systems. First, we recall the concept of vector relative degree.

\begin{definition} \label{RelativeDegreeMIMO}
Let $\mathcal{B} \in \mathcal{L}^{w}$  and $(\ve{A}, \ve{B}, \ve{C}, \ve{D})$ be a state space representation of $\Bc$. Then $\Bc$  has vector relative degree $\ve{r}(\Bc) = (r_1, \hdots, r_{p(\mathcal{B})}) \in \N_0^{p(\mathcal{B})}$ if, and only if,
\begin{itemize}
\begin{item}[(i)] 
$\ve{C_jA}^k\ve{B} = \ve{0} $ for all $1 \leq j \leq p(\mathcal{B}),\ 0 \leq k \leq r_j - 2$, and $\ve{D_j} = 0$ if $r_j > 0$, and
\end{item}
\begin{item}[(ii)]
(\textit{decoupling condition}) the decoupling matrix $\ve{G}\in\R^{p(\mathcal{B})\times m(\Bc)}$ consisting of the rows 
$$\ve{G_j} = \begin{cases} \ve{D_j}, & \text{ if } r_j = 0 \\ \ve{C_j A}^{r_j -1} \ve{B}, & \text{ if } r_j > 0 \end{cases} $$
for $1 \leq j \leq p(\mathcal{B})$ has rank $p(\mathcal{B})$.
\end{item}
\end{itemize}
The system $\Bc$ has strict relative degree $r(\Bc) \in \N_0$, if it has the vector relative degree $(r(\Bc), \hdots, r(\Bc))$.
\end{definition}

\begin{remark}
Similar to the case of SISO systems, the vector relative degree of the state space representation of a MIMO system is determined by its impulse response and hence is independent of the choice of representation, which allows us to define the vector relative degree of a behavior as that of an arbitrary state space representation of it.
\end{remark}

Any MIMO system $\Bc$ with $m$ inputs and $p$ outputs induces a set of $mp$ SISO systems $\mathcal{B}^{i,j}$ by determining how an output coordinate $y_i$, $1 \leq i \leq p$, depends on an input coordinate $u_j$, $1 \leq j \leq m$, while all other inputs $u_k, k \neq j$, are constantly zero. Output transformations may change the values $r_i$ that satisfy the first condition. This is solved by imposing the decoupling condition. If two state space representations can be derived from each other by an output transformation, and both have a vector relative degree, then these vector relative degrees are equal.

Obviously, controllability of a MIMO system $\mathcal{B}$ does not imply controllability of all induced SISO systems. Let  $(\ve{u_d},\ve{y_d}):[0,T]\to\R^w$ be a data sequence  that satisfies the requirements of Theorem~\ref{Fundamentallemma}. We employ this data sequence to characterize any $\mathcal{B}^{i,j}$ by picking suitable rows of the respective Hankel matrices.
\begin{align*}
\mathcal{B}_{[0,L-1]}^{i,j} &= \left\{ (u,y) : \exists\, (\ve{v},\ve{w}) \in \mathcal{B}_{[0,L-1]} \text{ s.t. } u = v_j, y = w_i \text{ and } v_k = 0, \forall\, k \neq j \right \} \\
&= \left\{ (u,y) : \exists\, \ve{g} \in \bigcap_{k \neq j} \ker \mathcal{H}_L(\ve{u_d})^k  \text{ s.t. } \ve{u} = \mathcal{H}_L(\ve{u_d})^{j} \ve{g}, \ve{y} = \mathcal{H}_L(\ve{y_d})^i \ve{g} \right \} \\
&= \left\{ \begin{bmatrix} \mathcal{H}_L(\ve{u_d})^{j} \\ \mathcal{H}_L(\ve{y_d})^i \end{bmatrix} \ve{g} : \ve{g} \in \bigcap_{k \neq j} \ker \mathcal{H}_L(\ve{u_d})^k \right \},
\end{align*}
where, for $\ve{w} \in (\R^v)^{T+1}$, $\mathcal{H}_L(\ve{w})^k$  denotes the submatrix of $\mathcal{H}_L(\ve{w})$ consisting of the rows with indices $k + av$ for $0 \leq a \leq L-1$. Therefore, we can characterize induced SISO systems that are not controllable using Theorem~\ref{Fundamentallemma} for the entire system. In this case, we can determine the relative degree of the induced SISO systems in the way discussed in Theorem~\ref{SISORelativgrad}, even if they are not controllable. 

\begin{remark} \label{MIMOfromSISO}
Note that if a MIMO system $\mathcal{B} \in \mathcal{L}^w$ has a vector relative degree $\ve{r}(\mathcal{B}) = (r_1, \hdots, r_{p(\mathcal{B})}) \in \N_0^{p(\mathcal{B})}$, we can determine the coordinate $r_i$ belonging to the output coordinate $y_i$, $1 \leq i \leq p(\mathcal{B})$, by taking the minimum $r_i = \min_{1 \leq j \leq m(\mathcal{B})} r_{ij} $ over all relative degrees $r_{ij} \in \N_0 \cup \{ \infty \}$ of induced SISO systems $\mathcal{B}^{i,j}$. This follows directly from the definition of the vector relative degree: Condition~(i) of Definition~\ref{RelativeDegreeMIMO} shows that $r_{ij} \geq r_i$ for all $1 \leq j \leq m(\mathcal{B})$, while the decoupling condition shows that there exists at least one index $j$ such that $(\ve{C_i A}^{r_i - 1} \ve{B})_j \neq 0$, yielding $r_j \leq r_{ij}$. It is possible for some $r_{ij}$ to be infinite, if there are output coordinates $y_i$ that are not affected by certain input coordinates $u_j$. If the system is controllable, every output coordinate will depend on at least one input coordinate, hence their minimum $r_i$ will always be finite. Therefore, it is sufficient to determine those relative degrees $r_{ij}$ for which the induced SISO system $\mathcal{B}^{i,j}$ has finite relative degree. 
\end{remark}

This interpretation of the vector relative degree and Theorem~\ref{SISORelativgrad} are used in the following theorem to determine the candidate for the vector relative degree as well as to prove its existence. We again assume controllability, hence all $r_i$ are finite. 

\begin{theorem} \label{MIMOSquareRelativgrad} 
Let $\mathcal{B} \in \mathcal{L}^{w}$ be controllable, and $(\ve{u_d},\ve{y_d}) \in \mathcal{B}_{[0,T]}$ with $\ve{u_d}$ persistently exciting of order $L + n(\mathcal{B})$, for some $L \geq l(\mathcal{B}) + n(\mathcal{B}) + 1$. Let $r_{ij}$ be the relative degree of the induced SISO system $\mathcal{B}^{i,j}$ determined by Theorem \ref{SISORelativgrad}, whenever it is controllable, and $r_{ij} = \infty$ else. Let $\ve{r} = (r_1, \hdots, r_{p(\mathcal{B})})$ where $r_i := \min_{1 \leq j \leq m(\mathcal{B})} r_{ij}$. \\
Let $\ve{e_j} \in \R^{m(\mathcal{B})}$ be the $j$th canonical unit vector for $1 \leq j \leq m(\mathcal{B})$, and determine, according to Corollary~\ref{uniqueOutputInitialized}, the unique output vectors $\ve{z^{(j)}}$ by solving 
$$ \begin{bmatrix} \mathcal{H}_L(\ve{u_d}) \\ \mathcal{H}_L(\ve{y_d}) \end{bmatrix} \ve{g} = \begin{pmatrix} \ve{0}_{m(\mathcal{B})l(\mathcal{B})} \\ \ve{e_j} \\ \ve{0}_{m(\mathcal{B})(L-l(\mathcal{B})-1)}  \\ \ve{0}_{m(\mathcal{B})l(\mathcal{B})} \\ \ve{z^{(j)}}  \end{pmatrix}, \text{ and split } \ve{z^{(j)}} = \begin{pmatrix} \ve{z_0^{(j)}} \\ \vdots \\ \ve{z_{L-l(\mathcal{B})-1}^{(j)}} \end{pmatrix},$$   
 where all $\ve{z_i^{(j)}} \in \R^{m(\mathcal{B})}$. Define a matrix by its entries $a_{ij} = \left(\ve{z_{r_i}^{(j)}}\right)_i \in \R$ for all $1 \leq i \leq p(\mathcal{B})$ and all $1 \leq j \leq m(\mathcal{B})$,  where $( \cdot )_i$ denotes the $i$th entry of the respective vector.\\
Then all $r_i$ are finite and $\ve{r} = \ve{r}(\Bc)$ if, and only if, the matrix $(a_{ij})_{1 \leq i \leq p(\mathcal{B}), 1 \leq j \leq m(\mathcal{B})}$ has full rank $p(\mathcal{B})$.
\end{theorem}

\begin{proof}
All $r_i$ that satisfy condition~(i) of Definition~\ref{RelativeDegreeMIMO} are finite, as discussed in Remark \ref{MIMOfromSISO}, due to controllability of the system. If the system does have a vector relative degree, it can be determined by all relative degrees $r_{ij}$ of the induced SISO systems $\mathcal{B}^{i,j}$. Hence, the chosen vector $\ve{r}$ is the only possible candidate for a vector relative degree, and it is left to show that the decoupling condition~(ii) of Definition~\ref{RelativeDegreeMIMO} is equivalent to the invertibility of the constructed matrix $(a_{ij})$.

To establish this equivalence, we first assume $n(\mathcal{B}) > 0$ and choose a minimal state space representation $(\ve{A}, \ve{B}, \ve{C}, \ve{D})$ of $\Bc$ and calculate the vectors $\ve{z^{(j)}}$ in terms of these matrices. By Lemma \ref{InitialCondition}, we have that the state of the system after the initializing input and output sequences is $\ve{x}(l(\mathcal{B}) + 1) = 0$. Hence, we can calculate the response to the input $\ve{e_j}$ for any $1 \leq j \leq m(\mathcal{B})$ as
$$\begin{cases}\ve{z_0^{(j)}} = \ve{De_j} = \ve{D_j}, \\ \ve{z_k^{(j)}} =\sum_{i = 1}^k \ve{CA}^{k-i}\ve{Bu}(i) = \ve{CA}^{k-1}\ve{B e_j}\quad \text{ for all } 1 \leq k \leq L - l(\mathcal{B}). \end{cases}$$ 
Due to the choice of the $r_i$, we have that the row $(\ve{CA}^{k-1}\ve{B})_i = \ve{C_i A}^{k-1} \ve{B} = \ve{0}$ for all $k < r_i$, as well as $\ve{D}_i = 0$ if $r_i > 0$, and
$$a_{ij} = \left(\ve{z_{r_i}^{(j)}}\right)_i = \begin{cases} (\ve{De_j})_i = \ve{D}_{ij} &\text{ if } r_i = 0 \\ (\ve{CA}^{r_i - 1} \ve{B e_j})_i = \ve{C_i A}^{r_i - 1} \ve{B e_j} &\text{ if } r_i > 0 \end{cases}$$ for all $1 \leq i \leq p(\mathcal{B})$ and $1 \leq j \leq m(\mathcal{B})$. This yields the equality for the row
$$(a_{i1}, \hdots, a_{i,m(\mathcal{B})}) = \begin{cases} \ve{D_i} &\text{ if } r_i = 0 \\ \ve{C_iA^{r_i - 1}B} &\text{ if } r_i > 0 \end{cases}$$ for all $1 \leq i \leq p(\mathcal{B})$, which are the rows $\ve{G_i}$ of the decoupling matrix. This shows that the matrix $(a_{ij})$ has rank $p(\mathcal{B})$ if, and only if, the decoupling matrix $\ve{G}$ has full rank $p(\mathcal{B})$, since the matrices are equal.

In the case $n(\mathcal{B}) = 0$, choose a state space representation $\ve{y}(t) = \ve{Du}(t)$, $t\in\N_0$. The system is assumed to be controllable, therefore all $r_i = 0$, and we can proceed as before, obtaining $\ve{z_0^{(j)}} = \ve{D_j}$ and $(a_{ij}) = \ve{D} = \ve{G}$. This completes the proof.
\end{proof}

\begin{remark} We like to emphasize that the proof of Theorem~\ref{MIMOSquareRelativgrad} is constructive and, in particular, the decoupling matrix $\ve{G} =(a_{ij})_{1 \leq i \leq p(\mathcal{B}), 1 \leq j \leq m(\mathcal{B})}$ (sometimes called high-gain matrix in the literature on adaptive control) can be explicitly constructed from the persistently exciting data sequence. This is important when additional conditions for the application of certain control algorithms, such as sign definiteness of the decoupling matrix required for funnel control~\cite{Berger.2021}, need to be checked.
\end{remark}

\subsection{Relative degree and data informativity}

In this section, we extend the previous considerations to the case where the behavior $\Bc$ might not be controllable and the input data sequence might not be persistently exciting, thus going beyond the scope of the Fundamental Lemma. The question is, what are the weakest conditions on the data, such that the information about the (vector) relative degree can be extracted from it. Hence, in the language of Definition~\ref{Def:informativity}, we seek to determine under which conditions a given data set is informative for the (vector) relative degree. We denote the set of all systems with (vector) relative degree $\ve{r}\in\N_0^p$ from the model class $\Sigma= \mathcal{L}^{m+p}$ as $\Sigma_{\ve{r}}$, where $m,p \in \N$.

First, we again consider SISO systems. The following theorem provides sufficient conditions, which allow to identify the relative degree, that are weaker than the conditions of Theorem~\ref{SISORelativgrad}.

\begin{theorem} \label{SISORelativgradScharf}
Let $\mathcal{B} \in \mathcal{L}^2$ be a SISO system, and $(\ve{u_d},\ve{y_d}) \in \mathcal{B}_{[0,T]}$ be a data sequence. Define
\[
    M_L := \im \begin{bmatrix} \mathcal{H}_L(\ve{u_d})_{1,\ldots,l(\Bc)}\\ \mathcal{H}_L(\ve{y_d})_{1,\ldots,l(\Bc)}\end{bmatrix}^\top,\quad L\ge 1.
\]
If there exists $L \geq l(\mathcal{B}) + n(\mathcal{B}) + 1$ such that the data sequence fulfills 
\begin{align} \label{ConditionRelativeDegree}
\mathcal{H}_L(\ve{y_d})_i^{\top} \notin M_L \text{ for some } 1 \leq i \leq T,
\end{align}
then the relative degree of the system is
$r(\Bc) = k_y - k_u$, where 
\begin{align*}
k_u &= \min \left\{ 1 \leq i \leq L : \mathcal{H}_L (\ve{u_d})_i^{\top} \notin M_L \right\}, \\
 k_y &= \min \left\{ 1 \leq i \leq L : \mathcal{H}_L (\ve{y_d})_i^{\top} \notin M_L \right\}.
\end{align*}
\end{theorem}
\begin{proof} 
We seek to determine how many time steps are needed until the output is affected by the input. In case of $n(\mathcal{B}) > 0$, we  fix the initial state $\ve{x}(l(\mathcal{B})) = \ve{0}$ using Lemma \ref{InitialCondition} by assuming the initializing behavior $(\ve{u}_{\text{past}}, \ve{y}_{\text{past}}) = \ve{0}_{2l(\mathcal{B})}$ for some fixed minimal representation of the system. In case of $n(\mathcal{B}) = 0$, we do not need to initialize as there is no state in the minimal representation. 

By assumption (\ref{ConditionRelativeDegree}), there exists $\ve{g} \in M_L^{\perp}$ such that $\ve{g}^\top \mathcal{H}_L(\ve{y_d})_{k_y} \neq 0$. This yields
$$ \begin{pmatrix} \ve{0}_{l(\mathcal{B})} \\ \ve{u} \\ \ve{0}_{l(\mathcal{B})} \\ \ve{y} \end{pmatrix} = \begin{bmatrix} \mathcal{H}_L(\ve{u_d}) \\ \mathcal{H}_L(\ve{y_d}) \end{bmatrix} \ve{g} $$
for some $\ve{u}, \ve{y} \in \R^{L-l(\mathcal{B})}$ with $y_{k_y} \neq 0$. First, we show that $\ve{u} \neq 0$. If $n(\mathcal{B}) = 0$, there is a matrix $\ve{D} \in \R^{1 \times 1}$ such that $y_{k_y} = \ve{D} u_{k_y}$, hence $\ve{u} \neq 0$. Furthermore, $\ve{D}\neq 0$ and hence $r(\Bc)=0$ in this case. If $n(\mathcal{B}) > 0$, the system equations show that the state would be trivial at every time step, after initialization, $\ve{x}(l(\mathcal{B}) + i) = \ve{0}$ for all $1 \leq i \leq L -  l(\mathcal{B})$. Since the output $y_{k_y}$ linearly depends on the state $\ve{x}(l(\mathcal{B}) + k_y - 2)$ in the previous time step and the input $u_{k_y}$ applied in the current time step, one of these two needs to be non-zero, hence $\ve{u} \neq 0$. This shows that there is an index $1 \leq i \leq L$ such that $\mathcal{H}_L (\ve{u_d})_i^{\top} \notin M_L$, thus we can conclude that $k_u$ is well-defined, and $k_y \geq k_u$.

By choice of $k_u$, there exists $\ve{\tilde{g}} \in M_L^{\perp} $ such that 
$$ \mathcal{H}_L(\ve{u_d}) \ve{\tilde{g}} = (\ve{0}_{l(\mathcal{B}) + k_u - 1} \ \tilde{u}_{k_u} \  \hdots)^{\top}, $$
for some $\tilde{u}_{k_u} \neq 0$. In the case of $n(\mathcal{B}) = 0$, this yields $ \mathcal{H}_L(\ve{y_d}) \tilde{g} = (\ve{0}_{l(\mathcal{B}) + k_u - 1} \ \tilde{y}_{k_u} \  \hdots)^{\top}$ for some $\tilde{y}_{k_u} = \ve{D} \tilde{u}_{k_u} \neq 0$, implying $k_u \geq k_y$, thus $k_u = k_y$ and $r(\Bc) = 0 = k_y - k_u$ as claimed. Now assume that $n(\mathcal{B}) > 0$ and, seeking a contradiction, $u_{k_u} = 0$. It is obvious that the state satisfies $\ve{x}(l(\mathcal{B}) + i) = 0$ for every $0 \leq i \leq k_u - 1$. By shift invariance of the behavior, this $\ve{\tilde{g}}$ yields an output vector $\ve{\tilde{y}}$ where $\tilde{y}_i \neq 0$ for some $i < k_y$, which contradicts the minimality in the choice of $k_y$. Hence, $\mathcal{H}_L(\ve{u_d})_{k_u}^\top \ve{g}  = u_{k_u} \neq 0$.

Since in the case of $n(\mathcal{B}) > 0$, we have that $\ve{x}(l(\mathcal{B}) + k_u) = \ve{0}$, we can express the output vector $\ve{y}$ by using the impulse response $\ve{H}$ of the system,
$$ y_i = \begin{cases} 0, &\text{ if } i < k_u \\ \sum_{j = 0}^{i-k_u} \ve{H}(i - k_u - j) u_{k_u + j}, &\text{ if } i \geq k_u.  \end{cases}$$
Since $\ve{H}(j) = 0$ for all $j < r(\Bc)$, we can conclude that $y_i = 0$ for all $i < k_u + r(\Bc)$, and $y_{k_u + r(\Bc)} = \ve{H}(r(\Bc))u_{k_u} \neq 0$. Therefore, $k_y = k_u + r(\Bc)$, which implies $r(\Bc) = k_y-k_u$, as claimed.
\end{proof}

\begin{remark}
With the help of Theorem~\ref{SISORelativgradScharf} it is possible to determine the relative degree of a SISO system $\mathcal{B} \in \mathcal{L}^2$ from data under strictly weaker conditions than in Theorem~\ref{SISORelativgrad}, since it does not require controllability of the system, and it does not require the input sequence to be persistently exciting. Instead, its assumptions depend on the output sequence, and it will therefore not be possible to check the conditions prior to collecting the output data. Additionally, note that Theorem~\ref{SISORelativgrad} is able to detect $r = \infty$, while the assumptions of Theorem~\ref{SISORelativgradScharf} imply that the relative degree of the system is finite. 
\end{remark}

In order to find conditions that are equivalent to the informativity for the relative degree, we express this property in terms of the MPUM.

\begin{lemma} \label{InformativityMPUM}
Let $\mathcal{B} \in \mathcal{L}^2$ be the true SISO system with known lag $l(\mathcal{B})$, let $\mathcal{D} = (\ve{u_d},\ve{y_d}) \in \mathcal{B}_{[0,T-1]}$ be a sampled data sequence and $r \in \N_0$. Then $\mathcal{D}$ is informative for the relative degree~$r$ if, and only if, there exists $a \in \R \setminus \{ 0 \}$ such that $\ve{v}(a,r) = (\ve{u}, \ve{y}(a,r)) \in \mathcal{B}_{MPUM,[0,l(\mathcal{B})+r]}$, where $\ve{u} = (\ve{0}_{l(\mathcal{B})}, \ 1,  \ \hdots ) \in \R^{l(\mathcal{B})+r+1}$ and $\ve{y}(a,r) = (\ve{0}_{l(\mathcal{B})+r }, \ a) \in \R^{l(\mathcal{B})+r+1}$ are sequences whose omitted entries can be chosen arbitrarily.
\end{lemma}
\begin{proof}
First, assume that $\ve{v}(a,r) \in \mathcal{B}_{MPUM,[0,l(\mathcal{B})+r]}$. Then $\ve{v}(a,r)$ is contained in any behavior that explains the data, hence the first $r+1$ entries of the impulse response of every such behavior are $(\ve{0}_{r}, a)$, which yields a relative degree of~$r$, since $a \neq 0$. Therefore, $\Sigma_{\mathcal{D}} \subseteq \Sigma_{r}$.

Now let $\mathcal{D}$ be informative for the relative degree $r$, and assume, seeking a contradiction, that there are two behaviors $\mathcal{B}_1, \mathcal{B}_2 \in \Sigma_{\mathcal{D}} \subseteq \Sigma_r$ with one-dimensional input and output, but do not have a trajectory $\ve{v}(a,r)$ with $a\neq 0$ in common. Observe that the behaviors each contain a sequence $\ve{v}(a_i,r_i) \in \mathcal{B}_{i, [0, l(\mathcal{B}) + r_i]}$ with $a_i \neq 0$ and $r_i\in\N_0$ for $i=1,2$. Since $r(\Bc_1)=r(\Bc_2)=r$, it follows that $r_1=r_2=r$ and $a_1 \neq a_2$ such that $\ve{v}(a_i,r) \in \mathcal{B}_{i,[0,l(\mathcal{B}) + r]}$, $i=1,2$. Using these coefficients, define the behavior 
$$\scalebox{0.95}{$\mathcal{B}_3 = \left\{ (\ve{u},\ve{y}) : (\ve{u}, \ve{y_1}) \in \mathcal{B}_1, (\ve{u}, \ve{y_2}) \in \mathcal{B}_2 \text{ and } \ve{y} = \frac{- a_2}{a_1 - a_2} \ve{y_1} + \frac{a_1}{a_1 - a_2} \ve{y_2} \right\},$}$$ which may be interpreted as a convex combination of the previously fixed behaviors. We can easily confirm that $\mathcal{B}_3 \in \mathcal{L}^{2}$ by fixing minimal representations $(\ve{A_1}, \ve{B_1}, \ve{C_1}, \ve{D_1})$ of $\mathcal{B}_1$ and $(\ve{A_2}, \ve{B_2}, \ve{C_2}, \ve{D_2})$ of $\mathcal{B}_2$. Let $\ve{x_1}, \ve{x_2}$ denote the state vectors corresponding to the state space representations. Then we can construct a state space representation of $\mathcal{B}_3$ with $\ve{x} = (\ve{x_1}^{\top}, \ve{x_2}^{\top})^{\top}$ as state vector, and the characterizing matrices $$\scalebox{0.95}{$\ve{A} = \begin{bmatrix} \ve{A_1} & \ve{0} \\ \ve{0} & \ve{A_2} \end{bmatrix},\ \ve{B} = \begin{bmatrix} \ve{B_1}\\ \ve{B_2} \end{bmatrix},\ \ve{C} = \left(\frac{-a_2}{a_1 - a_2}\ve{C_1}, \frac{a_1}{a_1 - a_2}\ve{C_2}\right),\ \ve{D} = \frac{-a_2}{a_1 - a_2}\ve{D_1} + \frac{a_1}{a_1 - a_2}\ve{D_2}.$}$$ This yields 
\begin{equation} \label{ConvexCombinationOfSystems}
\left.
\begin{aligned}
\ve{x}(t+1) &= \ve{Ax}(t) + \ve{Bu}(t) = \begin{pmatrix} \ve{A_1x_1}(t) + \ve{B_1u}(t) \\ \ve{A_2x_2}(t) + \ve{B_2u}(t) \end{pmatrix} = \begin{pmatrix} \ve{x_1}(t+1) \\ \ve{x_2}(t+1) \end{pmatrix}\\
\ve{y}(t) &= \frac{-a_2}{a_1 - a_2}\ve{C_1x_1}(t) +\frac{a_1}{a_1 - a_2}\ve{C_2x_2}(t) + \frac{-a_2}{a_1 - a_2}\ve{D_1u}(t) + \frac{a_1}{a_1 - a_2}\ve{D_2u}(t) \\
&= \frac{-a_2}{a_1 - a_2}\ve{y_1}(t) +\frac{a_1}{a_1 - a_2} \ve{y_2}(t).
\end{aligned}
\right\}
\end{equation}
Therefore, we can conclude $\mathcal{B}_3 \in \mathcal{L}^{2}$.  

Furthermore, $\mathcal{B}_3 \in \Sigma_{\mathcal{D}}$ since $(\ve{u_d}, \ve{y_d}) \in \mathcal{B}_{1,[0, T-1]} \cap \mathcal{B}_{2, [0, T-1]}$ implies that there are trajectories $(\ve{u_1},\ve{y_1}) \in \mathcal{B}_1$ and $(\ve{u_2}, \ve{y_2}) \in \mathcal{B}_2$ with $\ve{u}_{\ve{1} \vert [0,T-1]} = \ve{u_d} = \ve{u}_{\ve{2} \vert [0, T-1]}$ and $\ve{y}_{\ve{1} \vert [0,T-1]} = \ve{y_d} = \ve{y}_{\ve{2} \vert [0, T-1]}$. Since the input coordinates are free, we can assume without loss of generality that $\ve{u_1} = \ve{u_2}$. By definition, 
$$\displaystyle{(\ve{u},\ve{y}) \coloneqq \frac{-a_2} { a_1 - a_2} (\ve{u_1}, \ve{y_1}) + \frac{a_1}{a_1 - a_2} (\ve{u_2}, \ve{y_2}) = \left(\ve{u_2}, \ve{y_2} + \frac{a_2}{a_1 - a_2} (\ve{y_2}-\ve{y_1})\right) \in \mathcal{B}_3}.$$
Therefore, $\ve{u}_{\vert [0,T-1]} = \ve{u_d}$ and $\ve{y}_{\vert [0, T-1]} = \ve{y_d}$, thus $\mathcal{B}_3 \in \Sigma_{\mathcal{D}}$.

Next we show that $\mathcal{B}_3 \notin \Sigma_{r}$. Consider an input sequence $\ve{u}:\N_0\to\R$ with $u(i) = 0$ for all $i \neq L+1$, and  $u(L+1) = 1$ for some $L \geq \max \{l(\mathcal{B}_1), l(\mathcal{B}_2) \} $. 
Further, let $(\ve{u}, \ve{y_1}) \in \mathcal{B}_1 $ and $(\ve{u}, \ve{y_2}) \in \mathcal{B}_2$ be trajectories such that $\ve{y_1}(t) = 0 = \ve{y_2}(t) $ for all $t \leq L$. Then in case of $n(\mathcal{B}) > 0$, we obtain $\ve{x_1}(L+1) = \ve{0} = \ve{x_2}(L+1)$ by Lemma \ref{InitialCondition}, and therefore $\ve{y_1}(t) = \ve{0} = \ve{y_2}(t)$ for all $t \leq L + r$ by definition of the relative degree. In case of $n(\mathcal{B}) = 0$, we have $r = 0$ and hence $\ve{y_1}(t) = 0 = \ve{y_2}(t)$ for all $t \leq L + r = L$ obviously holds. Moreover, we can observe that $ \ve{y_1}(L+r+1) = a_1$ and $\ve{y_2}(L+r+1) = a_2$ by choice of the behaviors $\Bc_1,\Bc_2$. Now define 
\[
\ve{y_3} = \frac{-a_2}{a_1 - a_2} \ve{y_1} + \frac{a_1} { a_1 - a_2} \ve{y_2},
\]
for which we find that $(\ve{u},\ve{y_3})\in \mathcal{B}_3$ and $\ve{y_3}(t) = 0$ for all $t \leq L + r$. Considering the representation constructed in~\eqref{ConvexCombinationOfSystems}, we note that $\ve{x}(L+1) = 0$, and therefore the shifted output sequence $\ve{y_3}(L+1+\cdot)$ is the impulse response of the system $\Bc_3$. Since 
$$\ve{y_3}(L+1+r) = \frac{-a_2}{a_1 - a_2} \ve{y_1}(L+r+1) + \frac{a_1}{a_1 - a_2} \ve{y_2}(L+r+1) = \frac{-a_1a_2}{a_1 - a_2} + \frac{a_1a_2}{a_1 - a_2} = 0,$$ 
we find that $r(\mathcal{B}_3)>r$. This shows that $\mathcal{B}_3 \in \Sigma_{\mathcal{D}} \setminus \Sigma_{r}$, which contradicts the assumption that $\mathcal{D}$ is informative for the relative degree $r$. We can conclude that there exists $a\neq 0$ such that all behaviors $\mathcal{B} \in \Sigma_{\mathcal{D}}$ contain a trajectory $\ve{v}(a,r)$. Then the assertion of the lemma follows from the definition of $\Bc_{MPUM}$ as the intersection over all $\Bc\in\Sigma_\Dc$.
\end{proof}

Next, we provide a necessary and sufficient condition on the data to be informative for the relative degree. Even more so, this result is constructive and allows to be checked by an algorithmic procedure which is provided in Algorithm~\ref{AlgoRelativeDegree}. If the data are informative for the relative degree, then this algorithm additionally computes the value of it. 

\begin{algorithm} 
\caption{Determine Relative Degree of SISO system}
\label{AlgoRelativeDegree}
\begin{algorithmic}
\State \hspace*{\algorithmicindent} \textbf{Input} lag $l\coloneqq l (\mathcal{B})$, data $(\ve{u_d},\ve{y_d})$
\vspace{0.3cm}
\State $i \coloneqq l + 1$
\State $M_j \coloneqq \im \begin{bmatrix} \mathcal{H}_{l+1} (\ve{u_d})_{1,\hdots, j-1}\\ \mathcal{H}_{l+1} (\ve{y_d})_{1, \hdots, l} \end{bmatrix}^\top$ for all $1 \leq j \leq l+1$
\While {$i > 0$}
\If {$\mathcal{H}_{l+1} (\ve{u_d})_i^{\top} \in M_i$}
\State return $\tilde{r} = -1$
\ElsIf {$\mathcal{H}_{l+1} (\ve{y_d})_{l+1}^{\top} \notin M_i, $} 
\State return $\tilde{r} = l+1-i$
\EndIf
\State $i \coloneqq i-1$
\EndWhile
\State set $\tilde{r} = -1$
\vspace{0.3cm}
\State \hspace*{\algorithmicindent} \textbf{Output} $\tilde{r}$
\end{algorithmic}
\vspace{0.4em}
\hrule
\vspace{0.3em}
{\itshape The \textbf{return} statements indicate termination of the algorithm once a result is obtained.}
\end{algorithm}

\begin{theorem} \label{SISORelativgradAlgo} Let $\mathcal{B} \in \mathcal{L}^2$ be an unknown system with known lag $l(\mathcal{B})$, and let $\mathcal{D} = (\ve{u_d},\ve{y_d}) \in \mathcal{B}_{[0,T]}$ be a sampled data sequence. Then the following statements hold:
\begin{enumerate}
    \item[(i)] Fix $r\in\N_0$. The data $\mathcal{D}$ is informative for the relative degree~$r$ if, and only if, Algorithm~\ref{AlgoRelativeDegree} returns $\tilde r = r$. In this case we have $r(\Bc) = r$.
    \item[(ii)] The data $\mathcal{D}$ is not informative for the relative degree~$r$ for all $r \in \N_0$ if, and only if, Algorithm~\ref{AlgoRelativeDegree} returns $\tilde{r} = -1$.
\end{enumerate}
\end{theorem}
\begin{proof}
We will show the necessity and sufficiency for both claims individually.\\
\textbf{Step 1}: We show necessity in statement (i). Assume that the output of Algorithm~\ref{AlgoRelativeDegree} is $\tilde{r} \geq 0$. In this case, we have that the following conditions hold:
\begin{align} 
\mathcal{H}_{l(\mathcal{B})+1} (\ve{u_d})_i^{\top} &\notin M_i, \text{ for all } i \geq l(\mathcal{B})+1-\tilde{r}  \label{ConditionInputFree} \\
\mathcal{H}_{l(\mathcal{B})+1}(\ve{y_d})_{l(\mathcal{B})+1}^{\top} &\in M_i, \text{ for all } i \geq l(\mathcal{B})+2-\tilde{r} \label{ConditionOutputZero} \\
\mathcal{H}_{l(\mathcal{B})+1}(\ve{y_d})_{l(\mathcal{B})+1}^{\top} &\notin M_{l(\mathcal{B})+1-\tilde{r}}. \label{ConditionOutputResponse}  
\end{align}
We show that $\tilde{r} = r(\Bc)$. Using the notation from Lemma~\ref{InformativityMPUM}, it is sufficient to show that there exists $\tilde{a} \neq 0$ such that $\ve{v}(\tilde{a},\tilde{r}) \in \mathcal{B}_{MPUM, [0,l(\mathcal{B}) + \tilde{r}]}$, since this will imply that $\mathcal{D}$ is informative for the relative degree $\tilde{r}$, which in particular gives $\Bc\in \Sigma_\Dc \subseteq \Sigma_{\tilde r}$.

By the characterization of the lag in Proposition \ref{CharakterisierungenLag}, it is sufficient to ensure that all subsequences of length $l(\mathcal{B}) + 1$ of $\ve{v}(\tilde{a},\tilde{r})$, denoted by $ \ve{v_j} \coloneqq \ve{v}(\tilde{a},\tilde{r})_{j, \hdots, j+l(\mathcal{B})}$ for $0 \leq j \leq \tilde{r}$, are contained in $\mathcal{B}_{MPUM,[0,l(\mathcal{B})]} =:A$, which can be expressed explicitly by the data $(\ve{u_d}, \ve{y_d})$ using Lemma \ref{DetermineMPUM}. To this end, we show that $\ve{v_j} \in \im \begin{bmatrix} \mathcal{H}_{l(\mathcal{B}) + 1} (\ve{u_d}) \\ \mathcal{H}_{l(\mathcal{B}) + 1} (\ve{y_d}) \end{bmatrix} = A$ for $1 \leq j \leq \tilde{r}+1$ inductively. 

The induction hypothesis for some fixed $1 \leq j \leq \tilde{r} + 1$ consists of the following two assumptions:
\begin{enumerate}
    \item[(a)] $\ve{v_{k-1}} \in A$ holds for all $1 \leq k \leq j$.
    \item[(b)] $j \leq r(\Bc) + 2$.
\end{enumerate}
Obviously, the constant zero trajectory is contained in $\mathcal{B}_{MPUM},$ which implies that $\ve{0}_{2(l(\mathcal{B})+1)} =: \ve{v_0} \in A$. This shows (a) for $j = 1$. Since $r(\Bc)\in \mathbb{N}_0$, assumption (b) is trivial for the case $j = 1$. We will continue by induction over $1 \leq j \leq \tilde{r} + 1$. We can construct a composed trajectory $(\ve{0}_{l(\mathcal{B})}, 1, \hdots; \ve{0}_{l(\mathcal{B}) + j-1}) \in \mathcal{B}_{MPUM,[0, l(\mathcal{B}) + j - 2]}$, using the fact that the values past the $(l(\mathcal{B}) + 1)$st entry of the input sequence do not affect the output sequence, since the true relative degree $r(\Bc)$ of the system is by assumption at least $j-2$. This implies $j \leq r(\Bc)+1$, showing claim (b) for $j+1$. Condition (\ref{ConditionInputFree}) for $i = l(\mathcal{B})+2-j \geq l(\mathcal{B})+1-r(\Bc)$ shows that there is a sequence $(\ve{0}_{l(\mathcal{B})+1-j}, 1, \hdots; \ve{0}_{l(\mathcal{B})}, a) \in A$ for some $a\in\R$. Now, if $j\le \tilde{r}$, then condition~$(\ref{ConditionOutputZero})$ for $i = l(\mathcal{B})+2-j$ implies that $a = 0$, therefore $\ve{v_{j}} \in A$, showing assumption (a) for $j+1$. 

However, if $j=\tilde{r} +1$, then condition~$(\ref{ConditionOutputResponse})$ shows that there is $b \neq 0$ such that $(\ve{0}_{l(\mathcal{B})+1-j}, c, \hdots; \ve{0}_{l(\mathcal{B})}, b) \in A \subseteq \mathcal{B}_{[0,l(\mathcal{B})]}$ for some $c\in\R$. Applying Lemma \ref{InitialCondition} for an arbitrary minimal representation of the unknown system $\mathcal{B}$, we conclude that the state $\ve{x}(l(\mathcal{B}))$ and therefore $b$ is uniquely determined for every fixed choice of input sequence $(\ve{0}_{l(\mathcal{B})+1-j}, c, \hdots)$. Since $r(\Bc)\ge j-1$, the input values past the $(l(\mathcal{B})+2-j)$th index do not affect the output sequence in the observed window, hence there is a unique $(\ve{0}_{l(\mathcal{B})+1-j}, c, \hdots; \ve{0}_{l(\mathcal{B})}, b) \in A$. 

If $c = 0$, then $(\ve{0}_{l(\mathcal{B}) +1}; \ve{0}_{l(\mathcal{B})}, b  ) \in \mathcal{B}_{[0,l(\mathcal{B})]}$, implying $b = 0$ by Lemma \ref{InitialCondition}, which contradicts our assumption. Hence $c \neq 0$, and due to linearity $\displaystyle{\left(\ve{0}_{l(\mathcal{B})+1-j}, 1, \hdots ; \ve{0}_{l(\mathcal{B})}, \frac{b}{c}\right) \in A}$. Therefore, $\tilde{a} = a = \frac{b}{c} \neq 0$, and $\ve{v_j} = \ve{v_{\tilde{r} + 1}} \in A$. We can conclude that $\ve{v}(\tilde{a}, \tilde{r}) \in \mathcal{B}_{MPUM,[0,l(\mathcal{B}) + \tilde{r}]}$, hence $\mathcal{D}$ is informative for the relative degree $r = \tilde{r}$, and $r=r(\Bc)$.

\textbf{Step 2}: We show sufficiency in statement (i). Let $\Dc$ be informative for the relative degree~$r= r(\Bc)$. We show that the conditions $(\ref{ConditionInputFree}), (\ref{ConditionOutputZero})$ and $(\ref{ConditionOutputResponse})$ are satisfied for $\tilde{r} = r(\Bc)$, which is sufficient to conclude that Algorithm~\ref{AlgoRelativeDegree} returns $r(\Bc)$. By Lemma \ref{InformativityMPUM}, there exists $a \neq 0$ such that $v(a,r(\Bc)) \in \Bc_{MPUM,[0,l(\Bc)+r(\Bc)]}$. Since all parts of this sequence of length $l(\Bc)+1$ are contained in $\im \begin{bmatrix} \mathcal{H}_{l(\mathcal{B}) + 1} (\ve{u_d}) \\ \mathcal{H}_{l(\mathcal{B}) + 1} (\ve{y_d}) \end{bmatrix}$, there are vectors $\ve{g_i} \in \R^{T-l(\Bc)+1}$ for each $i\geq l(\Bc)+1-r(\Bc)$ such that $\ve{g_i} \perp M_i$ and $\mathcal{H}_{l(\Bc)+1}(\ve{u_d})_{i}\,\ve{g_i} = 1$. We can conclude that condition (\ref{ConditionInputFree}) holds for $\tilde{r} = r(\Bc)$.  
Additionally, we know that for each $i \geq l(\Bc)+2+r(\Bc)$, each sequence of the structure $(\ve{0}_{i-1}, b, \hdots; \ve{0}_{l(\Bc)}, c)$ with $b, c \in \R$ has to satisfy $c = 0$ by the uniqueness of the relative degree. Therefore, every $\ve{g_i} \perp M_i$ satisfies $\mathcal{H}_{l(\Bc)+1}(\ve{y_d})_{l(\Bc)+1}\,\ve{g_i} = 0$, yielding condition~$(\ref{ConditionOutputZero})$ for $\tilde{r} = r(\Bc)$.
Moreover, the last part of the sequence $v(a,r)$ shows that there is a vector $\ve{\tilde{g}} \in \R^{T-l(\Bc)+1}$ such that $\ve{\tilde{g}} \perp M_{l(\Bc)+1-r(\Bc)}$ and $ \mathcal{H}_{l(\mathcal{B})+1}(\ve{y_d})_{l(\mathcal{B})+1}\, \ve{\tilde{g}} = a \neq 0$, implying condition $(\ref{ConditionOutputResponse})$ for $\tilde{r} = r(\Bc)$. 

\textbf{Step 3}: We show necessity in statement (ii). Assume that the output of Algorithm~\ref{AlgoRelativeDegree} is $\tilde{r} = -1$. We show that it is not possible to determine the true relative degree of the system $\mathcal{B}$ by only using the provided data $\mathcal{D} = (\ve{u_d}, \ve{y_d})$. For this, by Lemma~\ref{InformativityMPUM} it is sufficient  to show that there is no trajectory $\ve{v}(a,r)$ for $r \in \N_0$ and $a \neq 0$ such that each window $\ve{v_j}$ of length $l(\mathcal{B}) + 1$ is contained in $A$. 

There are two lines in the algorithm that cause the output to be $\tilde{r} = -1$. We first want to discuss the case that there is $1\le i \le l(\Bc)+1$ such that 
\begin{align}
\mathcal{H}_{l(\mathcal{B})+1} (\ve{u_d})_i^{\top} & \in M_i \text{, and } \label{ConditionUiZero} \\
\mathcal{H}_{l(\mathcal{B})+1}(\ve{y_d})_{l(\mathcal{B})+1}^{\top} &\in M_j \text{ for all }i< j \le l(\Bc)+1. \label{ConditionOutputIsZero}
\end{align}
Assume that the data $\mathcal{D}$ is informative for the relative degree~$r(\Bc)$, which means by Lemma \ref{InformativityMPUM} that there is a $\ve{v}(a,r(\Bc))$ for some $a \neq 0$ such that all windows of length $l(\mathcal{B}) + 1$ are contained in $A$. Condition (\ref{ConditionOutputIsZero}) for $j = i + 1 $ shows that $(\ve{0}_{i}, \hdots; \ve{0}_{l(\mathcal{B})}, a) \in A$ implies $a = 0$. Therefore, the input values $u(i+1), \hdots, u(l(\mathcal{B}) +1)$ do not affect the output $y(l(\mathcal{B}) + 1)$. Hence, $r(\Bc) \geq (l(\mathcal{B} ) + 2) - (i + 1) = l(\mathcal{B}) +1-i$, implying that $\ve{v_{l(\mathcal{B}) + 1 - (i-1)}} = (\ve{0}_{i-1}, 1, \hdots; \ve{0}_{l(\mathcal{B})}, a) \in A $. By definition of $\ve{A}$, this sequence $\ve{v_{l(\Bc)+2-i}}$ is contained in $\im \begin{bmatrix} \mathcal{H}_{l(\mathcal{B}) + 1} (\ve{u_d}) \\ \mathcal{H}_{l(\mathcal{B}) + 1} (\ve{y_d}) \end{bmatrix}$, hence by choice of the subspace $M_i$ there exists $\ve{g} \in\R^{T-l(\Bc)+1}$ such that $\ve{g} \perp M_i$ and $\mathcal{H}_{l(\mathcal{B}) + 1} (\ve{u_d})_i\, \ve{g} = 1$. Therefore, $\mathcal{H}_{l(\Bc)+1}(\ve{u_d})_i^{\top} \notin M_i$, contradicting condition (\ref{ConditionUiZero}). This shows that there exists no sequence of the shape $\ve{v}(a,r)$ for $a \neq 0$ such that all windows of length $l(\mathcal{B}) + 1$ are contained in $A$.

It remains to discuss the case of the output $\tilde{r} = -1$ being caused by the conditions
\begin{align}
\mathcal{H}_{l(\mathcal{B})+1}(\ve{u_d})_i^{\top} &\notin M_i \text{ for all } 1 \leq i \leq l(\mathcal{B})+1 \\
\mathcal{H}_{l(\mathcal{B})+1}(\ve{y_d})_{l(\mathcal{B})+1}^{\top} &\in M_i  \text{ for all } 1 \leq i \leq l(\mathcal{B})+1. \label{ConditionOutputFixed}
\end{align}
Note that condition (\ref{ConditionOutputFixed}) for $i = 1$ shows that any $(\hdots; \ve{0}_{l(\mathcal{B})}, a) \in A$ has to satisfy $a = 0$, hence there is no sequence of the shape $\ve{v}(a,r)$ with $a \neq 0$ for which every window of length $l(\mathcal{B}) + 1$ is contained in $A$. 

\textbf{Step 4}: We show sufficiency in statement (ii). Assume that  $\mathcal{D}$ is not informative for the relative degree~$r$ for all $r \in \N_0$. Seeking a contradiction, assume that Algorithm~\ref{AlgoRelativeDegree} does not return $\tilde{r} = -1$, which means that it returns $\tilde r \ge 0$. Then statement~(i) implies that $\mathcal{D}$ is informative for the relative degree~$\tilde r$, a contradiction. This completes the proof.
\end{proof}

\begin{remark} \label{ACanBeDetermined}
A proper inspection of the proof of Theorem~\ref{SISORelativgradAlgo} reveals that, whenever the data is informative for the relative degree~$r$, then it is also possible to determine the first non-zero value of the impulse response $\ve{H}(r-1) \neq 0$ from the data.
\end{remark}

\begin{remark}
We can observe that Algorithm~\ref{AlgoRelativeDegree} relies on previous knowledge about the lag of the system. Data sequences can be used to determine an upper and lower bound on the lag of the system, if an upper bound on the McMillan degree $n(\mathcal{B})$ is known, see~\cite{Camlibel.2024}. If we use a potentially larger value of $l$ than the true lag of the system, Algorithm~\ref{AlgoRelativeDegree} may indicate that the data is not informative of the relative degree of the system, even though it is informative of the relative degree of the system when the correct lag $l(\mathcal{B})$ is used. However, if the algorithm claims informativity for a certain relative degree, this result is still reliable. If a potentially smaller value of $l$ is used, then the data set may incorrectly be determined to be informative for some relative degree, even though the data set is not informative of any relative degree or may be informative of a different relative degree. However, if the algorithm claims that the data set is not informative for any relative degree, this result is reliable.
\end{remark}

In order to extend Theorem~\ref{SISORelativgradAlgo} to the MIMO case, we do not need to know the relative degree of every induced SISO system precisely. Since we will minimize over a certain set of relative degrees, it will be sufficient to find lower bounds on the relative degrees that we cannot determine.

\begin{lemma}\label{Lem:lower_bound_RG}
Consider a SISO system $\mathcal{B} \in \mathcal{L}^2$ with known lag $l(\mathcal{B})$ and a data sequence $\mathcal{D} = (\ve{u_d},\ve{y_d}) \in \mathcal{B}_{[0,T]}$. Then we have $r(\Bc) \geq l(\mathcal{B}) + 2 - j$, where 
$$ j = \min \left\{ 1 \leq i \leq l(\mathcal{B}) + 1 : \mathcal{H}_{l(\mathcal{B}) + 1}(\ve{u_d})_i^{\top} \notin \im\begin{bmatrix} \mathcal{H}_{l(\mathcal{B}) + 1}(\ve{u_d})_{1,\hdots,i-1}\\ \mathcal{H}_{l(\mathcal{B}) + 1}(\ve{y_d})\end{bmatrix}^\top \right\}.$$
\end{lemma}
\begin{proof}
By definition of~$j$ there exists $(\ve{0}_{j-1}, 1, \hdots; \ve{0}_{l(\mathcal{B}) + 1}) \in \mathcal{B}_{[0,l(\mathcal{B})]}$. By definition of the relative degree, there is a sequence $\ve{v}(a,r(\Bc)) \in \mathcal{B}_{[0,l(\mathcal{B})+r(\Bc)]}$ for some $a \neq 0$. Hence, $(\ve{0}_{j-1}, 1, \hdots; \ve{0}_{j-1+ r(\Bc)}, a, \hdots)\in\Bc$, and therefore $j-1+r(\Bc) \geq l(\mathcal{B}) + 1$.
\end{proof}

We are now in the position to determine the vector relative degree of a MIMO system from data, where instead of assuming informativity for each relative degree $r_{ij}$ of the induced SISO systems $\mathcal{B}^{i,j}$, we rely on Remark~\ref{MIMOfromSISO} and lower bounds for the unidentifiable relative degrees in terms of the known ones. In order to check whether these conditions hold true, the lower bounds from Lemma~\ref{Lem:lower_bound_RG} may be employed.

\begin{corollary} \label{RelativgradMIMOVerschärft}
Let $\mathcal{B} \in \mathcal{L}^{w}$ be an unknown system and $\mathcal{D} = (\ve{u_d},\ve{y_d}) \in \mathcal{B}_{[0,T]}$ be a data sequence. Let $A_{\mathcal{D}}$ be the set of pairs $(i,j)$ with $1 \leq i \leq p(\mathcal{B})$ and $1 \leq j \leq m(\mathcal{B})$, so that $\mathcal{D}$ is informative for the relative degree $r_{ij}$ of the induced SISO system $\mathcal{B}^{i,j}$. Assume that for each $1 \leq i \leq p(\mathcal{B})$, there exists $1\leq j \leq m(\mathcal{B})$ such that $(i,j) \in A_{\mathcal{D}}$, and for each pair $(i,j) \notin A_{\mathcal{D}}$, we have 
\begin{equation} \label{eq:CondRij}
r_{ij} \geq \min_{(i,k) \in A_{\mathcal{D}}} r_{ik} + 1.
\end{equation}
Let $\ve{G}\in\R^{p(\Bc)\times m(\Bc)}$ be the matrix with entries $\ve{G}_{ij}$ constructed according to Remark~\ref{ACanBeDetermined} for $(i,j) \in A_{\mathcal{D}}$ with $r_{ij} = \min_{(i,k) \in A_{\mathcal{D}}} r_{ik}$, and $\ve{G}_{ij}=0$ else.
Then the data $\mathcal{D}$ is informative for the system having decoupling matrix $\ve{G}$, and if the vector relative degree exists, then $\mathcal{D}$ is informative for the vector relative degree $\ve{r} = (r_1, \hdots, r_{p(\mathcal{B})})$, where $r_i = \min_{(i,j) \in A_{\mathcal{D}}} r_{ij}$.
\end{corollary}
\begin{proof}
Let $\tilde \Bc\in\Sigma_\Dc$. We show that $\tilde \Bc$ has decoupling matrix $\ve{G}$. Observe that, for $1 \leq i \leq p(\mathcal{B})$, $r_i(\tilde \Bc)$ is  well-defined by the assumption that $(i,j) \in A_{\mathcal{D}}$ for at least one $1 \leq j \leq m(\mathcal{B})$, and that $r_{i} = \min_{(i,j) \in A_{\mathcal{D}}} r_{ij} = \min_{1 \leq j \leq m(\mathcal{B})} r_{ij}$ by the imposed lower bound. Therefore, $\ve{r}= (r_1, \hdots, r_{p(\mathcal{B})})$ is the candidate for the vector relative degree $\ve{r}(\tilde \Bc)$. For each pair $(i,j) \notin A_{\mathcal{D}}$, the assumed lower bound $r_{ij} \geq\min_{(i,k) \in A_{\mathcal{D}}} r_{ik} + 1 > r_i$ implies that the $(i,j)$-entry of the decoupling matrix of $\tilde \Bc$ is zero. Similarly, for each pair $(i,j) \in A_{\mathcal{D}}$ with $r_{ij} > \min_{(i,k) \in A_{\mathcal{D}}} r_{ik}$, the $(i,j)$-entry of the decoupling matrix of $\tilde \Bc$ is zero. Furthermore, by the construction in Remark~\ref{ACanBeDetermined}, $\Dc$ is informative for $\mathcal{B}^{i,j}$ having decoupling matrix (a scalar in this case) $\ve{G}_{ij}$. As a consequence, the data $\mathcal{D}$ is informative for the system having decoupling matrix $\ve{G}$.

If the vector relative degree exists, then the decoupling matrix of $\Bc$ has full row rank, and since $\Bc\in\Sigma_\Dc$ it follows that $\rk \ve{G} = p(\Bc)$. This implies that any $\tilde \Bc\in\Sigma_\Dc$ has vector relative degree $\ve{r} = (r_1, \hdots, r_{p(\mathcal{B})})$.
\end{proof}

Note that Corollary \ref{RelativgradMIMOVerschärft} is not an equivalence as Theorem \ref{SISORelativgradAlgo}, since it may be possible to check the decoupling condition without knowledge of all entries of the decoupling matrix. In this case, the lower bound can be relaxed to $r_{ij} \geq \min_{(i,k) \in A_{\mathcal{D}}} r_{ik}$ for some pairs $(i,j)$, and the respective entry of the decoupling matrix may be zero or non-zero. This will be demonstrated in the following example.

\begin{example}
We want to consider the following data sequence for a system $\mathcal{B} \in \mathcal{L}^4$ with $m(\mathcal{B}) = p(\mathcal{B}) = l(\mathcal{B}) = 2$, which consists of input and output vectors for the time steps $0$ through $8$:
\[
\begin{array}{ccccccccccc}
  &       &       & \!\!\!\downarrow &       &         &       &       & \!\!\!\downarrow &     \\
  \ve{U} = \bigg[ &
  \begin{pmatrix} 0 \\ 0 \end{pmatrix}, &
  \begin{pmatrix} 0 \\ 0 \end{pmatrix}, &
  \begin{pmatrix} 1 \\ 0 \end{pmatrix}, &
  \begin{pmatrix} 0 \\ 0 \end{pmatrix}, &
  \begin{pmatrix} 0 \\ 0 \end{pmatrix}, &
  \begin{pmatrix} 0 \\ 0 \end{pmatrix}, &
  \begin{pmatrix} 0 \\ 0 \end{pmatrix}, &
  \begin{pmatrix} 0 \\ 1 \end{pmatrix}, &
  \begin{pmatrix} 0 \\ 0 \end{pmatrix}
  \bigg], \vspace{0.3cm} \\
  \ve{Y} = \bigg[ &
  \begin{pmatrix} 0 \\ 0 \end{pmatrix}, &
  \begin{pmatrix} 0 \\ 0 \end{pmatrix}, &
  \begin{pmatrix} 0 \\ 0 \end{pmatrix}, &
  \begin{pmatrix} 0 \\ 0 \end{pmatrix}, &
  \begin{pmatrix} 1 \\ 0 \end{pmatrix}, &
  \begin{pmatrix} 0 \\ 0 \end{pmatrix}, &
  \begin{pmatrix} 0 \\ 0 \end{pmatrix}, &
  \begin{pmatrix} 0 \\ 0 \end{pmatrix}, &
  \begin{pmatrix} 0 \\ 1 \end{pmatrix}
  \bigg].
\end{array}
\]

Observe that, by Lemma \ref{InitialCondition}, the state of any minimal representation $(\ve{A},\ve{B},\ve{C},\ve{D})$ of $\Bc$ at the indicated time steps 2 and 7 is zero by linearity. This shows that the vectors $\ve{y}(2), \hdots, \ve{y}(6)$ are the impulse response of the system corresponding to the first input coordinate, that is
\[
   \ve{C A}^k\ve{B} \begin{pmatrix} 1 \\ 0 \end{pmatrix} = \ve{y}(2+k),\quad k=0,\ldots,4.
\]
We can conclude that the relative degrees $r_{ij}$ of the induced SISO systems $\Bc^{i,j}$ for $j=1$ are $r_{11} = 2$, while $r_{21} = \infty$ by Lemma \ref{rLeqL} and the fact that the observed impulse response covers more than $l(\mathcal{B}) = 2$ time steps. Moreover, the impulse response corresponding to the second input coordinate is given by $\ve{C A}^k\ve{B} \begin{pmatrix} 0 \\ 1 \end{pmatrix} = \ve{y}(7+k)$ for $k=0,1$, thus $r_{22} = 1$ and $r_{12} > 1$. While the lower bound for the relative degree $r_{12} \geq 2$ does not satisfy the condition required in Corollary~\ref{RelativgradMIMOVerschärft}, this allows us to conclude that the decoupling matrix is of the shape $\ve{G} = \begin{bmatrix} 1 & b \\ 0 & 1 \end{bmatrix}$ for some $b \in \R$, hence the system has vector relative degree $\ve{r} = (2,1)$, since $\ve{G}$ is invertible for every $b \in \R$. 

However, it is not possible to determine whether $b$ is zero or non-zero from the provided data, i.e., it is not possible to completely identify the decoupling matrix. We will construct a representation of a system $\mathcal{B}_a \in \mathcal{L}^4$ depending on the parameter $a \in \R$ that satisfies $(\ve{U},\ve{Y}) \in \mathcal{B}_{a, [0,8]}$ for each choice of the parameter $a$ such that the undetermined entry in the decoupling matrix will be $a$. Consider the state space representation
\begin{align*}
\ve{x_1}(t+1) &= \ve{x_2}(t) + a\ve{x_3}(t)\\
\ve{x_2}(t+1) &= \ve{u_1}(t) \\
\ve{x_3}(t+1) &= \ve{u_2}(t) \\
\ve{y_1}(t) &=  \ve{x_1}(t)\\
\ve{y_2}(t) &= \ve{x_3}(t)
\end{align*}
and let $\Bc_a$ be the behavior induced by it. Then $l(\mathcal{B}_a) = 2$, since we can determine the truncated observability matrix $\mathcal{O}_2(\ve{C},\ve{A}) = \begin{bmatrix} 1 & 0 & 0 \\ 0 & 0 & 1 \\ 0 & 1 & 0 \\ 0 & 0 & 0 \end{bmatrix}$, which has full rank $3$, and it is easy to confirm that $(\ve{U},\ve{Y}) \in \mathcal{B}_{a, [0,8]}$ with the corresponding state sequence 
\[
\begin{array}{ccccccccccc}
  \ve{X} = \Bigg[ &
  \begin{pmatrix} 0 \\ 0 \\ 0 \end{pmatrix}, &
  \begin{pmatrix} 0 \\ 0 \\ 0 \end{pmatrix}, &
  \begin{pmatrix} 0 \\ 0 \\ 0 \end{pmatrix}, &
  \begin{pmatrix} 0 \\ 1 \\ 0 \end{pmatrix}, &
  \begin{pmatrix} 1 \\ 0 \\ 0 \end{pmatrix}, &
  \begin{pmatrix} 0 \\ 0 \\ 0 \end{pmatrix}, &
  \begin{pmatrix} 0 \\ 0 \\ 0 \end{pmatrix}, &
  \begin{pmatrix} 0 \\ 0 \\ 0 \end{pmatrix}, &
  \begin{pmatrix} 0 \\ 0 \\ 1 \end{pmatrix}
  \Bigg].
\end{array}
\]
However, the relative degree of the induced SISO system $\mathcal{B}^{1,2}_a$ is $$r_{12} = \begin{cases} 2, & \text{ if } a \neq 0 \\ \infty, & \text{ if } a = 0 \end{cases}$$ and the decoupling matrix of $\mathcal{B}_a$ is $\ve{G} = \begin{bmatrix} 1 & a \\ 0 & 1 \end{bmatrix}$, which depends on the choice of $a \in \R$.
\end{example}

\begin{remark}
Observe that the conditions stated in Corollary~\ref{RelativgradMIMOVerschärft} are equivalently characterizing the informativity for the decoupling matrix of the system. The informativity of the vector relative degree is a weaker condition, as seen in the previous example. By exchanging condition~\eqref{eq:CondRij} with 
\begin{equation}
r_{ij} \geq \min_{(i,k) \in A_{\mathcal{D}}} r_{ik},
\end{equation}
we can determine the subset of entries of the decoupling matrix satisfying either $(i,j) \in A_{\Dc}$, or condition~\eqref{eq:CondRij} holds. These entries characterize all restrictions that a decoupling matrix of this system needs to satisfy, hence describing the set of possible decoupling matrices $G_{\Dc}$. The data set $\Dc$ is informative for the vector relative degree of the system if, and only if, $G_{\Dc} \subseteq \{ A \in \R^{p(\mathcal{B})\times m(\mathcal{B})} \vert \operatorname{rank} A = p(\mathcal{B}) \}$ holds.
\end{remark}

\section{Zero dynamics of linear systems}\label{secZD}

The zero dynamics of a system refers to the internal dynamics that remain under the constraint that the output is constantly zero. The condition of stability of the zero dynamics is essential for many control approaches, because it ensures the boundedness of internal states under output tracking. It is a classical assumption in adaptive control~\cite{ByrnWill84, KhalSabe87, Mare84, Mors83}, often known as \textit{minimum phase}, see also~\cite{IlchWirt13}. Furthermore, it is relevant for the closed-loop stability when using plant inversion feedback control methods~\cite{Buffington.1998}. If the zero dynamics are unstable, it may not be possible to stabilize the system with certain control strategies, as discussed in~\cite{Hu.2023}. This motivates the need for identification of stability of the zero dynamics for systems without a fully known model.

\subsection{Stability of zero dynamics using persistently exciting input data}

In this section, we determine the stability of the zero dynamics of a system by using input/output data. Recall that the zero dynamics of a system $\mathcal{B} \in \mathcal{L}^w$ is defined by
\[
    \mathcal{B}_{ZD} = \{ \ve{u}: \N_0 \rightarrow \R^{m(\mathcal{B})} :  (\ve{u},\ve{0}) \in \mathcal{B}\}.
\]
Stability of the zero dynamics is defined as follows.

\begin{definition}
For a system $\Bc\in\Lc^w$, we call its zero dynamics $\mathcal{B}_{ZD}$ stable, if for any $\ve{u} \in \mathcal{B}_{ZD}$ and any minimal state space representation~\eqref{isosystem} of $\mathcal{B}$, the state sequence for any initial condition $\ve{x}(0) = \ve{x^0}$ and the input sequence converge to zero, $(\ve{x}(i), \ve{u}(i)) \rightarrow \ve{0}$ for $i \rightarrow \infty$. 
\end{definition}

\begin{remark} \label{StabilityDependsOnlyOnU}
The stability of the zero dynamics of a system $\Bc\in\Lc^w$ does not depend on the choice of the minimal state space representation, as we show in the following. If $n(\Bc)=0$, then any minimal representation has the form $\ve{y}(t) = \ve{Du}(t)$, $t\in\N_0$, for some suitable matrix $\ve{D}$, with zero-dimensional state space. In this case, $\ve{u} \in \mathcal{B}_{ZD}$ is equivalent to $\ve{u}(i) \in \ker \ve{D}$ for all $i \in \N_0$. Therefore, stability of the zero dynamics is equivalent to $\ker \ve{D} = \{ \ve{0} \}$.

If $n(\Bc) > 0$, fix two minimal representations characterized by matrices $(\ve{A},\ve{B},\ve{C},\ve{D})$ and $(\tilde{\ve{A}},\tilde{\ve{B}},\tilde{\ve{C}},\tilde{\ve{D}})$, respectively, and choose  arbitrary $\ve{z} \in \R^{n(\Bc)}$. The zero input response $(\ve{0}; (\ve{CA^kz})_{k\in \N_0})$ is contained in the first behavior, and since both representations are behaviorally equivalent, there exists an initial state $\ve{\tilde{z}} \in \R^{n(\Bc)}$ such that $(\ve{0}; (\ve{CA^kz})_{k\in \N_0}) = (\ve{0}; (\ve{\tilde{C}\tilde{A}^k\tilde{z}})_{k\in \N_0})$.
Define the transformation $\ve{\beta}:\R^{n(\Bc)}\rightarrow \R^{n(\Bc)}$ via $z\mapsto \tilde{z}$. Observe that this is well-defined, since $\tilde{z}$ is uniquely determined due to the observability (by Lemma~\ref{minimalObservable}) of the second minimal representation. Moreover, $\ve{\beta}$ is linear, and it is injective due to the observability of the first minimal representation. Since $\ve{\beta}$ maps between two spaces with matching dimensions, it is therefore bijective. Fix a trajectory $\ve{u} \in \mathcal{B}_{ZD}$ and let $(\ve{x}(k))_{k \in \N_0}, (\ve{\tilde{x}}(k))_{k \in \N_0}$ be the state trajectories corresponding to the two representations. For each $k \in \N_0$, the application of the input sequence $(\ve{u}(k), \hdots)$ to the initial state $\ve{\beta}(\ve{x}(k))$ of the second representation yields a constant zero output sequence, since the zero input response of this initial state satisfies $(\ve{\tilde{C}\tilde{A}}^j\ve{\beta}(\ve{x}(k)))_{j\in \N_0} = (\ve{CA}^j\ve{x}(k))_{j \in \N_0}$, the impulse responses of both representations are equal by Lemma~\ref{behaviorallyEquivalent} and 
\[
    \left((\ve{u}(k), \hdots),\,\Big(\ve{CA}^j\ve{x}(k) + \sum\nolimits_{l=0}^{j-1} \ve{CA}^l\ve{B}u(k+j-l-1)\Big)_{j \in \N_0}\right) \in \Bc.
\]

Hence, the state trajectory $(\ve{\beta}(\ve{x}(k)))_{k \in \N_0}$ is a possible choice for $(\ve{\tilde{x}}(k))_{k \in \N_0}$. Due to observability of minimal representations, this choice is unique, hence the states satisfy $\beta(\ve{x}(k)) = \ve{\tilde{x}}(k)$ for each $k \in \N_0$. Assume that for the first representation we have $(\ve{x}(i), \ve{u}(i)) \rightarrow \ve{0}$ for $i \rightarrow \infty$. Then $(\tilde{\ve{x}}(i), \ve{u}(i)) = (\ve{\beta}(\ve{x}(i)), \ve{u}(i)) \rightarrow \ve{0}$ for $i \rightarrow \infty$, hence the stability of the zero dynamics does not depend on the choice of the representation.
\end{remark}

Next we discuss the notion of stable zero dynamics in the case $n(\Bc)=0$.

\begin{proposition}
Let $\mathcal{B} \in \mathcal{L}^w$ be a system with McMillan degree $n(\mathcal{B}) = 0$ and $\mathcal{D} = (\ve{u_d},\ve{y_d})\in \Bc_{[0,T]}$ be a sampled data sequence. Then $\Dc$ is informative for stability of the zero dynamics if, and only if, the condition $\rk \mathcal{H}_1(\ve{y_d})=m(\mathcal{B})$ holds. 
\end{proposition}
\begin{proof}
Since $n(\mathcal{B}) = 0$, any minimal state space representation of an arbitrary $\Bc' \in \Sigma_{\Dc}$ has the form $\ve{y}(t) = \ve{Du}(t)$, $t\in\N_0$, for some $\ve{D} \in \R^{p(\mathcal{B}) \times m(\mathcal{B})}$. For any such $\ve{D}$ we have $\im \mathcal{H}_1(\ve{y_d}) \subseteq \im \ve{D}$ thus $\rk \mathcal{H}_1(\ve{y_d}) \le \rk \ve{D}$. Now, if $\rk \mathcal{H}_1(\ve{y_d})=m(\mathcal{B})$, then  $m(\mathcal{B})\le \rk \ve{D} \le m(\mathcal{B})$, hence $\rk \ve{D} = m(\mathcal{B})$ and $\ker \ve{D} = \{ \ve{0} \}$. This shows that all $\mathcal{B'}_{ZD}$ are stable, hence $\Dc$ is informative of the stability of the zero dynamics.

Now let $\Dc$ be informative for the stability of the zero dynamics  of the system, which in particular means that $\mathcal{B}_{ZD}$ are stable. Hence, for any $\ve{D} \in \R^{p(\mathcal{B}) \times m(\mathcal{B})}$ such that $\ve{y}(t) = \ve{Du}(t)$, $t\in\N_0$, is a minimal state space representation of $\Bc$, we have that $\ker \ve{D} = \{ \ve{0} \}$. Fix any such $\ve{D}$. Then, in particular, we have $\ve{y_d}(i) = \ve{D u_d}(i)$ for $0\le i\le T$, thus $\mathcal{H}_1(\ve{y_d}) = \ve{D}\mathcal{H}_1(\ve{u_d})$ and hence, by Sylvester's rank inequality,
\[
    \rk \mathcal{H}_1(\ve{y_d}) \ge \rk\ve{D} + \rk \mathcal{H}_1(\ve{u_d}) - m(\Bc).
\]
Now, seeking a contradiction, assume that $\rk \mathcal{H}_1(\ve{y_d}) < m(\mathcal{B})$. If $\rk \mathcal{H}_1(\ve{u_d}) = m(\mathcal{B})$, then the above inequality implies $\rk \mathcal{H}_1(\ve{y_d}) \ge \rk\ve{D} = m(\Bc)$ as $\ve{D}$ has full column rank, a contradiction. If, on the other hand, $\rk \mathcal{H}_1(\ve{u_d}) < m(\mathcal{B})$, then there exists $\ve{u} \in \R^{m(\mathcal{B})} \setminus \im \mathcal{H}_1(\ve{u_d}) $. Next, we define a matrix $\tilde{\ve{D}} \in \R^{p(\mathcal{B}) \times m(\mathcal{B})}$ such that $\tilde{\ve{D}} \ve{u_d}(i) = \ve{y_d}(i)$, for all $0 \leq i \leq T$, and 
$\tilde{\ve{D}} \ve{u} = \ve{0}$. To this end, let $V\in\R^{m(\Bc)\times k}$ be a matrix with $\rk V = k$ and $\im V = \im  \mathcal{H}_1(\ve{u_d})$. Further let $W\in\R^{m(\Bc)\times (m(\Bc)-k)}$ be such that $[V,W]$ is invertible. Then there exist $U\in \R^{k\times (T+1)}$ and $z\in\R^{m(\Bc)-k}$ such that $\mathcal{H}_1(\ve{u_d}) = VU$ and $\ve{u} = W z$. Define 
\[
    \tilde{\ve{D}} := {\ve{D}} [V,0] [V,W]^{-1}.
\]
Then we have 
\begin{align*}
    \tilde{\ve{D}}\mathcal{H}_1(\ve{u_d}) &= \tilde{\ve{D}} [V,W] \begin{bmatrix} U \\ 0 \end{bmatrix} = {\ve{D}} [V,0] \begin{bmatrix} U \\ 0 \end{bmatrix} =  {\ve{D}} \mathcal{H}_1(\ve{u_d}),\\
     \tilde{\ve{D}} \ve{u} &=  \tilde{\ve{D}} [V,W] \begin{pmatrix} 0 \\ z \end{pmatrix} = {\ve{D}} [V,0] \begin{pmatrix} 0 \\ z \end{pmatrix} = \ve{0}
\end{align*}
as desired. Therefore, $\ve{0} \neq \ve{u} \in \ker \tilde{\ve{D}}$, but $\ve{y}(t) = \tilde{\ve{D}} \ve{u}(t)$, $t\in\N_0$, is also the minimal state space representation of a system $\mathcal{B}' \in \Sigma_{\mathcal{D}}$ that explains the data, and does not have stable zero dynamics as $(i\mapsto \ve{u}(i) = \ve{u})\in \Bc_{ZD}'$, a contradiction. 
\end{proof}

\begin{remark}
Observe that in the case of $n(\mathcal{B}) = 0$, it is only possible for the data $\mathcal{D}$ to be informative for the stability of the zero dynamics of the system $\Bc$, if it can be uniquely identified by the data provided. Furthermore, the case of multiple data sequences $(\ve{u_d^{i}}, \ve{y_d^{i}}) \in \mathcal{B}_{[0,T_i]}$ for $1 \leq i \leq q$, is equivalent to the composed data sequence $(\ve{u_d^{1}}, \hdots, \ve{u_d^{q}}; \ve{y_d^{1}}, \hdots, \ve{y_d^{q}}) \in \mathcal{B}_{[0, \sum_{i = 1}^q T_i]}$ , since there is no dependency between time steps. Hence, the same rank condition $m(\mathcal{B}) = \rk \begin{bmatrix} \mathcal{H}_1(\ve{y_d^1}) & \hdots & \mathcal{H}_1(\ve{y_d^q}) \end{bmatrix}$ for the composed Hankel matrix determines whether the data set is informative for the stability of the zero dynamics. 
\end{remark}

Now we turn attention to the case $n(\mathcal{B}) > 0$. We will restrict ourselves to the analysis of square systems where $m(\mathcal{B}) = p(\mathcal{B})$ in order to use the following condition.

\begin{lemma} \label{SameDimensions} 
Consider a square system $\mathcal{B} \in \mathcal{L}^{2m(\mathcal{B})}$ with McMillan degree $n(\mathcal{B}) > 0$ and vector relative degree $\ve{r}(\Bc) = (r_1,\ldots,r_{p(\mathcal{B})}) \in \N_0^{p(\mathcal{B})}$, and let $L \geq l(\mathcal{B}) + \max_{1 \leq i \leq p(\mathcal{B})} r_i + 1$ and $(\ve{u_d},\ve{y_d}) \in \mathcal{B}_{[0,L-1]}$ be a sampled data sequence. Define
\begin{align} \label{DefineSubspaceM}
M = \{ u : (u,0)\in \Bc_{MPUM,[0,L-1]} \} 
\end{align}
and let $\ve{\pi_k}:\R^{m(\mathcal{B})L} \rightarrow \R^{m(\mathcal{B})k}$ be the projection onto the first $m(\mathcal{B}) k$ entries, that is $\ve{\pi_k}(x_1, \hdots, x_{m(\mathcal{B})L}) = (x_1, \hdots, x_{m(\mathcal{B})k})$, for any $1 \leq k \leq L$. Then $$\dim \ve{\pi_{l(\mathcal{B}) + 1}} M = \dim \ve{\pi_{l(\mathcal{B})}} M.$$
\end{lemma}

\begin{proof}
Clearly, $\dim \ve{\pi_{l(\mathcal{B})}} M\le \dim \ve{\pi_{l(\mathcal{B}) + 1}} M$. Seeking a contradiction, assume that $\dim \ve{\pi_{l(\mathcal{B}) + 1}} M > \dim \ve{\pi_{l(\mathcal{B})}} M $. Then there exist $\ve{u} \neq \ve{v} \in M$ such that their projections onto the first $l(\mathcal{B})$ elements coincide, $(\ve{u_0}, \hdots, \ve{u_{l(\mathcal{B}) -1}}) = (\ve{v_0}, \hdots, \ve{v_{l(\mathcal{B}) -1}})$, but their $(l(\mathcal{B}) + 1)$st entry differs, $\ve{u_{l(\mathcal{B})}} \neq \ve{v_{l(\mathcal{B})}}$. Their difference is also contained in the subspace, $\ve{q} \coloneqq \ve{\pi_{l(\mathcal{B}) + 1}} (\ve{u}) - \ve{\pi_{l(\mathcal{B}) + 1}} (\ve{v}) \in \ve{\pi_{l(\mathcal{B}) + 1}} M$, and satisfies $(\ve{q_0}, \hdots, \ve{q_{l(\mathcal{B}) -1}}) = \ve{0}$. By definition of $\ve{\pi_{l(\mathcal{B}) + 1}}$, there exists a suitable continuation $\ve{\tilde{q}} \in (\R^{m(\mathcal{B})})^{L-l(\mathcal{B}) -1} $ of the input sequence $\ve{q}$ such that $(\ve{q}, \ve{\tilde{q}}) \in M$. By definition of $M$, this means that  $\left( \ve{q}, \ve{\tilde{q}}, \ve{0}_{L} \right) \in \Bc_{MPUM,[0,L-1]}.$ 
Since $\Bc\in \Sigma_\Dc$, it follows from the definition of the MPUM that $(\ve{q}, \ve{\tilde{q}}; \ve{0}) \in \mathcal{B}_{[0, L-1]}$. 

By Lemma~\ref{InitialCondition}, the state in any minimal representation of the system will be initialised to $\ve{x}(l(\mathcal{B})) = 0$ by the initializing sequence $(\ve{q}_{1, \hdots l(\mathcal{B})},\ve{0}_{l(\mathcal{B})}) \in \mathcal{B}_{[0,l(\mathcal{B})-1]}$. Therefore, the output sequence $(\ve{y}(l(\mathcal{B})), \hdots, \ve{y}(l(\mathcal{B}) + \max_{1 \leq i \leq p(\mathcal{B})} r_i))$ does not contain a nontrivial zero input response component. Since $\ve{q_{l(\mathcal{B})}} = \ve{u_{l(\mathcal{B})}} - \ve{v_{l(\mathcal{B})}} \neq 0$, the output sequence $$\left(\ve{y}(l(\mathcal{B})), \hdots, \ve{y}(l(\mathcal{B}) + \max_{1 \leq i \leq p(\mathcal{B})} r_i) \right)$$ will contain the impulse response to the nonzero input sequence $(\ve{q_{l(\mathcal{B})}}, \ve{\tilde{q}})$. This will be nonzero by the definition of the vector relative degree, and the invertibility of the decoupling matrix, which contradicts $(\ve{q}, \ve{\tilde{q}}) \in M$. Hence, $\dim \ve{\pi_{l(\mathcal{B}) + 1}} M = \dim \ve{\pi_{l(\mathcal{B})}} M$.
\end{proof}

Lemma~\ref{SameDimensions} shows that for a fixed state on a trajectory in $\mathcal{B}_{ZD}$ of a square system, there is a uniquely determined input sequence such that $\ve{y} = 0$ will be satisfied in all future time steps. Note that in the case of a controllable system $\mathcal{B}$ and an input sequence that is persistently exciting of order $L + n(\mathcal{B})$, we can utilize Theorem~\ref{Fundamentallemma} to conclude that $M$ contains all possible input sequences of length $L$ that are in $\mathcal{B}_{ZD}$. 

In the following, we use the projected subspace $\ve{\pi_{l(\mathcal{B})}} M$ to determine the stability of the zero dynamics. The following normal form provides a decomposition of the state space into the observed states $\ve{\xi}$ and the states $\ve{\eta}$ that can be reached on trajectories $\ve{u} \in \mathcal{B}_{ZD}$.\\

Recall that any minimal representation of a square system $\mathcal{B} \in \mathcal{L}^{2m(\mathcal{B})}$ which has a vector relative degree $\ve{r}(\Bc) = (r_1, \hdots, r_{m(\mathcal{B})}) \in \N^{m(\mathcal{B})}$ can be transformed into Byrnes-Isidori normal form to explicitly decouple the zero dynamics, see \cite[Thm.~2.4]{Mueller.2009}, by application of a suitable state transformation onto some minimal representation of the system, yielding the state $\ve{x} = \begin{pmatrix} \ve{\xi} \\ \ve{\eta} \end{pmatrix}$, where $\ve{\xi} =(\ve{\xi_1}, \hdots, \ve{\xi_{m(\mathcal{B})}})^{\top} $, and $\ve{\xi_i} = (\xi_{i,1}, \hdots, \xi_{i,r_i})^\top$ for $1\le i\le r_i$. Then the system equations are of the shape
\begin{equation}\label{eq:BIF}
\begin{aligned}
\xi_{i,j}(t+1) &= \xi_{i,j+1}(t),\qquad\qquad \text{ for all }  1 \leq j \leq r_i - 1, \\
\xi_{i,r_i}(t+1) &= \ve{\alpha_i} \begin{pmatrix} \ve{\xi}(t) \\ \ve{\eta}(t) \end{pmatrix} + \ve{C_i A}^{r_i - 1} \ve{B u}(t),\\
\ve{\eta}(t+1) &= \ve{Q \eta}(t) + \ve{P y}(t),\\
y_i(t) &=  \xi_{i,1}(t)  ,\qquad\qquad\quad  \text{ for all } 1 \leq i \leq m(\mathcal{B}), 
\end{aligned}
\end{equation}
for some $\ve{\alpha_i} \in \R^{1 \times n(\mathcal{B})}$, $\ve{Q} \in \R^{\left(n(\mathcal{B}) - \sum_{i = 1}^{m(\mathcal{B})} r_i \right) \times \left( n(\mathcal{B}) - \sum_{i=1}^{m(\mathcal{B})} r_i \right)}$, $\ve{P} \in \R^{\left( n(\mathcal{B}) - \sum_{i  = 1}^{m(\mathcal{B})} r_i \right) \times m(\mathcal{B})}$. This shows that the set of possible initial states for trajectories in $\mathcal{B}_{ZD}$ is a $\left(n(\mathcal{B}) - \sum_{i = 1}^{m(\mathcal{B})} r_i\right)$-dimensional vector space. Observe that a system $\mathcal{B} \in \mathcal{L}^{2m(\mathcal{B})}$ has stable zero dynamics if, and only if, it has a minimal representation in Byrnes-Isidori normal form where $\ve{Q}$ is a stable matrix, i.e., all eigenvalues of $\ve{Q}$ have absolute value strictly less than 1. In this case, the matrix $\ve{Q}$ is stable for every Byrnes-Isidori normal form.

Next, we describe the zero dynamics using the data sequences that initialize the states $(\ve{0},\ve{\eta})$ that can be reached along trajectories $ \ve{u} \in \mathcal{B}_{ZD}$.

\begin{lemma} \label{UniqueContinuationInZD}
Consider a system $\mathcal{B} \in \mathcal{L}^{2m(\mathcal{B})}$ with vector relative degree $\ve{r}(\Bc) = (r_1,\ldots,r_{m(\mathcal{B})}) \in \N_0^{m(\mathcal{B})}$ and let $(\ve{u_d},\ve{y_d}) \in \mathcal{B}_{[0,L-1]}$ be a sampled data sequence for some $L \geq l(\mathcal{B}) + \max_{1 \leq i \leq m(\mathcal{B})} r_i + 1$ such that $d \coloneqq \dim \ve{\pi_{l(\mathcal{B})}} M = n(\mathcal{B}) - \sum_{i=1}^{m(\mathcal{B})} r_i$, where $M$ is defined as in (\ref{DefineSubspaceM}). Fix an arbitrary basis $\{ \ve{v_1}, \hdots, \ve{v_{d}} \}$ of $\ve{\pi_{l(\mathcal{B})}} M$. Then for each $1 \leq k \leq d$, there is a unique continuation $\ve{u_k} \in \R^{m(\mathcal{B})}$ such that $(\ve{v_k}, \ve{u_k}) \in \ve{\pi_{l(\mathcal{B}) + 1}} M$, and $\ve{z_k} \coloneqq (\ve{v_{k,2}}, \hdots, \ve{v_{k,l(\mathcal{B})}}, \ve{u_{k}}) \in \ve{\pi_{l(\mathcal{B})}} M.$
\end{lemma}

\begin{proof}
By definition of $\ve{\pi_{l(\mathcal{B})}} M $, we have that for each $1 \leq k \leq d$, there exists some continuation $\ve{u_k} \in \R^{m(\mathcal{B})}$ such that $(\ve{v_k}, \ve{u_k}) \in \ve{\pi_{l(\mathcal{B}) + 1}} M$. We show that this continuation is unique. Fix some $1 \leq k \leq d$ and, seeking a contradiction, assume that $\ve{u_k} \neq \ve{\tilde{u}_k} \in \R^{m(\mathcal{B})}$ are two distinct continuations of the input sequence $\ve{v_k}$ such that $(\ve{v_k}, \ve{u_k}), (\ve{v_k}, \ve{\tilde{u}_k}) \in \ve{\pi_{l(\mathcal{B}) + 1}} M$. Then $(\ve{0}_{m(\mathcal{B}) l(\mathcal{B})} , \ve{u_k} - \ve{\tilde{u}_k}) \in \ve{\pi_{l(\mathcal{B}) + 1}} M$, and therefore the set $\{ (\ve{v_1},\ve{0}_{m(\mathcal{B})}) , \hdots, (\ve{v_d}, \ve{0}_{m(\mathcal{B})}), (\ve{0}_{m(\mathcal{B}) l(\mathcal{B})} , \ve{u_k} - \ve{\tilde{u}_k}) \}$ is a linearly independent subset of $\ve{\pi_{l(\mathcal{B}) + 1}} M$ that contains $d+1$ elements. Since $\dim \ve{\pi_{l(\mathcal{B}) + 1}} M = \dim \ve{\pi_{l(\mathcal{B})}} M = d$ by Lemma~\ref{SameDimensions}, this is a contradiction. 

Next, we show that for each $1 \leq k \leq d$ and $\ve{v_k} = (\ve{v_{k,1}}, \hdots, \ve{v_{k,l(\mathcal{B})}})$ the sequence $\ve{z_k} = (\ve{v_{k,2}}, \hdots, \ve{v_{k,l(\mathcal{B})}}, \ve{u_{k}})$ is also contained in $\ve{\pi_{l(\mathcal{B})}} M$. Since $(\ve{v_{k,1}}, \ve{z_k}) \in \ve{\pi_{l(\mathcal{B}) + 1}} M$, there is a continuation $\ve{a} \in \left(\R^{m(\mathcal{B})}\right)^{L - l(\mathcal{B}) - 1}$ such that $(\ve{v_{k,1}}, \ve{z_k}, \ve{a}) \in M$. Fix a minimal representation of the system in Byrnes-Isidori normal form, and let $(\ve{\xi}(l(\mathcal{B}) + 1), \ve{\eta}(l(\mathcal{B}) + 1))$ be the state, that is initialized by the sequence $\ve{z_k}$ as discussed in Lemma~\ref{InitialCondition}. We can conclude that $\ve{\xi}(l(\mathcal{B}) + 1) = \ve{0}$ holds, because any non-zero entry of $\ve{\xi}(l(\mathcal{B}) + 1) $ will cause a non-trivial output within the subsequent $\max_{1\leq i \leq m(\mathcal{B})} r_i$ time steps, and by choice of $L$, the continuation $\ve{a} $ covers at least those $\max_{1 \leq i \leq m(\mathcal{B}) } r_i$ time steps. Therefore, there is a trajectory in the zero dynamics starting with $(\ve{v_{k,1}}, \ve{z_k})$. This implies that $\ve{z_k} \in \ve{\pi_{l(\mathcal{B})}}M$. 
\end{proof}

Observe that, using the notation from the previous proof, by choice of $M$, we can conclude that $\ve{z_k} \in \mathcal{B}_{MPUM,ZD,[0,l(\mathcal{B}) - 1]} \subseteq \mathcal{B}_{ZD,[0,l(\mathcal{B}) - 1]}$ for every $1 \leq k \leq d$, and finally use that, by the Byrnes-Isidori normal form, the minimal dimension of the state space of the zero dynamics $\mathcal{B}_{ZD}$ is $d = n(\mathcal{B})-\sum_{i = 1}^{m(\mathcal{B})} r_i$, to conclude that 
\begin{align} \label{viExpressBZD}
\mathcal{B}_{ZD,[0,l(\mathcal{B})-1]} = \operatorname{span}\{\ve{v_1}, \hdots, \ve{v_{d}}\}.
\end{align}
Hence, there are unique coefficients $q_{k,i} \in \R$ such that $\displaystyle{\ve{z_k} = \sum_{i=1}^{d} q_{k,i} \ve{v_i}}$. The following lemma shows how these coefficients determine the stability of the zero dynamics of the system.

\begin{lemma} \label{xConvergence}
Let $\mathcal{B} \in \mathcal{L}^{2m(\mathcal{B})}$ with $n(\mathcal{B}) > 0$ be a controllable system with vector relative degree $\ve{r}(\Bc) = (r_1,\ldots,r_{m(\mathcal{B})}) \in \N_0^{m(\mathcal{B})}$, and let $\mathcal{D} = (\ve{u_d},\ve{y_d})\in\Bc_{[0,T]}$ be a sampled data sequence such that $\ve{u_d}$ is persistently exciting of order $L + n(\mathcal{B})$ for some $L \geq l(\mathcal{B}) + \max_{1 \leq j \leq m(\mathcal{B})} r_j + 1$. Construct the coefficients $q_{k,i}$ for $1 \leq k,i \leq n(\mathcal{B}) - \sum_{j = 1}^{m(\mathcal{B})} r_j =: d$ as shown above. Then, for all $\ve{u} \in \mathcal{B}_{ZD}$, the corresponding sequence of states~$\ve{x}$ of any minimal representation of $\Bc$ satisfies $\ve{x}(i) \rightarrow \ve{0}$ for $i \rightarrow \infty$  if, and only if, $d  = 0$ or $\ve{\tilde{Q}} = (q_{k,i})_{1 \leq k,i \leq d}$ is a stable matrix, i.e., all eigenvalues of $\ve{\tilde{Q}}$ have absolute value strictly less than 1.
\end{lemma}

\begin{proof}
If $d = 0$, then $\mathcal{B}_{ZD} = \left\{ \ve{0} \in (\R^{m(\mathcal{B})})^{\N_0} \right\}$, and therefore $\ve{x}(i) = 0 $ for all $i \in \N_0$ for each sequence of states on a trajectory $\ve{u} \in \mathcal{B}_{ZD}$. In case $d > 0$, observe that by Theorem~\ref{Fundamentallemma}, we can uniquely identify $\Bc$ from the data $\mathcal{D}$, hence we can construct matrices that parametrize a minimal representation of the system as shown in~\cite[Thm.~6.4.2]{PoldWill98}. Then we may apply a state space transformation of the system into (some) Byrnes-Isidori normal form. 

We will show that the matrix $\ve{Q}$ of this representation is similar to the constructed matrix~$\ve{\tilde{Q}}$. Let $\ve{w_i} \in \left(\R^{m(\mathcal{B})}\right)^{l(\mathcal{B})}$ be the initializing sequence for the state $\eta_i$, for each $1 \leq i \leq d$. Then $\operatorname{span}\{\ve{v_1}, \hdots, \ve{v_{d}}\} = \mathcal{B}_{ZD,[0,l(\mathcal{B}) - 1]} = \operatorname{span}\{\ve{w_1}, \hdots, \ve{w_{d}}\}$ is a $d$-dimensional vector space, hence there is an invertible matrix $\ve{T} \in \R^{d \times d}$ such that 
\[
     [\ve{w_1}, \ldots, \ve{w_d}] = [\ve{v_1}, \ldots, \ve{v_d}]\, \ve{T}.
\] 
By construction, the application of the coordinate transformation $\begin{bmatrix} \ve{I_{n(\mathcal{B}) - d}} & \ve{0} \\ \ve{0} & \ve{T} \end{bmatrix}$ to the state vector $(\ve{\xi}, \ve{\eta})$ will yield a representation in Byrnes-Isidori normal form where the zero dynamics are characterized via $\ve{T}^{-1 } \ve{Q} \ve{T} = \ve{\tilde{Q}}$. Hence, the zero dynamics of the system $\mathcal{B}$ are stable if, and only if, $\ve{\tilde{Q}}$ is stable.
\end{proof}

We can finally characterize the stability of the zero dynamics of $\mathcal{B}$ by using a persistently exciting input data sequence.

\begin{theorem} \label{ZeroDynamicsStability}
Let $\mathcal{B} \in \mathcal{L}^{2m(\mathcal{B})}$ with $n(\mathcal{B}) > 0$ be a controllable system with vector relative degree $\ve{r}(\Bc) = (r_1,\ldots,r_{m(\mathcal{B})}) \in \N_0^{m(\mathcal{B})}$, and let $\mathcal{D} = (\ve{u_d},\ve{y_d})\in\Bc_{[0,T]}$ a sampled data sequence such that $\ve{u_d}$ is persistently exciting of order $L + n(\mathcal{B})$ for some $L \geq l(\mathcal{B}) + \max_{1 \leq j \leq m(\mathcal{B})} r_j + 1$. Construct $\ve{\tilde{Q}}$ as in Lemma \ref{xConvergence}. Then the zero dynamics of $\mathcal{B} $ are stable if, and only if, $\ve{\tilde{Q}}$ is a stable matrix.
\end{theorem}
\begin{proof}
First, assume that $\ve{\tilde{Q}}$ is stable, and fix an arbitrary minimal state space representation of $\Bc$. If $n(\mathcal{B}) > 0$, consider the function $\ve{\varphi} : \mathcal{B}_{ZD,[0,l(\mathcal{B}) - 1]} \rightarrow \R^{n(\mathcal{B})},(\ve{u}(0), \hdots, \ve{u}(l(\mathcal{B}) -1)) \mapsto \ve{x}(l(\mathcal{B}))$, where $\ve{x}(l(\mathcal{B}))$ is the state of the fixed minimal representation corresponding to the initializing input sequence $(\ve{u}(0), \hdots, \ve{u}(l(\mathcal{B}) -1))$. Then $\ve{\varphi}$ is well-defined by Lemma~\ref{InitialCondition}, and linear by the linearity of the system. By Lemma~\ref{xConvergence}, each $\ve{u} \in \mathcal{B}_{ZD}$ fulfills $\ve{\varphi}(\ve{u}(i-l(\mathcal{B})), \hdots, \ve{u}(i-1)) = \ve{x}(i) \rightarrow \ve{0}$ for $i\to \infty$. Due to linearity of $\ve{\varphi}$, we can decompose $\ve{u} = \ve{u_{0}}+ \ve{u_{\ker}}$, where $(\ve{u_0}(i-l(\mathcal{B})), \hdots, \ve{u_0}(i-1)) \rightarrow \ve{0}$, and $\ve{\varphi}(\ve{u_{\ker}}(i-l(\mathcal{B})) , \hdots, \ve{u_{\ker}} (i-1)) = 0$ for all $i \geq l(\mathcal{B})$. Since $\ve{u_{\ker}}$ initializes the state to be $\ve{x}(i) = 0$ for all $i \geq l(\mathcal{B})$, and all $r_i$ are assumed to be finite and therefore $r_i \leq l(\mathcal{B})$ by applying Lemma~\ref{rLeqL} to all induced SISO systems, we can conclude $\ve{u_{\ker}}(i+l(\mathcal{B})) = 0$, for all $i \geq 0$. This shows that $\ve{u}(i)=\ve{u_0}(i) \rightarrow \ve{0}$, thus the zero dynamics of $\mathcal{B}$ are stable. 

If we assume that the zero dynamics of $\mathcal{B}$ are stable, then, for any minimal representation of $\Bc$, $\ve{x}(i) \rightarrow \ve{0}$ for $i \rightarrow \infty$ for any possible sequence of states corresponding to any $\ve{u}\in\mathcal{B}_{ZD}$, hence $\ve{\tilde{Q}}$ is stable by Lemma~\ref{xConvergence}.
\end{proof}

\subsection{Stability of zero dynamics and data informativity}

Theorem~\ref{ZeroDynamicsStability} allows us to decide on the stability of the zero dynamics of the system $\mathcal{B}$ by using a data sequence that is sufficiently informative to uniquely identify the system. In this section, we derive a result that allows us to determine the stability of the zero dynamics under weaker conditions using the data informativity framework. The set of all systems $\Bc\in\Lc^w$ with stable zero dynamics is denoted by $\Sigma_{\text{stable}}$. Recall that a data set $\mathcal{D}$ is called informative for stability of the zero dynamics, if $\Sigma_{\mathcal{D}} \subseteq \Sigma_{\text{stable}}$, and it is called informative for instability of the zero dynamics, if $\Sigma_{\mathcal{D}} \cap \Sigma_{\text{stable}} = \emptyset$.

\begin{lemma} \label{StabilityMPUM}
Let $\mathcal{D} = (\ve{u_d}, \ve{y_d})\in\Bc_{[0,T-1]}$ be a data sequence with $T \geq l(\mathcal{B}) + \max_{1 \leq i \leq p(\mathcal{B})} r_i + 1$, sampled from an unknown square system $\mathcal{B} \in \mathcal{L}^{2m(\mathcal{B})}$ with $n(\mathcal{B})> 0$, vector relative degree $\ve{r}(\Bc) = (r_1, \hdots, r_{p(\mathcal{B})}) \in \N_0^{p(\mathcal{B})}$, and known lag $l(\mathcal{B})$. Then there exists a matrix $\ve{\tilde{Q}}\in\R^{d\times d}, d\in\N_0$, uniquely determined by the data sequence $\mathcal{D}$, such that the zero dynamics of $\mathcal{B}_{MPUM}$ are stable if, and only if, $\ve{\tilde{Q}}$ is stable.
\end{lemma}
\begin{proof}
Using Lemma \ref{ExtendedSequencesForm}, we can explicitly determine $N \coloneqq \mathcal{B}_{MPUM,[0,2l(\mathcal{B})]}$ from the data $\mathcal{D}$. Since the vector relative degree of $\Bc$ exists by assumption, we can conclude that $l(\mathcal{B}) \geq r_i$ for all $1 \leq i \leq p(\mathcal{B})$ by applying Lemma~\ref{rLeqL} to all induced SISO systems. Observe that $$ \ve{\pi_{l(\mathcal{B}) + 1}} \{ \ve{u} : (\ve{u},\ve{0}) \in N \} =: N_{l(\mathcal{B}) + 1}$$
contains all possible input sequences of length $l(\mathcal{B}) + 1$ of $\mathcal{B}_{MPUM, ZD}$. Moreover, $N_{l(\Bc)+1} = \ve{\pi_{l(\Bc)+1}}M$ for $M$ as defined in~\eqref{DefineSubspaceM} with $L = 2l(\Bc)+1$. Let $d:= \dim \ve{\pi_{l(\mathcal{B})}} N_{l(\mathcal{B}) + 1}$ and $\left\{\ve{v_1}, \hdots, \ve{v_d} \right\}$ be a basis of $\ve{\pi_{l(\mathcal{B})}} N_{l(\mathcal{B}) + 1}$. As we have seen in Lemma~\ref{UniqueContinuationInZD}, the unique continuation to $(\ve{v_k}, \ve{u_k}) \in N_{l(\mathcal{B}) + 1}$ will allow us to determine the coefficients $\tilde{q}_{k,i} \in \R$ that satisfy $(\ve{v_{k,2}}, \hdots, \ve{v_{k,l(\mathcal{B})}}, \ve{u_k})  = \sum_{i=1}^{d} \tilde{q}_{k,i} \ve{v_i} $. Then $\mathcal{B}_{MPUM,ZD} $ is stable if, and only if, the matrix $\ve{\tilde{Q}} = (\tilde{q}_{k,i})_{1 \leq i,k \leq d}$ is stable, as shown in Theorem~\ref{ZeroDynamicsStability}.
\end{proof}

\begin{lemma} \label{DetermineUnstableZD}
Let $\mathcal{B} \in \mathcal{L}^w$ be a system and $\mathcal{D} = (\ve{u_d}, \ve{y_d})\in\Bc_{[0,T-1]}$ be a sampled data sequence. 
If the zero dynamics of $\mathcal{B}_{MPUM}$ are not stable, then $\mathcal{D} $ is informative for instability of the zero dynamics.
\end{lemma}

\begin{proof}
By definition of the MPUM, we have $\mathcal{B}_{MPUM} \subseteq \mathcal{B}'$ for all $\Bc'\in\Sigma_\Dc$.  The assumption of the zero dynamics of $\mathcal{B}_{MPUM}$ not being stable implies that there is a trajectory $(\ve{u},\ve{0}) \in \mathcal{B}_{MPUM}$ such that $(\ve{u}(k),\ve{x}(k)) \nrightarrow \ve{0}$ for a corresponding state sequence $\ve{x}\in\left(\R^{n(\mathcal{B}_{MPUM})}\right)^{\N_0}$ of a minimal representation of $\mathcal{B}_{MPUM}$. Now fix $\Bc'\in\Sigma_\Dc$  and choose a minimal representation of $\mathcal{B}'$. Since $\mathcal{B}_{MPUM} \subseteq \mathcal{B}'$ and $\ve{u}\in \mathcal{B}'_{ZD}$, there exists a state sequence $\ve{x}'\in\left(\R^{n(\mathcal{B}')}\right)^{\N_0}$ and a matrix $T\in\R^{n(\Bc')\times n(\mathcal{B}_{MPUM})}$ such that $\ve{x}'(k) = T \ve{x}(k)$ for all $k\in\N_0$. Hence $(\ve{u}(k),\ve{x}'(k)) \nrightarrow \ve{0}$ and we can conclude that $\mathcal{B}'$ does not have stable zero dynamics.
\end{proof}

By Lemma~\ref{DetermineUnstableZD}, unstable zero dynamics of the most powerful unfalsified model implies that the zero dynamics of any system that explains the data cannot be stable. In the following we determine further conditions on the data that ensure that the stability of the zero dynamics of the MPUM yields stability of the zero dynamics of all systems that explain the data.

\begin{lemma} \label{EquivalentCharacterizationStateDimension}
Let $\mathcal{B} \in \mathcal{L}^{w}$ be an unknown system  with known lag $l(\mathcal{B})$ and known McMillan degree $n(\mathcal{B}) > 0$, and let $\mathcal{D} = (\ve{u_d}, \ve{y_d})\in\Bc_{[0,T-1]}$ be a sampled data sequence of length $T \geq 2l(\mathcal{B}) + 1 $. Then 
$$n(\mathcal{B}) = \dim \left( \mathcal{H}_{l(\mathcal{B}) + 1}(\ve{y_d}) \ker(\mathcal{H}_{l(\mathcal{B}) + 1} (\ve{u_d}) ) \right) $$ 
if, and only if, $n(\mathcal{B}) = n(\mathcal{B}_{MPUM})$.
\end{lemma}

\begin{proof}
Assume that the first equation holds. The dimension of the state space of a minimal representation of $\mathcal{B}_{MPUM}$ is non-zero, since 
$$\dim \left( \mathcal{H}_{l(\mathcal{B}) + 1}(\ve{y_d}) \ker \mathcal{H}_{l(\mathcal{B}) + 1} (\ve{u_d})  \right) = n(\mathcal{B}) > 0$$
together with Lemma~\ref{DetermineMPUM} implies that there is $(\ve{u},\ve{y}) \in \mathcal{B}_{MPUM}$ with $\ve{u}(i) = \ve{0}$ and $\ve{y}(i) \neq\ve{ 0}$ for some $i \in \N_0$, hence the behavior $\mathcal{B}_{MPUM}$ cannot be represented by $\ve{y}(t) = \ve{Du}(t)$, $t\in\N_0$, for some matrix $\ve{D}$. In case the second equation holds, $n(\mathcal{B}_{MPUM}) = n(\mathcal{B}) > 0$. This shows that the McMillan degree of $\mathcal{B}_{MPUM}$ is non-trivial, as soon as either one of the two equations hold.

Now choose a minimal representation $(\ve{A}, \ve{B}, \ve{C}, \ve{D})$  of $\mathcal{B}_{MPUM}$. By the definition of the lag, we have $\rk \mathcal{O}_{l(\mathcal{B}_{MPUM})} (\ve{C},\ve{A}) = n(\mathcal{B}_{MPUM})$, thus $\mathcal{O}_{l(\mathcal{B}_{MPUM}) }(\ve{C},\ve{A}) $ is left invertible. This implies that for any two distinct initial states $\ve{x^0} \neq \ve{\tilde{x}^0} \in \R^{n(\mathcal{B}_{MPUM})}$,  we have $\ve{x^0} - \ve{\tilde{x}^0}\notin\ker \mathcal{O}_{l(\mathcal{B}_{MPUM}) } (\ve{C},\ve{A}) = \{ \ve{0} \}$, hence the zero input responses differ, $(\ve{CA}^k\ve{x^0})_{0 \leq k \leq l(\mathcal{B}_{MPUM})} \neq (\ve{CA}^{k}\ve{\tilde{x}^0})_{0 \leq k \leq l(\mathcal{B}_{MPUM})}$. 
It follows from Lemma~\ref{DetermineMPUM} in combination with Proposition~\ref{CharakterisierungenLag} that $l(\mathcal{B}_{MPUM}) \leq l(\mathcal{B})$, thus $(\ve{CA}^k\ve{x^0})_{0 \leq k \leq l(\mathcal{B})} \neq (\ve{CA}^{k}\ve{\tilde{x}^0})_{0 \leq k \leq l(\mathcal{B})}$. We can thus define a linear map $\ve{\psi}: \R^{n(\mathcal{B}_{MPUM})} \rightarrow \left( \R^{p(\mathcal{B})} \right)^{l(\mathcal{B}) + 1}$, $\ve{x^0} \mapsto(\ve{CA}^k\ve{x^0})_{0 \leq k \leq l(\mathcal{B})}$ that is injective. Therefore, 
\begin{align*}
n(\mathcal{B}_{MPUM}) &= \dim \im \ve{\psi} =   \dim \left\{ (\ve{0}, \ve{y}) \in \mathcal{B}_{MPUM,[0,l(\mathcal{B})]} \right\} \\
&= \dim \left( \mathcal{H}_{l(\mathcal{B}) + 1}(\ve{y_d}) \ker \mathcal{H}_{l(\mathcal{B}) + 1} (\ve{u_d}) \right),
\end{align*}
where the last equality follows with Lemma \ref{DetermineMPUM}. This shows the claimed equivalence.
\end{proof}
  
We will refer to the property that any system $\mathcal{B} \in \Sigma_\Dc$ has a well-defined vector relative degree as $\mathcal{D}$ being informative for the existence of a vector relative degree. Furthermore, for a system $\Bc\in\Lc^w$ which has vector relative degree $\ve{r}(\Bc) = (r_1, \hdots, r_{p(\mathcal{B})}) \in \N_0^{p(\mathcal{B})}$, we denote by $r_s(\Bc) = \sum_{i=1}^{p(\mathcal{B})} r_i$ the vector relative degree sum of $\Bc$. For fixed $s\in\N_0$, the class of all systems with vector relative degree sum~$s$ is denoted by $\Sigma^{\rm sum}_s$. The data $\Dc$ are called informative for vector relative degree sum~$s$, if $\Sigma_\Dc\subseteq\Sigma^{\rm sum}_s$.

\begin{theorem} \label{StrictCharacterizationZeroDynamics}
Let $\mathcal{B} \in \mathcal{L}^{2m(\mathcal{B})}$ be an unknown square system with known lag $l(\mathcal{B})$, known McMillan degree $n(\mathcal{B}) > 0$, and known vector relative degree sum~$r_s(\Bc)$. Further, let $\mathcal{D} = (\ve{u_d}, \ve{y_d})\in\Bc_{[0,T-1]}$ be a sampled data sequence of length $T \geq 2l(\mathcal{B})+ 1$. Then $\mathcal{D}$ is informative for the existence of the vector relative degree and informative for the stability of the zero dynamics if, and only if, 
\begin{enumerate}[label = (\alph*)]
\item \label{SameDimensionN} $n(\mathcal{B}) = \dim \left( \mathcal{H}_{l(\mathcal{B}) + 1}(\ve{y_d}) \ker(\mathcal{H}_{l(\mathcal{B}) + 1} (\ve{u_d}) ) \right)$, 
\item \label{InformativeOfRelativeDegree} $\mathcal{D}$ is informative for vector relative degree sum~$r_s(\Bc)$, and 
\item \label{RequireMPUMStable} the zero dynamics of $\mathcal{B}_{MPUM}$ are stable. 
\end{enumerate}
\end{theorem}

\begin{proof}
\textbf{Step 1: }
We show that the conditions~\ref{SameDimensionN}, \ref{InformativeOfRelativeDegree} and~\ref{RequireMPUMStable} are sufficient to conclude the informativity for the stability of the zero dynamics and the informativity for the existence of the vector relative degree. By condition~\ref{InformativeOfRelativeDegree} we have $r_s(\tilde\Bc) = r_s(\Bc)$ for all $\tilde{\mathcal{B}} \in \Sigma_{\mathcal{D}}$. Therefore, it is obvious that $\mathcal{D}$ is informative for the existence of a vector relative degree. By the Byrnes-Isidori normal form~\eqref{eq:BIF}, a minimal representation of the zero dynamics of $\tilde{\mathcal{B}}$ has a $k(\tilde{\mathcal{B}})\coloneqq \big(n(\tilde\Bc) - r_s(\tilde \Bc)\big)$-dimensional vector space $V_{\tilde\Bc}$ of initial conditions. This dimension $k:=k(\tilde\Bc)$ is the same for any $\tilde{\mathcal{B}} \in \Sigma_{\mathcal{D}}$ by condition~\ref{InformativeOfRelativeDegree} and the fact that
\[
    n(\mathcal{B}_{MPUM}) \le n(\tilde\Bc)  \le n(\mathcal{B}) 
\]
and $n(\mathcal{B}) = n(\mathcal{B}_{MPUM})$ by condition~\ref{SameDimensionN} and Lemma~\ref{EquivalentCharacterizationStateDimension}.

We can conclude that the initial conditions of trajectories in $\mathcal{B}_{MPUM,ZD}$ with respect to any minimal representation of $\mathcal{B}_{MPUM}$ form a $k$-dimensional vector space $V$. For any basis of initial values $\{ \ve{x_i}(0) : 1 \leq i \leq k \}$ of $V$, let $\ve{u_i} \in \mathcal{B}_{MPUM,ZD}$ be an input sequence that needs to be applied to the initial value $\ve{x_i}(0)$ to cause a constant zero output, for each $1 \leq i \leq k$. The sequences $\ve{u_i}$ are uniquely determined, as we can conclude from Lemma~\ref{UniqueContinuationInZD}, and the family $\{ \ve{u_i} : 1 \leq i \leq k \} \subseteq \mathcal{B}_{MPUM,ZD}$ is linearly independent by the observability of the minimal representation, which follows from Lemma~\ref{minimalObservable}.

Fix an arbitrary $\tilde{\mathcal{B}} \in \Sigma_{\mathcal{D}}$. Since $\ve{u_i} \in \mathcal{B}_{MPUM,ZD} \subseteq \tilde{\mathcal{B}}_{ZD}$, there are initial conditions $\ve{\tilde{x}_i}(0) \in V_{\tilde\Bc}$ for the zero dynamics of system $\tilde{\mathcal{B}}$ with corresponding input sequence $\ve{u_i}$, for each $1 \leq i \leq k$. The family $\{  \ve{\tilde{x}_i}(0) : 1 \leq i \leq k \}$ is a basis of $V_{\tilde\Bc}$, since  $\dim V = k = \dim V_{\tilde\Bc}$ and any non-trivial linear combination of zero means that the same linear combination of the corresponding input sequences from  $\{ \ve{u_i} : 1 \leq i \leq k \}$ is zero, which contradicts its linear independence. Define the linear map $\ve{T}: V \rightarrow V_{\tilde\Bc}$ such that $\ve{T}(\ve{x_i}(0)) = \ve{\tilde{x}_i}(0) $ for each $1 \leq i \leq k$, and observe that~$\ve{T}$ is well defined and bijective, as it is a basis transformation. 

Now, seeking a contradiction, assume that the zero dynamics of $\tilde{\mathcal{B}} \in \Sigma_{\mathcal{D}}$ are not stable, and choose an initial value $\ve{\tilde{x}_0} \in V_{\tilde\Bc}$ and the respective input sequence $\ve{u} \in \tilde{\mathcal{B}}_{ZD}$ such that the resulting trajectory does not converge to zero, $(\ve{u},\ve{\tilde x}) \nrightarrow \ve{0}$. Let $\ve{x_0} \coloneqq \ve{T}^{-1}(\ve{\tilde{x}_0})$ be the respective initial value for the system $\mathcal{B}_{MPUM}$. Then $(\ve{u},\ve{0}) \in \mathcal{B}_{MPUM}$, and by construction, the state sequence along this trajectory is $\ve{{x}}(i)=\ve{T}^{-1}(\ve{\tilde{x}}(i))$, $i\in\N_0$, and since $\ve{T}$ is bijective, $(\ve{u},\ve{x}) \nrightarrow \ve{0}$, which contradicts condition~\ref{RequireMPUMStable}. We can conclude that any $\tilde{\mathcal{B}} \in \Sigma_{\mathcal{D}}$ has stable zero dynamics, hence $\mathcal{D}$ is informative for the stability of the zero dynamics.

\textbf{Step 2: }
Assume that the data $\mathcal{D}$ is informative for the stability of the zero dynamics and informative for the existence of a vector relative degree. Thus, as a consequence of Lemma~\ref{DetermineUnstableZD}, condition \ref{RequireMPUMStable} holds and, furthermore, the vector relative degree of $\mathcal{B}_{MPUM}$ exists. In particular, its decoupling matrix $\ve{G}\in\R^{m(\Bc)\times m(\Bc)}$ (see Definition~\ref{RelativeDegreeMIMO}) is invertible and for any minimal state space representation of $\mathcal{B}_{MPUM}$ there exists a Byrnes-Isidori normal form (see~\eqref{eq:BIF}). \\
\textbf{Step 2a:} 
We show condition \ref{SameDimensionN}. We have $n(\mathcal{B}_{MPUM}) \leq n(\mathcal{B})$, since $\mathcal{B}_{MPUM} \subseteq \mathcal{B}$. In order to show equality, seeking a contradiction, we 
assume that $n(\mathcal{B}_{MPUM}) < n(\mathcal{B})$. We aim to construct a behavior $\mathcal{B}_{\Lambda} \in \Sigma_{\mathcal{D}}$ with additional states $\ve{z}$, such that $n(\mathcal{B}_{\Lambda}) > n(\mathcal{B}_{MPUM})$ and the zero dynamics $\mathcal{B}_{\Lambda,ZD}$ are not stable. To this end, let $\Lambda = \text{diag}(\lambda_1, \hdots, \lambda_{n(\mathcal{B}) - n(\mathcal{B}_{MPUM})}) \in \R^{(n(\mathcal{B}) - n(\mathcal{B}_{MPUM})) \times (n(\mathcal{B}) - n(\mathcal{B}_{MPUM}))}$ be a diagonal matrix with $\lambda_i \in \R$ and $|\lambda_1|\ge 1$, and define $\mathcal{B}_{\Lambda}$ by the following state-space representation for $1 \leq i \leq m(\mathcal{B})$:
\begin{align*}
\xi_{i,j}(t+1) &= \xi_{i,j+1} (t)\qquad\qquad \text{ for all }  1 \leq j \leq r_i - 1 \\
\xi_{i,r_i}(t+1) &= \ve{\beta_i} \begin{pmatrix} \ve{\xi} (t) \\ \ve{\eta} (t) \end{pmatrix} +  \ve{G_i u}(t) + \ve{\alpha}^{\top} \ve{z}(t)\\
\ve{\eta}(t+1) &= \ve{Q \eta} (t)+ \ve{P y}(t)\\
\ve{z}(t+1) &= \ve{\lambda z}(t)\\
y_i(t) &=  \xi_{i,1}(t),
\end{align*}
where $\ve{\xi} =(\ve{\xi_1}, \hdots, \ve{\xi_{m(\mathcal{B})}})^{\top} $ with $\ve{\xi_i} = (\xi_{i,1}, \hdots, \xi_{i,r_i})^\top$, and
\begin{itemize}
    \item $\ve{G_i}$ are the rows of the decoupling matrix $\ve{G}\in\R^{m(\Bc)\times m(\Bc)}$ of $\mathcal{B}_{MPUM}$,
    \item $\ve{\beta_i} \in \R^{1 \times n(\mathcal{B}_{MPUM})}$, $\ve{Q} \in \R^{\left(n(\mathcal{B}_{MPUM}) - r \right) \times \left( n(\mathcal{B}_{MPUM}) - r\right)}$  and $\ve{P} \in \R^{\left( n(\mathcal{B}_{MPUM}) - r \right) \times m(\mathcal{B})}$ are chosen such
    that they parametrize a Byrnes-Isidori normal form of a minimal representation of  $\mathcal{B}_{MPUM}$, $r = r_s(\mathcal{B}_{MPUM})$,    and
    \item $\ve{\alpha} \in \R^{n(\mathcal{B}) - n(\mathcal{B}_{MPUM})}$ with $\alpha_k \neq 0$ for all $1 \leq k \leq n(\mathcal{B}) - n(\mathcal{B}_{MPUM})$.
\end{itemize}
The behavior $\mathcal{B}_{\Lambda}$ may be interpreted as an extension of a Byrnes-Isidori normal form of $\mathcal{B}_{MPUM}$ onto a higher dimensional state space. Observe that $\mathcal{B}_{\Lambda} \in \mathcal{L}^{2m(\Bc)}$, and that $\mathcal{B}_{MPUM} \subseteq \mathcal{B}_{\Lambda}$ by choosing the initial value $\ve{z}(0) =\ve{ 0}$ for the additional coordinates, therefore $\mathcal{B}_{\Lambda } \in \Sigma_{\mathcal{D}}$. However, since $\alpha_1 \neq 0$ and $\vert \lambda_1 \vert \geq 1$, the behavior $\mathcal{B}_{\Lambda}$ does not have stable zero dynamics, because by choosing the initial values $\ve{\eta} (0) = \ve{0}$, $\ve{\xi}(0) = \ve{0}$ and $\ve{z}(0) = (1 \ 0 \  \hdots \ 0)^{\top} \in \R^{n(\mathcal{B}) - n(\mathcal{B}_{MPUM})}$, we can find the unique input sequence $\ve{u} \in \mathcal{B}_{\Lambda, ZD}$ that will cause a constant zero output sequence when applied to the system under the those initial conditions. This sequence satisfies 
$$\forall\, i=1,\ldots,m(\Bc):\ \ve{G_i}\ve{u}(t) = -\alpha_1 \ve{z_1}(t) = -\alpha_1 \lambda_1^t,$$
hence
\[
    \ve{u}(t) = -\alpha_1 \lambda_1^t \ve{G}^{-1} \begin{pmatrix} 1\\ \vdots\\ 1\end{pmatrix} \to \infty,\quad t\to\infty,
\]
since $\ve{G}$ is invertible, $\alpha_1 \neq 0$ and $\vert \lambda_1 \vert \geq 1$. Hence, $\mathcal{B}_{\Lambda} \in \Sigma_{\mathcal{D}}$ does not have stable zero dynamics, which contradicts the assumption that $\mathcal{D}$ is informative for the stability of the zero dynamics. Thus, $n(\mathcal{B}_{MPUM}) = n(\mathcal{B})$ and then Lemma~\ref{EquivalentCharacterizationStateDimension} implies condition~\ref{SameDimensionN}.\\
\textbf{Step 2b:}
We show condition \ref{InformativeOfRelativeDegree}. Without loss of generality we assume that the entries of the vector relative degree of $\mathcal{B}_{MPUM}$ are ordered in ascending order, $r_1(\mathcal{B}_{MPUM}) \leq \hdots \leq r_{m(\mathcal{B})}(\mathcal{B}_{MPUM})$, which can be achieved by permuting the output coordinates. Define a transformation of the input coordinates $\R^{m(\mathcal{B})}\ni \ve{u} \mapsto \ve{G}^{-1}\ve{u} =:\ve{v}$. Let $(\ve{A},\ve{B},\ve{C},\ve{D})$ be a minimal state space representation of $\mathcal{B}_{MPUM}$ with respect to the original input coordinates. Then $(\ve{A}, \ve{BG}^{-1}, \ve{C}, \ve{DG}^{-1})$ is a minimal representation of it with respect to the input coordinates~$\ve{v}$. We can determine the decoupling matrix $\ve{G_v}$ of this transformed system  via 
$$(\ve{G_v})_i = \left. \begin{cases} \ve{D_iG}^{-1}, & \text{ if } r_i(\mathcal{B}_{MPUM}) = 0 \\ \ve{C_iA}^{r_i(\mathcal{B}_{MPUM}) - 1}\ve{BG}^{-1},& \text{ if } r_i(\mathcal{B}_{MPUM}) > 0 \end{cases} \right\} = \ve{G}_i \ve{G}^{-1},\quad 1\le i\le m(\Bc),$$
hence the decoupling matrix is the identity matrix, $\ve{G_v} = \ve{I_{m(\mathcal{B})}}$.

Now let $\tilde{\mathcal{B}} \in \Sigma_{\mathcal{D}}$ be arbitrary. By assumption, $\tilde \Bc$ has a well-defined vector relative degree $\ve{r}(\tilde{\mathcal{B}}) = (r_1(\tilde{\mathcal{B}}), \hdots, r_{m(\mathcal{B})}(\tilde{\mathcal{B}}))$. 
Denoting with $r_{kj} (\tilde{\mathcal{B}})$ the relative degree of the induced SISO system $\tilde{\mathcal{B}}^{k,j}$, we have that $r_{ij}(\mathcal{B}_{MPUM}) \in \{ r_{ij}(\tilde{\mathcal{B}}), \infty \}$. This holds since $r_{ij}(\mathcal{B}_{MPUM})  < \infty$ implies by Lemma~\ref{InformativityMPUM} that there is $a \neq 0$ such that $\ve{v}(a,r_{ij}(\mathcal{B}_{MPUM})) \in \mathcal{B}_{MPUM,[0,l(\mathcal{B})+r_{ij}(\mathcal{B}_{MPUM})]}^{i,j}$. This trajectory is also contained in $\tilde{\mathcal{B}}_{[0,l(\mathcal{B}) + r_{ij}(\mathcal{B}_{MPUM})]}^{i,j}$, yielding equality of the two relative degrees $r_{ij}(\mathcal{B}_{MPUM}) = r_{ij}(\tilde{\mathcal{B}})$. Moreover, this shows that 
\begin{equation} \label{RDMPUMgeqTrue} 
r_i(\mathcal{B}_{MPUM}) = \min_{1 \leq j \leq m(\mathcal{B})} r_{ij}(\mathcal{B}_{MPUM}) \geq \min_{1 \leq j \leq m(\mathcal{B})}  r_{ij}(\tilde{\mathcal{B}})  = r_i(\tilde{\mathcal{B}})
\end{equation} 
for any $1 \leq i \leq m(\mathcal{B})$.

We aim to show that  $r_s(\tilde{\mathcal{B}}) =r_s({\mathcal{B}})$. Since $\Bc\in\Sigma_\Dc$, it suffices to show that $r_s(\tilde{\mathcal{B}}) =r_s(\mathcal{B}_{MPUM})$. Seeking a contradiction, assume that $r_s(\tilde{\mathcal{B}}) \neq r_s(\mathcal{B}_{MPUM})$ and let $1 \leq i \leq m(\mathcal{B})$ be the maximal index such that $r_i(\tilde{\mathcal{B}}) \neq r_i(\mathcal{B}_{MPUM})$. Using the transformation to the input coordinates~$\ve{v}$, we investigate the structure of the decoupling matrix $\ve{G_v}(\tilde{\mathcal{B}})$. 

Choose arbitrary indices $1 \leq k \leq i$ and $i \leq j \leq m(\mathcal{B})$. Seeking a contradiction, assume that $\ve{G_v}(\tilde{\mathcal{B}})_{kj} \neq 0$. Then we have
\begin{align} \label{RelativeDegreeInequality}
r_{kj} (\tilde{\mathcal{B}}) &= r_k(\tilde{\mathcal{B}}) \stackrel{\eqref{RDMPUMgeqTrue}}{\leq} r_k(\mathcal{B}_{MPUM}) \leq r_j(\mathcal{B}_{MPUM}) = r_{jj}(\mathcal{B}_{MPUM}). 
\end{align}
If $j = i = k$, then we have $r_i(\tilde{\mathcal{B}} ) < r_i(\mathcal{B}_{MPUM}) = r_{ii}(\mathcal{B}_{MPUM})$ by choice of the index $i$ and equation~\eqref{RDMPUMgeqTrue}. This yields $r_{ii}(\tilde{\mathcal{B}}) = r_{ii}(\mathcal{B}_{MPUM}) > r_i(\tilde{\mathcal{B}})$, and therefore $\ve{G_v(}\tilde{\mathcal{B}})_{ii} = 0$. In all other cases, $k < j$, and $r_{kj}(\mathcal{B}_{MPUM}) = \infty$. 

Note that $\mathcal{B}_{MPUM}$ is completely determined by the data $\Dc$ due to Lemma~\ref{DetermineMPUM}, hence all relative degrees $r_{kj}(\mathcal{B}_{MPUM})$ of the induced SISO systems and the vector relative degree of $\mathcal{B}_{MPUM}$ is determined by the data and can be assumed to be known. Then, by Lemma~\ref{InformativityMPUM}, there exists $a \neq 0$ such that $(\ve{0}_{l(\mathcal{B}) - r_{jj}(\mathcal{B}_{MPUM})}, 1, \hdots; \ve{0}_{l(\mathcal{B})}, a) \in \mathcal{B}^{j,j}_{MPUM,[0,l(\mathcal{B})]}$. Due to the characterization of the most powerful unfalsified model obtained in Lemma \ref{DetermineMPUM}, the input vector $\ve{u} \coloneqq \begin{pmatrix} \ve{0}_{l(\mathcal{B}) - r_{jj}(\mathcal{B}_{MPUM})} & 1 & \hdots \end{pmatrix}^{\top} \in \im \mathcal{H}_{l(\mathcal{B}) + 1} (\ve{u_d^j})$, where $\ve{u_d^j}$ refers to the sequence of the $j$th entries of $\ve{u_d}$. We can conclude that $\begin{pmatrix} \ve{u} \\ \ve{y} \end{pmatrix} \in \im \begin{bmatrix} \mathcal{H}_{l(\mathcal{B}) +  1}(\ve{u_d^j}) \\ \mathcal{H}_{l(\mathcal{B}) + 1}(\ve{y_d^k}) \end{bmatrix} = \mathcal{B}_{MPUM,[0,l(\mathcal{B})]}^{k,j}$ for some vector $\ve{y} \in (\R^{m(\mathcal{B})})^{l(\mathcal{B}) + 1}$. Let $\ve{x}(0) \in \R^{n(\mathcal{B})}$ be the unique initial state of the representation $(\ve{A},\ve{B},\ve{C},\ve{D})$ corresponding to the trajectory $(\ve{u},\ve{y}) \in \mathcal{B}_{MPUM,[0,l(\mathcal{B})]}^{k,j}$. Due to condition~\ref{SameDimensionN}, 
we can observe $n(\mathcal{B})$ linearly independent zero input responses, implying that the zero input response $(\ve{CA}^i\ve{x}(0))_{0 \leq i \leq l(\mathcal{B})}$ is contained in $\mathcal{H}_{l(\mathcal{B}) + 1} (\ve{y_d}) \ker (\mathcal{H}_{l(\mathcal{B}) + 1} (\ve{u_d})).$ Considering only the $k$th output coordinate, this shows that  $\left(\ve{0}_{l(\mathcal{B}) + 1}; (\ve{C_kA}^{p} \ve{x}(0))_{0 \leq p \leq l(\mathcal{B})}\right) \in \mathcal{B}_{MPUM,[0,l(\mathcal{B})]}^{k,j}$. Moreover, since $\mathcal{B}_{MPUM}^{k,j}$ is a vector space, the difference of these trajectories is also contained in the behavior,
$$ \begin{pmatrix} \ve{u} \\ \ve{\tilde{y}} \end{pmatrix} = \begin{pmatrix} \ve{u} \\ \ve{y} - (\ve{C_kA}^{p} \ve{x}(0))_{0 \leq p \leq l(\mathcal{B})} \end{pmatrix} \in \mathcal{B}_{MPUM,[0,l(\mathcal{B})]}^{k,j},$$
and has initial state $\ve{\tilde{x}}(0) = \ve{0}$ by construction. This shows that a suitable part of this output sequence, $\ve{\tilde{y}}_{l(\mathcal{B}) - r_{jj}(\mathcal{B}_{MPUM}) + 2, \hdots, l(\mathcal{B})+1}$, contains the first $r_{jj}(\mathcal{B}_{MPUM})$ samples 
of the impulse response of $\mathcal{B}_{MPUM}^{k,j}$. Since this trajectory is also contained in $\tilde{\mathcal{B}}^{k,j}$ by the definition of $\mathcal{B}_{MPUM}$, these are also the first $r_{jj}(\mathcal{B}_{MPUM})$ samples of the impulse response of $\tilde{\mathcal{B}}^{k,j}$.
Since we have seen in (\ref{RelativeDegreeInequality}) that $r_{kj}(\tilde{\mathcal{B}})  \leq r_{jj}(\mathcal{B}_{MPUM})$, not all of these entries are zero, implying that $v(a,r_{kj}(\tilde{\Bc})) \in \Bc_{MPUM}^{k,j}$ for some $a \neq 0$. 
Then Lemma~\ref{InformativityMPUM} implies that $\mathcal{D}$ is informative for the $(k,j)$ induced SISO system (of any system explaining the data) having relative degree $r_{kj}(\tilde{\Bc}) < \infty$, contradicting $r_{kj}(\mathcal{B}_{MPUM}) = \infty$, which was found before. 

This implies that $\ve{G_v}(\tilde{\mathcal{B}})_{kj} = 0$ for all $1 \leq k \leq i$ and $i \leq j \leq m(\mathcal{B})$, and hence the first $i$ rows of $\ve{G_v}(\tilde{\mathcal{B}})$ are linearly dependent. Since the matrix is square, it does not have full rank, and $\tilde{\mathcal{B}} \in \Sigma_{\mathcal{D}}$ does not have a well-defined vector relative degree. This contradicts the premise, hence $r_s(\tilde{\mathcal{B}}) =r_s({\mathcal{B}})$ for all $\tilde\Bc\in\Sigma_\Dc$, i.e., condition~\ref{InformativeOfRelativeDegree} holds, which completes the proof.
\end{proof}

\begin{remark}
Observe that equation~\eqref{RDMPUMgeqTrue} shows that the informativity for the vector relative degree sum $r_s(\mathcal{B})$ and the informativity for the vector relative degree $\ve{r}(\mathcal{B})$ are equivalent conditions. However, note that verification of both conditions imposes different difficulties. For informativity for the vector relative degree sum not each entry of the vector $\ve{r}(\mathcal{B})$ must be known, but only their sum $r_s(\mathcal{B})$.
\end{remark}

Theorem~\ref{StrictCharacterizationZeroDynamics} provides conditions under which a data set $\mathcal{D}$ is informative for the stability of the zero dynamics, while Lemmas~\ref{StabilityMPUM} and~\ref{DetermineUnstableZD} allow to check whether the data is informative for the instability of the zero dynamics. We can combine the checks for both into an algorithm, that will return whether the provided data is informative for either the stability or the instability of the zero dynamics, or neither.

\begin{algorithm}[h]
\caption{Determine Stability of Zero Dynamics}
\label{AlgoZeroDynamics}
\begin{algorithmic}
\State \hspace*{\algorithmicindent} \textbf{Input} lag $l\coloneqq l (\mathcal{B})$, McMillan degree $n \coloneqq n(\mathcal{B})$, vector relative degree sum~$r_s:=r_s(\Bc)$, data $\mathcal{D} = (\ve{u_d},\ve{y_d})$
\vspace{0.3cm}
\State $M \coloneqq \mathcal{B}_{MPUM, [0, 2l]}$
\State $d \coloneqq \dim \ve{\pi_l} \{ \ve{u} : (\ve{u},\ve{0}) \in M \}$
\State $\{ \ve{v_1}, \hdots, \ve{v_d} \}$ basis of $\ve{\pi_l} \{ \ve{u} : (\ve{u},\ve{0}) \in M \}$
\For {$1 \leq i \leq d$}
\State Pick $\ve{u_i} \in \left( \left\{ \ve{\pi_{l+1}} \{ \ve{u} : (\ve{u}, \ve{0}) \in M \} - (\ve{v_i}, \ve{0}) \right\} \cap \left\{ \{ \ve{0}_{m(\mathcal{B})} \}^{l} \times \R^{m(\mathcal{B})} \right\} \right)_{lm(\mathcal{B})+1, \hdots, (l+1)m(\mathcal{B})}$
\EndFor
\State $\ve{V^{\dagger}} \in \R^{m(\mathcal{B})l \times d}$ right inverse of $\begin{bmatrix} \ve{v_1} \\ \vdots \\ \ve{v_d} \end{bmatrix} \in \R^{d \times m(\mathcal{B})l}$
\State $\ve{Q} \coloneqq \begin{bmatrix} \ve{v_{1,2}}, & \hdots, & \ve{v_{1,l}}, & \ve{u_1} \\ & \vdots & & \\ \ve{v_{d,2}}, & \hdots, & \ve{v_{d,l}}, & \ve{u_d} \end{bmatrix} \ve{V^{\dagger}} \in \R^{d \times d}$
\If {$\ve{Q}$ not stable}
\State return $s = -1$
\EndIf
\If {$n = \dim \left( \mathcal{H}_{l+1}(\ve{y_d}) \ker (\mathcal{H}_{l+1} (\ve{u_d})) \right)$ and $\mathcal{D}$ is informative for the vector relative degree sum~$r_s$} 
\State return $s = 1$
\Else
\State return $s = 0$
\EndIf
\vspace{0.3cm}
\end{algorithmic}
\vspace{0.4em}
\hrule
\vspace{0.3em}
{\itshape The \textbf{return} statements indicate termination of the algorithm once a result is obtained.}
\end{algorithm}

\begin{theorem}
Let $\mathcal{B} \in \mathcal{L}^{2m(\mathcal{B})}$ be an unknown square system with known lag $l(\mathcal{B})$, known McMillan degree $n(\mathcal{B}) > 0$, and known vector relative degree sum~$r_s(\Bc)$. Further, let $\mathcal{D} = (\ve{u_d}, \ve{y_d})\in\Bc_{[0,T-1]}$ be a sampled data sequence of length $T \geq 2l(\mathcal{B})+ 1$ that is informative for the existence of the vector relative degree. Then the following statements hold:
\begin{enumerate}
    \item[(i)] $\Dc$ is informative for the stability of the zero dynamics if, and only if, Algorithm~\ref{AlgoZeroDynamics} returns $s=1$.
    \item[(ii)] $\Dc$ is informative for the instability of the zero dynamics if, and only if, Algorithm~\ref{AlgoZeroDynamics} returns $s=-1$.
    \item[(iii)] $\Sigma_\Dc\cap \Sigma_{\rm stable}\neq\emptyset$ and $\Sigma_\Dc\subsetneq \Sigma_{\rm stable}$ if, and only if,  Algorithm~\ref{AlgoZeroDynamics} returns $s=0$.
\end{enumerate}
\end{theorem}

\begin{proof}
Observe that there is a unique choice for each $\ve{u_i} \in \R^{m(\mathcal{B})}$ in Algorithm~\ref{AlgoZeroDynamics}, which is the unique continuation constructed in Lemma \ref{UniqueContinuationInZD}. By definition of $\ve{Q} = (q_{k,i})_{1 \leq k,i \leq d}$, the rows satisfy $(\ve{v_{k,2}}, \hdots, \ve{v_{k,l}}, \ve{u_k}) = \sum_{i = 1}^{d} q_{k,i} \ve{v_i}$ for each $1 \leq k \leq d$. Therefore, $\ve{Q}$ is the matrix that was constructed in the proof of Lemma \ref{StabilityMPUM}. The matrix $\ve{Q}$ is not stable if, and only if, the zero dynamics of $\mathcal{B}_{MPUM}$ are not stable, which implies informativity for the instability of the zero dynamics by Lemma \ref{DetermineUnstableZD}. Conversely, if the data $\mathcal{D}$ is informative for the instability of the zero dynamics, then the zero dynamics of $\mathcal{B}_{MPUM}$ are unstable, hence $\ve{Q}$ is not stable by Lemma~\ref{StabilityMPUM}. This yields claim~(ii). If $\ve{Q}$ is stable, the prerequisite~\ref{RequireMPUMStable} of Theorem~\ref{StrictCharacterizationZeroDynamics} is satisfied. Due to the assumption of $\mathcal{D}$ being informative for the existence of the vector relative degree, in this case the conditions \ref{InformativeOfRelativeDegree} and \ref{SameDimensionN} are  equivalent to the informativity for the stability of the zero dynamics by Theorem~\ref{StrictCharacterizationZeroDynamics}. The algorithm returns $s = 1$ if, and only if, those two conditions are satisfied, yielding claim~(i). Since the set of possible outputs of the algorithm is $\{ -1, 0, 1\}$, claim~(iii) follows directly from the previous two claims.
\end{proof}

\begin{remark}
We can observe that Algorithm~\ref{AlgoZeroDynamics} relies on the previous knowledge of the system parameters $l(\mathcal{B})$ and $n(\mathcal{B})$. Assume that an estimate $N$ of the McMillan degree $n(\mathcal{B})$ is available. If $N > n(\mathcal{B})$, then condition~\ref{SameDimensionN} cannot be satisfied for any data set, hence Algorithm~\ref{AlgoZeroDynamics} cannot detect the informativity for the stability of the zero dynamics. Similarly, if $N < n(\mathcal{B})$, then the algorithm may incorrectly indicate that the data set is informative for the stability of the zero dynamics. However, in both cases, the potential result of informativity for the instability of the zero dynamics is reliable. \\
If the McMillan degree is known, and we use the results developed in~\cite{Camlibel.2024} in order to determine an upper bound of the lag of the system, then the algorithm might claim that a data set which is informative of the stability of the zero dynamics is not informative of either stability or instability of the zero dynamics. However, whenever informativity for either the stability or the instability of the zero dynamics is claimed, this result is reliable, and every data set that is informative of the instability of the zero dynamics is characterized correctly. 
\end{remark}

\begin{remark}
The arguments made in the Theorems~\ref{SISORelativgrad},~\ref{MIMOSquareRelativgrad}, and~\ref{ZeroDynamicsStability} can be extended to the case of multiple data sets using \cite[Thm.~2]{Waarde.2020}, if the data sets are collectively persistently exciting of order $L + n(\mathcal{B})$. Since in this case, all trajectories of length $L$ are characterized as the image of a composed Hankel matrix, the proofs of all statements extend, \textit{mutatis mutandis}, to this case. Similarly, all arguments from the proof of Theorem~\ref{StrictCharacterizationZeroDynamics} also apply to the case of multiple data sequences.
\end{remark}

\subsection{Numerical Example} \label{sec:Example}

To illustrate the previous algorithms, we will consider a simple numerical example. Consider the true system defined by the matrices 
\begin{equation}
\ve{A} = \begin{bmatrix}  0 & 1 & 0 & 0 \\ 0 & 0 & 1 & 0 \\ 0 & 0 & 0 & 1 \\ -0.1 & 0.5 & -1 & 1.5 \end{bmatrix}, \quad B = \begin{bmatrix} 0 \\ 0 \\ 0 \\ 1 \end{bmatrix}, \quad C = \begin{bmatrix} 2 & -3 & 1 & 0 \end{bmatrix}, \quad D = 0,
\end{equation}
and observe that it satisfies $n(\mathcal{B}) = 4 = l(\mathcal{B})$ and $r(\mathcal{B}) = 2$. Application of the input sequence 
$$ \ve{u_d} = \begin{pmatrix} 1 & 0 & 1 & 1 & 1 & 0 & 1 & 0 & 1 & 1 & 0 & 1 & 0 \end{pmatrix}, $$
which is not persistently exciting of order $l(\mathcal{B}) + n(\mathcal{B})$ and hence not sufficient for the identification of the system via Theorem~\ref{Fundamentallemma}, to the initial value $ \ve{x_0} = (1, 0, 0, 1)^{\top}$ yields the output sequence 
$$ \ve{y_d} = \begin{pmatrix}  2 & 1 & -0.6 & -2.6 & 0 & 0.2 & -0.94 & -2.35 & 0.515 & -0.3675 & 0.85275 & 0.139125 & -1.8793125 \end{pmatrix} . $$
Application of Algorithm~\ref{AlgoRelativeDegree} yields the informativity of the data sequence for the relative degree $r = 2$, while application of the Algorithm~\ref{AlgoZeroDynamics} yields the informativity for the instability of the zero dynamics. 

\begin{remark}
In real-world applications, the collected data is not noise-free, so accounting for noise is an important practical consideration. The presented algorithms are sensitive to noise, as the subspace conditions do not vary continuously with the collected data. For this reason, it is beneficial to incorporate numerical stabilization techniques in the implementation of the presented algorithms to improve their robustness and practical reliability.
\end{remark}

\section{Application to continuous-time systems}\label{secContinuousTime}

In applications, the relevant systems are generally continuous rather than discrete in time. In this section, we investigate whether it is possible to apply the results obtained in the previous sections so that conclusions about the relative degree and the stability of the zero dynamics of continuous-time systems can be drawn. 

Since data cannot be collected in continuous time, we first recall the concept of sampling by Zero-order Hold inputs, see e.g.~\cite{Pechlivanidou.2022}. 

\begin{definition} 
A linear, time-invariant continuous-time system with $m$-dimensional input and $p$-dimensional output is defined by matrices $\ve{A} \in \R^{n \times n}, \ve{B} \in \R^{n \times m}, \ve{C} \in \R^{p \times n}$ and $\ve{D} \in \R^{p \times m}$ for some $n \in \N_0$ as the set of trajectories $(\ve{u},\ve{y}) : \R_{\geq 0} \rightarrow \R^{m+p}$ satisfying
\begin{equation} \label{contisystem}
\begin{aligned}
\dot{\ve{x}}(t) &= \ve{Ax}(t) + \ve{Bu}(t) \\
\ve{y}(t) &= \ve{Cx}(t) + \ve{Du}(t)
\end{aligned}
\end{equation}
for some state trajectory $\ve{x}:\R_{\geq 0} \rightarrow \R^{n}$. 
The set of all these systems is denoted by $\mathcal{C}^{m,p}$.
\end{definition}

Observe that any linear, time-invariant continuous-time system $\mathcal{S} \in \mathcal{C}^{m,p}$ is a dynamical system (in the context of Section~\ref{Ssec:behaviors}) where $T = \R_{\geq 0}$ and $W = \R^{m+p}$. Hence the differential equation~\eqref{contisystem} is only one possible state space representation $(\ve{A}, \ve{B}, \ve{C}, \ve{D})$ for the set of trajectories (the behavior) of $\Sc$; the notion of a minimal representation carries over from the discrete-time case. Likewise, the (vector) relative degree of a continuous-time system~\eqref{contisystem} is defined as in the discrete-time case. 

Let 
\[
    \Uc_h = \{ \ve{u}: \R_{\geq 0} \rightarrow \R^m : \ve{u}_{\vert [kh, (k+1)h)} = \ve{c_k} \ \text{for some}\ \ve{c_k} \in \R^m,\  k \in \N_0 \}
\]
denote the set of piecewise constant functions with constant step size $h > 0$. The restriction of the input sequences $\ve{u}:\R_{\geq 0} \rightarrow \R^m$ to $\Uc_h$ (the Zero-order hold) allows us to associate  a state space representation $(\ve{A}, \ve{B}, \ve{C}, \ve{D})$ of a system $\mathcal{S} \in \mathcal{C}^{m,p}$ to the state space representation of a discrete-time system $\mathcal{B} \in \mathcal{L}^{m+p}$, as derived in \cite[Chap.~4.2]{Chen.1999}.

\begin{lemma}\label{Lem:sampling}
Let $\mathcal{S} \in \mathcal{C}^{m,p}$, $h> 0$ and $(\ve{A}, \ve{B}, \ve{C}, \ve{D})$ be a state space representation of $\Sc$. Define the matrices 
$$ \ve{A_d} = e^{\ve{A}h}, \ \ \ve{B_d} = \int_0^h e^{\ve{A}s} ds \ve{B},\ \ \ve{C_d} =\ve{ C}, \ \ \ve{D_d} = \ve{D} $$ and let $\mathcal{B} \in \mathcal{L}^{m+p}$ be the (discrete-time) behavior induced by the state space representation $(\ve{A_d}, \ve{B_d}, \ve{C_d}, \ve{D_d})$. Then $ (\ve{u},\ve{y}) \in \mathcal{B} $ if, and only if there is $ (\ve{\tilde{u}},\ve{\tilde{y}}) \in \mathcal{S}$ such that $\ve{\tilde{u}} \in \Uc_h $, and $(\ve{\tilde{u}}(hk), \ve{\tilde{y}}(hk)) = (\ve{u}(k), \ve{y}(k))$ for all $k \in \N_0$.
\end{lemma}

This state space representation $(\ve{A_d}, \ve{B_d}, \ve{C_d}, \ve{D_d})$ of $\mathcal{B} \in \mathcal{L}^{m+p}$ associated to  the state space representation $(\ve{A}, \ve{B}, \ve{C}, \ve{D})$ of $\mathcal{S} \in \mathcal{C}^{m,p}$ by Lemma~\ref{Lem:sampling} is called the (sampling) discretization of $(\ve{A}, \ve{B}, \ve{C}, \ve{D})$ with respect to the sampling time $h > 0$. In the following we consider the situation that, for an unknown continuous-time system, we can collect data from its discretization with respect to a certain sampling time and seek to draw conclusions about the relative degree and the stability of the zero dynamics of the continuous-time system. 

However, in this framework we face some obstacles. First of all, it is possible for two representations of two different continuous-time systems to have the same discretization, as the following example illustrates.

\begin{example}
Fix a sampling time $h > 0$ and define two continuous-time systems from $\mathcal{C}^{1,1}$ with representations  $(\ve{A}_i, \ve{B}_i, \ve{C}, \ve{D})$, $i=1,2$, by $\ve{A_1} = \begin{bmatrix} 0 & -1 \\ 1 & 0 \end{bmatrix}$, $\ve{A_2} = \begin{bmatrix} 0 & -1-\frac{2\pi}{h} \\ 1+\frac{2\pi}{h} & 0 \end{bmatrix}$, $\ve{B_2} = (1+\frac{2\pi}{h})\ve{B_1}$, and  arbitrary $\ve{B_1}, \ve{C}$, $\ve{D}$. We can determine the matrices that characterize the discretizations with respect to the sampling time $h >0$ via
\begin{align*}
e^{\ve{A_1}h} &= \begin{bmatrix} \cos(h) & -\sin(h) \\ \sin(h) & \cos(h) \end{bmatrix} =:\ve{A}, \\
e^{\ve{A_2}h} &= \begin{bmatrix} \cos(h+2\pi) & -\sin(h+2\pi) \\ \sin(h+2\pi) & \cos(h+2\pi) \end{bmatrix} = \ve{A}, \\ 
\int_{0}^h e^{\ve{A_1}h} ds \ve{B_1} &= \begin{bmatrix} \sin(h) & \cos(h) - 1 \\ - \cos(h) +1 & \sin(h) \end{bmatrix} \ve{B_1} =: \ve{B}, \\
\int_0^h e^{\ve{A_2} h} ds \ve{B_2} &= \frac{1}{1+\frac{2\pi}{h}}  \begin{bmatrix} \sin(h) & \cos(h) - 1 \\ - \cos(h) +1 & \sin(h) \end{bmatrix} \ve{B_2} = \ve{B}.
\end{align*}
Hence, both representations yield the discretization characterized by $(\ve{A},\ve{B},\ve{C},\ve{D})$, even though the induced continuous-time systems differ.

We can extend this example in order to observe that it is possible for the two representations of different continuous-time systems to have different relative degrees. Choose the matrices 
$$\ve{A_1} = \begin{bmatrix} 0 & -1 & 0 & 0\\ 1 & 0 & 0 & 0 \\ 0 & 0 & 0 & -1 \\ 0 & 0 & 1 & 0 \end{bmatrix} \text{ and } \ve{A_2} = \begin{bmatrix} 0 & -1-\frac{2\pi}{h} & 0 & 0\\ 1 + \frac{2\pi}{h} & 0 & 0 & 0 \\ 0 & 0 & 0 & -1-\frac{4\pi}{h} \\ 0 & 0 & 1+\frac{4\pi}{h} & 0 \end{bmatrix}.$$
By a similar argument as before, we can conclude that any choice $\ve{B_1} = \begin{bmatrix} b_1 \ b_2 \ b_3 \ b_4 \end{bmatrix}^{\top}$ and $\ve{B_2} = \begin{bmatrix} (1+2\pi)b_1 \ (1+2\pi)b_2 \ (1+4\pi)b_3 \ (1+4\pi)b_4 \end{bmatrix}^{\top}$ will result in equal discretizations with respect to the sampling time $h > 0$. However, choosing $\ve{C} = \begin{bmatrix}-b_3 \ 0 \ b_1 \ 0 \end{bmatrix}$ with $b_1 b_3\neq 0$, and $\ve{D} = 0$, we observe that the relative degree of the system $(\ve{A_1}, \ve{B_1}, \ve{C},\ve{D})$ is strictly greater than $1$, while the relative degree of $(\ve{A_2}, \ve{B_2}, \ve{C}, \ve{D})$ is~$1$.
\end{example}

\begin{remark}
Although an identification of the relative degree of the continuous-time system from data of its discretization is not possible in general, the case of relative degree~$0$ presents an exception. Since the matrix $\ve{D}$ is the same for the continuous-time system and each of its discretizations, the continuous-time system has relative degree $0$ if, and only if, every discretization has relative degree $0$. This implies that any sampled data set that is informative for relative degree~$0$ of the discrete-time system allows to identify the relative degree~$0$ of the continuous-time system.
\end{remark}

Similarly, we can construct an example of two representations of two differing continuous-time systems that share a discretization with respect to some sampling time $h > 0$, such that one has stable zero dynamics, while the other one does not.

\begin{example}
Consider the continuous-time systems with representations $(\ve{A}_i, \ve{B}_i, \ve{C}, \ve{D})$, $i=1,2$, defined by 
$$ \ve{A_1} = \begin{bmatrix} -1 & -3 \\ 3 & -1 \end{bmatrix},\ \ve{A_2} = \begin{bmatrix} -1 & -3-\frac{2\pi}{h} \\ 3+ \frac{2\pi}{h} & -1 \end{bmatrix},\ \ve{ B_1} = \begin{bmatrix} 0 \\ 1 \end{bmatrix},\ \ve{C} = \begin{bmatrix} -2 & \frac{1}{4} \end{bmatrix},\ \ve{D} = 0,$$
and $\ve{B_2}$ is yet to be determined such that both systems have the same discretization with respect to the sampling time $h = 0.1$. Observe that $e^{\ve{A_1}h} = e^{\ve{A_2}h}$, hence we need to choose~$\ve{B_2}$ such that 
$\displaystyle{\int_0^h e^{\ve{A_2}s} ds \ve{B_2} = \int_0^h e^{\ve{A_1} h} ds \ve{B_1}}$. Since both $\ve{A_1}$ and $\ve{A_2}$ are invertible, we can determine the matrix $\ve{B_2}$ uniquely via
\begin{align*}
\ve{B_2} &= \left( (e^{\ve{A_2}h} - \ve{I})\ve{A_2}^{-1} \right)^{-1}  (e^{\ve{A_1}h} - \ve{I}) \ve{A_1}^{-1} \ve{B_1}  = \ve{A_2} \ve{A_1}^{-1} \ve{B_1}\\
 &= \begin{bmatrix} -1 & -3-\frac{2\pi}{h} \\ 3+\frac{2\pi}{h} & -1 \end{bmatrix} \begin{bmatrix} 0.3 \\ -0.1 \end{bmatrix}  = \begin{bmatrix} 0.2 \frac{2\pi}{h}  \\ 1+0.3\frac{2\pi}{h} \end{bmatrix}.
\end{align*}
We claim that $(\ve{A}_1, \ve{B}_1, \ve{C}, \ve{D})$ has stable zero dynamics, while $(\ve{A}_2, \ve{B}_2, \ve{C}, \ve{D})$ has unstable zero dynamics. Due to the choice of the matrix $\ve{C}$, we observe that the subspace of states which induce a zero output is $\operatorname{span}\left\{ (1, 8)^\top \right\}$. At each state $\ve{x} = (1, 8)^\top \xi$ for $\xi \in \R$, we can calculate the necessary input 
$$ u(\xi) = \frac{-\ve{CA_1x}}{\ve{CB_1}} = -195\,  \xi, $$
yielding the zero dynamics
$$ \dot \xi(t) = -25\, \xi(t),\quad u(t) = -195\,  \xi(t), $$
which are stable. The same calculation for the second system yields the input
$$ u(\xi) = \frac{-\ve{CA_2x}}{\ve{CB_2}} \approx 140.685\, \xi, $$
which leads to the zero dynamics
$$ \dot \xi(t) \approx 356.3\, \xi(t),\quad u(t)  \approx 140.685\, \xi(t),$$
which are unstable.
\end{example}

We can conclude from the previous examples that it is, in general, impossible to draw conclusions about the  relative degree and the stability of the zero dynamics of the continuous-time system from information about a single discretization with respect to a sampling time $h > 0$. However, in both of the previous examples, we can interpret the difference of both continuous-time systems as an oscillation with a period being a multiple of the sampling time. Therefore, this issue is caused by the particular choice of the sampling time and can be resolved by determining the discretization with respect to another sampling time such that both have an irrational ratio.

\begin{theorem} \label{TwoFrequenciesSufficient}
Let $\mathcal{S} \in \mathcal{C}^{m,p}$ and $(\ve{A}, \ve{B}, \ve{C}, \ve{D})$ be a state space representation of $\Sc$, and let $h_1, h_2 > 0$ be two sampling times such that $\frac{h_1}{h_2} \notin \mathbb{Q}$. Further let $(\ve{A_i}, \ve{B_i}, \ve{C_i}, \ve{D_i})$ be the discretization of $(\ve{A}, \ve{B}, \ve{C}, \ve{D})$ with respect to the sampling time $h_i$, for $i=1,2$. Then for any representation $(\ve{\tilde{A}}, \ve{\tilde{B}}, \ve{\tilde{C}}, \ve{\tilde{D}})$ of a system $\tilde{\Sc} \in \mathcal{C}^{m,p}$ such that $(\ve{A_i}, \ve{B_i}, \ve{C_i}, \ve{D_i})$ is  the discretization of this representation with respect to the sampling time $h_i$ for $i=1,2$, we have that $(\ve{A}, \ve{B}, \ve{C}, \ve{D}) = (\ve{\tilde{A}}, \ve{\tilde{B}}, \ve{\tilde{C}}, \ve{\tilde{D}})$.
\end{theorem}

\begin{proof}
First observe that $\ve{C}=\ve{C}_1=\ve{C}_2=\ve{\tilde{C}}$ and $\ve{D}=\ve{D}_1=\ve{D}_2=\ve{\tilde{D}}$. In view of Lemma~\ref{Lem:sampling},
$$ e^{\ve{A}h_1} = \ve{A_1}= e^{\ve{\tilde{A}}h_1}, $$
and therefore the difference of the matrices $\ve{A}h_1$ and $\ve{\tilde{A}}h_1$ (which, in particular, have the same size) is similar to a diagonal matrix with diagonal elements from the set $\frac{2\pi i}{h_1} \mathbb{Z}$, see \cite[Ch.~6]{Horn.1991}. The same argument applies to the matrices $\ve{A}h_2$ and $\ve{\tilde{A}}h_2$, thus the eigenvalues of $\ve{A} - \ve{\tilde{A}}$ are in the set 
$$  \frac{2\pi  i}{h_1} \mathbb{Z} \cap \frac{2 \pi i}{h_2} \mathbb{Z} = \{ 0 \}, $$
since $\frac{h_1}{h_2} \notin \mathbb{Q}$ by assumption. Hence, $\ve{A} - \ve{\tilde{A}}$ is a diagonalizable matrix with spectrum $\{ 0 \}$, yielding $\ve{A} = \ve{\tilde{A}}$.

Moreover, the matrix $e^{\ve{A}s} = e^{\ve{\tilde{A}}s}$ has strictly positive eigenvalues for all $s \in \R$, hence $\int_0^{h_1} e^{\ve{A}s} ds$ has strictly positive eigenvalues and is therefore invertible. By Lemma~\ref{Lem:sampling} we  have $\int_0^{h_1} e^{\ve{A}s} ds \ve{B} = \ve{B_1} = \int_0^{h_1} e^{\ve{\tilde{A}}s} ds \ve{\tilde{B}}$, which hence implies that $\ve{B} = \ve{\tilde{B}}$, completing the proof.
\end{proof}

In order to explicitly construct a representation of the continuous-time system, additional to the two discretizations as in Theorem~\ref{TwoFrequenciesSufficient} we require a third one, such that the sampling times of all three discretizations have  rationally independent reciprocals.

\begin{proposition} \label{ConstructFromTwoFrequencies}
Let $\mathcal{S} \in \mathcal{C}^{m,p}$ be unknown and $h_1, h_2, h_3 > 0$ be known sampling times with rationally independent reciprocals, i.e., $q_1h_1^{-1} + q_2h_2^{-1}+q_3h_3^{-1} = 0$ for any $q_1,q_2,q_3\in \Q$ implies $q_1 = q_2 = q_3 = 0$. Further let $(\ve{A_i}, \ve{B_i}, \ve{C_i}, \ve{D_i})$ be the discretization of a fixed (but unknown) representation $(\ve{A},\ve{B},\ve{C},\ve{D})$ of $\Sc$ with respect to the sampling time $h_i$ for $i=1,2,3$. Then a representation of $\Sc$ can be explicitly constructed from the three discretizations.
\end{proposition}

\begin{proof}
By construction of the matrices $\ve{A_i} \in \R^{n \times n}$, there exist invertible matrices $\ve{T_1}, \ve{T_2}, \ve{T_3} \in \C^{n \times n}$ such that $$ \ve{T_1}^{-1} \ve{A_1 T_1} = e^{h_1 \ve{J}}, \ve{T_2}^{-1} \ve{A_2 T_2} = e^{h_2 \ve{J}} \text{ and } \ve{T_3}^{-1} \ve{A_3 T_3} = e^{h_3 \ve{J}} ,$$where $\ve{J} \in \C^{n \times n}$ is the Jordan canonical form of the matrix $\ve{A}$. 

By the structure of the matrix exponential of a Jordan block, the entries on the diagonal of $\ve{T_i}^{-1} \ve{A_i T_i}$ are $E_i \coloneq \{ e^{\lambda_j h_i} : \lambda_j \text{ eigenvalue of } \ve{A}, 1 \leq j \leq n \} $ for $i = 1,2,3$, counted with their multiplicity. This allows us to determine the set of eigenvalues of $\ve{A}$ and their multiplicities as follows.
We seek to show that we are able to determine the triples of eigenvalues $\left(\mu_j^{(1)}, \mu_j^{(2)}, \mu_j^{(3)}\right) \in E_1 \times E_2 \times E_3$ that correspond to the same eigenvalue $\lambda_j$ of $\ve{A}$. From these, we can  determine the eigenvalues $\lambda_j$ via 
$$\operatorname{Re}(\lambda_j) = \frac{\ln(\vert \mu_j^{(1)} \vert )}{h_1},\ \ \operatorname{Im}(\lambda_j) \in \frac{\operatorname{arg}\left(\mu_j^{(1)}\right)+2\pi \Z}{h_1} \cap \frac{\operatorname{arg}\left(\mu_j^{(2)}\right)+2\pi\Z}{h_2}.$$
The latter set is nonempty by choice of the triple $\left(\mu_j^{(1)}, \mu_j^{(2)}, \mu_j^{(3)}\right)$, and hence contains exactly one element due to the incommensurability of the sampling times $h_1, h_2$, since for any two elements $a,b$ of this set, their difference satisfies $a-b \in \frac{2\pi \Z}{h_1} \cap \frac{2 \pi \Z}{h_2} = \{0 \}$, yielding a unique solution. Observe that the absolute values of the elements $\mu_j^{(i)}$ satisfy the relation 
\begin{equation} \label{Eq:RealParts}
\left\vert \mu_j^{(1)} \right\vert^{\frac{1}{h_1}} = e^{\operatorname{Re}(\lambda_j)} = \left\vert \mu_j^{(2)} \right\vert^{\frac{1}{h_2}}. 
\end{equation}
For each $i = 1,2,3$, let $(S_{i,1}, \hdots, S_{i,k})$ denote the ordered partition of $E_i$ where $\vert \mu \vert = \vert \tilde{\mu} \vert$ for each $\mu, \tilde{\mu} \in S_{i,r}$, and $\vert \mu \vert < \vert \tilde{\mu} \vert$ for each $\mu \in S_{i,r}, \tilde{\mu} \in S_{i, \tilde{r}}$ for $r < \tilde{r}$. By equation~\eqref{Eq:RealParts}, observe that $v_r \coloneq \vert S_{i,r} \vert = \vert S_{j,r} \vert$ for any $1 \leq i,j \leq 3$ and $1 \leq r \leq k$. Seeking a contradiction, assume that there is an index $1 \leq r \leq k$ such that there are two differing pairings of the eigenvalues, defined by the bijective maps $\rho, \gamma, \tilde{\rho}, \tilde{\gamma} : \{ 1,\hdots, v_r \} \rightarrow \{ 1, \hdots, v_r \}$ operating on the indices of the sets $S_{i,r} = \{ \mu_{1}^{(i)},\hdots, \mu_{v_r}^{(i)} \}$, such that for each $1\leq i\leq v_r,$ there exist $\lambda, \tilde{\lambda} \in \C$ with
$$ e^{h_1 \lambda} = \mu_{i}^{(1)},\ e^{h_2 \lambda} = \mu_{\rho(i)}^{(2)},\ e^{h_3 \lambda} = \mu_{\gamma(i)}^{(3)} $$
and, similarly,
$$ e^{h_1 \tilde{\lambda}} = \mu_{i}^{(i)},\ e^{h_2 \tilde{\lambda}} = \mu_{\tilde{\rho}(i)}^{(2)},\ e^{h_3 \tilde{\lambda}} = \mu_{\tilde{\gamma}(i)}^{(3)}. $$ We want to show that in this case, $\lambda = \tilde{\lambda}$ has to hold. Without loss of generality, assume that $(\mu_i^{(1)}, \mu_i^{(2)}, \mu_i^{(3)})_{1 \leq i \leq v_r}$ is the first such pairing, determining the eigenvalues $\lambda_i \in \C$ uniquely due to the incommensurability of the sampling times, and let $\rho, \gamma : \{ 1, \hdots, v_r \} \rightarrow \{ 1, \hdots, v_r \} $ be the bijective maps, such that $(\mu_i^{(1)}, \mu_{\rho(i)}^{(2)}, \mu_{\gamma(i)}^{(3)})_{1 \leq i \leq v_r}$ is the second pairing, determining the eigenvalues $\tilde{\lambda}_i \in \C$. Assume without loss of generality that $\rho$ is cyclic, else restrict the following argument to the orbits of $\rho$.
Fix an arbitrary index $1 \leq i \leq v_r$. Since the eigenvalue $\mu_i^{(1)}$ corresponds to both $\lambda_i$ and $\tilde{\lambda}_i$, we can conclude that $\lambda_i - \tilde{\lambda}_i \in \frac{2 \pi i}{h_1}\Z$. Moreover, since $\mu_i^{(2)}$ corresponds to both $\lambda_i$ and $\tilde{\lambda}_{\rho^{-1}(i)}$, we can conclude $\lambda_i - \tilde{\lambda}_{\rho^{-1}(i)} \in \frac{2 \pi i}{h_2} \Z$, and similarly $\lambda_i - \tilde{\lambda}_{\gamma^{-1}(i)} \in \frac{2 \pi i}{h_3} \Z$. Fix an index $1 \leq j \leq v_r$, and let $s \in \N_0$ be such that $\rho^{-s}(j) = \gamma^{-1}(j)$. Then there is $k \in \Z$ such that
\begin{align*}
\frac{2\pi i k}{h_3} &= \lambda_j - \tilde{\lambda}_{\gamma^{-1}(j)} =
\sum_{p=1}^{s} (\lambda_{\rho^{1-p}(j)}- \lambda_{\rho^{-p}(j)})\\
&= \sum_{p=1}^{s} \big(\lambda_{\rho^{1-p}(j)} - \tilde{\lambda}_{\rho^{1-p}(j)} + (\tilde{\lambda}_{\rho^{1-p}(j)} - \lambda_{\rho^{-p}(j)})\big) \\
&= \sum_{p=1}^s \frac{2 \pi i a_p}{h_1} + \frac{2 \pi i b_p}{h_2} = \frac{2 \pi i a}{h_1} + \frac{2 \pi i b}{h_2}
\end{align*}
for some coefficients $a, b, a_p, b_p \in \Z$. This yields 
$$ 0 = \frac{a}{h_1} + \frac{b}{h_2} - \frac{k}{h_3},$$
and by the assumed rational independence, we can conclude that $\lambda_j = \tilde{\lambda}_{\gamma^{-1}(j)}.$ Hence, the set of eigenvalues $\lambda_j$ of $\ve{A}$ is uniquely determined. We can uniquely determine the Jordan normal form $\ve{J}$ of the matrix $\ve{A}$, since we know the structure of the Jordan blocks, which is the same as the structure of $\ve{T_1}^{-1} \ve{A_1 T_1}$. Choose a state space transformation defined by $\ve{T}$ such that $\ve{TJT}^{-1} \in \R^{n \times n}$ is a real matrix. 
Additionally, we can apply the state space transformation defined by $\ve{T_1}$ to the remaining matrices of the first discretization, $\ve{\tilde{B}} = \ve{T_1}^{-1} \ve{B_1}, \ve{\tilde{C}} = \ve{CT_1}$, $\ve{\tilde{D}} = \ve{D}$. With these choices, $\Big(\ve{TJT}^{-1}, \ve{T} \left( \int_{0}^{h_1} e^{\ve{J}s} ds \right)^{-1} \ve{\tilde{B}}, \ve{\tilde{C}T}^{-1}, \ve{\tilde{D}}\Big)$ is a representation of~$\mathcal{S}$.
\end{proof}

\begin{remark}
The proof of the previous proposition sketches the process how to determine the triples of eigenvalues. After determining the parts $S_{i,j}$, we can calculate the set of possible triples $(\mu_j^{(1)}, \mu_j^{(2)}, \mu_j^{(3)})$ by checking the condition
\begin{equation*}
\frac{\operatorname{arg} \left( \mu_j^{(1)} \right) + 2\pi \Z}{h_1} \cap \frac{\operatorname{arg} \left( \mu_j^{(2)} \right) + 2\pi \Z}{h_2} \cap \frac{\operatorname{arg} \left( \mu_j^{(3)} \right) + 2\pi \Z}{h_3} \neq \emptyset.
\end{equation*}
Due to the claim of the proposition, there will be a unique subset of these triples such that each eigenvalue is contained exactly once. These triples can be used to determine the corresponding eigenvalue. 
\end{remark}

\begin{remark}
A priori system knowledge may be used to weaken the requirement of sampling times with rationally independent reciprocals. Assume that a bound on the spectrum is known, say $\vert \lambda_i \vert < M$ for all $1 \leq i \leq n$ for some constant $M > 0$. Then a qualitatively sufficiently incommensurable choice of sampling times is such that each nontrivial solution of $\frac{q_1}{h_1} + \frac{q_2}{h_2} + \frac{q_3}{h_3} = 0$ with $q_i \in \Q$ satisfies $\frac{q_i}{h_i} \geq \frac{M}{\pi}$ for at least one $1\leq i \leq 3$. This ensures that $\vert \lambda_j - \tilde{\lambda}_{\gamma^{-1}(j)} \vert \geq 2M$ in the proof of Proposition~\ref{ConstructFromTwoFrequencies}, hence the solution contained in $\{ \lambda : \vert \lambda \vert < M \}$ is unique.
\end{remark}

We can combine the previous result with Theorem~\ref{Fundamentallemma} in order to derive conditions on sampled data that are sufficient to identify the continuous-time system.

\begin{corollary}\label{Cor:identifyS}
Let $\mathcal{S} \in \mathcal{C}^{m,p}$ and $(\ve{A}, \ve{B}, \ve{C}, \ve{D})$ be a state space representation of $\Sc$, and let $h_1, h_2, h_3 > 0$ be three sampling times with rationally independent reciprocals. Further let $(\ve{A_i}, \ve{B_i}, \ve{C_i}, \ve{D_i})$ be the discretization of $(\ve{A}, \ve{B}, \ve{C}, \ve{D})$ with respect to the sampling time $h_i$, for $i=1,2,3$, and denote its induced behavior by $\Bc_i$. Let $\mathcal{D}_{h_i} = (\ve{u_{d,i}}, \ve{y_{d,i}})\in\Bc_{i,[0,T_i]}$ be sampled data sequences for $i=1,2,3$. If all three systems $\mathcal{B}_i$ are controllable, and the input sequences $\ve{u_{d,i}}$  each are persistently exciting of the orders $n(\mathcal{B}_i) + l(\mathcal{B}_i) + 1$, respectively, then the data set $\mathcal{D} = \mathcal{D}_{h_1} \cup \mathcal{D}_{h_2} \cup \mathcal{D}_{h_3}$ allows to explicitly construct a representation of $\Sc$, that has the same state space dimension as $(\ve{A}, \ve{B}, \ve{C}, \ve{D})$.
\end{corollary}

\begin{proof}
Observe that the requirements on the two discretizations are sufficient for identifying them by Theorem~\ref{Fundamentallemma}. Then the assertion follows from Proposition \ref{ConstructFromTwoFrequencies}.
\end{proof}

The conditions provided in Corollary~\ref{Cor:identifyS} are sufficient to identify the underlying continuous-time system, hence they also allow to identify the relative degree of the system and to determine either the stability or the instability of its zero dynamics. Nevertheless, these conditions are not necessary and further research should reveal weaker conditions on the three data sets $\mathcal{D}_{h_1}$, $\mathcal{D}_{h_2}$ and $\mathcal{D}_{h_3}$, which allow to characterize the informativity for the (vector) relative degree and the informativity for the stability of the zero dynamics.

\section{Conclusion}\label{secConclusion}

In this work, we have developed new methods to identify the (vector) relative degree and the stability of the zero dynamics of a linear system from a data set of a sampled trajectory. The methods build upon two powerful tools from data-driven control theory: the Fundamental Lemma by Willems et al.\ and the data informativity framework. While the former allows to draw more direct conclusions, the latter comes with weaker conditions on the data; in particular, persistently exciting signals are not required. As one of the main results, we proposed an algorithm that determines whether the provided data set sampled from a discrete-time SISO system $\mathcal{B} \in \mathcal{L}^{2}$ is informative for its relative degree $r \in \N_0$, and have extended the result in order to determine the vector relative degree of MIMO systems in many cases. Further, we have developed conditions under which a data set is informative for the stability of the zero dynamics of the system.

Combining both results allows us to characterize the relative degree and the stability of the zero dynamics of sampling discretizations obtained from continuous-time systems via Zero-order Hold. We have shown how the continuous-time system can be reconstructed from three discretizations obtained at suitable sampling times. This allows to study the design and feasibility of certain control strategies such as sliding mode control and funnel control. Future work will focus on developing weaker conditions on the data to identify the required properties of the continuous-time system, possibly using other discretization methods, and on the extension of the results to nonlinear systems.

\section*{Declarations}

\textbf{Competing Interests} The authors declare no competing interests.

\bibliographystyle{abbrv}
\bibliography{lit}

@article{Bisoffi.2022,
 author = {Bisoffi, Andrea and De Persis, Claudio and Tesi, Pietro},
 year = {2022},
 title = {{Data-driven control via Petersen's lemma}},
 pages = {110537},
 volume = {145},
 issn = {0005-1098},
 journal = {{Automatica}},
 doi = {10.1016/j.automatica.2022.110537}
}

@article{Buffington.1998,
 author = {Buffington, James M. and Enns, Dale F. and Teel, Andrew R.},
 year = {1998},
 title = {{Control Allocation and Zero Dynamics}},
 pages = {458--464},
 volume = {21},
 number = {3},
 journal = {{Journal of Guidance, Control, and Dynamics}},
 doi = {10.2514/2.4258}
}

@article{C.DePersis.2020,
 author = {De Persis, Claudio and Tesi, Pietro},
 year = {2020},
 title = {{Formulas for Data-Driven Control: Stabilization, Optimality, and Robustness}},
 pages = {909--924},
 volume = {65},
 number = {3},
 issn = {1558-2523},
 journal = {{IEEE Transactions on Automatic Control}},
 doi = {10.1109/TAC.2019.2959924}
}

@article{Camlibel.2024,
 author = {Camlibel, Kanat and Rapisarda, Paolo},
 year = {2024},
 title = {{Beyond the fundamental lemma: from finite time series to linear system}},
 doi = {10.48550/arXiv.2405.18962},
 journal = {arXiv preprint}
}

@book{Chen.1999,
 author = {Chen, Chi-Tsong},
 year = {1999},
 title = {{Linear system theory and design}},
 price = {paper : {\pounds} 25.00},
 address = {New York},
 edition = {3. ed.},
 publisher = {{Oxford Univ. Press}},
 isbn = {0-19-511777-8},
 series = {{The Oxford series in electrical and computer engineering}},
doi = {10.1002/1099-1239(20001230)10:15<1360::AID-RNC539>3.0.CO;2-C}
}

@article{Desborough.2002,
 author = {Desborough, Lane and Miller, Randy},
 year = {2002},
 title = {{Increasing Customer Value of Industrial Control Performance Monitoring - Honeywell's Experience}},
 number = {326},
 volume = {98},
pages = { 169-–189},
 journal = {{AIChE Symposium Series}},
}

@article{Ender.1993,
 author = {Ender, David B.},
 year = {1993},
 title = {{Process Control Performance: Not as Good as You Think.}},
 journal = {{Control Engineering}}, 
volume = {40},
pages = {180--190},
}

@incollection{Heij.1990,
 author = {Heij, Christiaan},
 title = {{System identifiability for the procedure of most powerful unfalsified modelling}},
 pages = {437--447},
 publisher = {{Birkh{\"a}user}},
 editor = {Kaashoek, M. A. and {van Schuppen}, J. H. and Ran, A. C. M.},
 booktitle = {{Realization and Modelling in System Theory}},
 series = {Progress in Systems and Control Theory},
 volume = 3,
 year = {1990},
 address = {Boston, MA},
 doi = {10.1007/978-1-4612-3462-3\textunderscore49}
}

@book{Horn.1991,
 year = {1991},
 title = {{Topics in Matrix Analysis}},
 address = {Cambridge},
 publisher = {{Cambridge University Press}},
 isbn = {9780521467131},
 editor = {Horn, Roger A. and Johnson, Charles R.},
doi = {10.1017/CBO9780511840371},
}

@article{Hou.2013,
 author = {Hou, Zhong-Sheng and Wang, Zhuo},
 year = {2013},
 title = {{From model-based control to data-driven control: Survey, classification and perspective}},
 pages = {3--35},
 volume = {235},
 issn = {0020-0255},
 journal = {{Information Sciences}},
 doi = {10.1016/j.ins.2012.07.014}
}

@article{Hu.2023,
 author = {Hu, Xiao and Wang, Ping and Cai, Shuo and Zhang, Lin and Hu, Yunfeng and Chen, Hong},
 year = {2023},
 title = {{Vehicle Stability Analysis by Zero Dynamics to Improve Control Performance}},
 pages = {2365--2379},
 volume = {31},
 number = {6},
 journal = {{IEEE Transactions on Control Systems Technology}},
 doi = {10.1109/TCST.2023.3257682}
}

@article{Ilchmann.2002,
title={Tracking with prescribed transient behaviour}, volume={7}, DOI={10.1051/cocv:2002064}, journal={ESAIM: Control, Optimisation and Calculus of Variations}, author={Ilchmann, Achim and Ryan, E. P. and Sangwin, C. J.}, year={2002}, pages={471–493}
}

@book{Isidori.1995,
 author = {Isidori, Alberto},
 year = {1995},
 title = {{Nonlinear Control Systems}},
 address = {London},
 edition = {3rd},
 publisher = {Springer},
 isbn = {978-1-4471-3909-6},
 series = {{Communications and Control Engineering Series}},
 doi = {10.1007/978-1-84628-615-5}
}

@article{IvanMarkovsky.2008,
 author = {Markovsky, Ivan and Rapisarda, Paolo},
 year = {2008},
 title = {{Data-driven simulation and control}},
 pages = {1946--1959},
 volume = {81},
 journal = {{International Journal of Control}},
doi = {10.1080/00207170801942170},
}

@article{J.Berberich.2021,
 author = {Berberich, Julian and K{\"o}hler, Johannes and M{\"u}ller, Matthias A. and Allg{\"o}wer, Frank},
 year = {2021},
 title = {{Data-Driven Model Predictive Control With Stability and Robustness Guarantees}},
 pages = {1702--1717},
 volume = {66},
 number = {4},
 issn = {1558-2523},
 journal = {{IEEE Transactions on Automatic Control}},
 doi = {10.1109/TAC.2020.3000182}
}

@article{J.C.Willems.2004,
 author = {Willems, Jan C.  and Rapisarda, Paolo and Markovsky, Ivan and De Moor, Bart L.M.},
 journal = {Systems \& Control Letters},
 title = {{A note on persistency of excitation}},
 pages = {325--329},
 volume = 54,
 number = 4,
 year = {2005},
 doi = {10.1016/j.sysconle.2004.09.003}
}

@inproceedings{J.Coulson.2019,
 author = {Coulson, Jeremy and Lygeros, John and D{\"o}rfler, Florian},
 title = {{Data-Enabled Predictive Control: In the Shallows of the DeePC}},
 pages = {307--312},
 booktitle = {{18th European Control Conference (ECC)}},
 year = {2019},
 doi = {10.23919/ECC.2019.8795639}
}

@article{Levant.2003,
 author = {Levant, Arie},
 year = {2003},
 title = {{High-order sliding modes: Differentiation and output feedback control}},
 pages = {427--434},
 volume = {76},
 journal = {{International Journal of Control}},
doi = {10.1080/0020717031000099029},
}

@article{Ling.2012,
 author = {Ling, Rui and Wu, Meirong and Dong, Yan and Chai, Yi},
 year = {2012},
 title = {{High order sliding-mode control for uncertain nonlinear systems with relative degree three}},
 pages = {3406--3416},
 volume = {17},
 number = {8},
 issn = {1007-5704},
 journal = {{Communications in Nonlinear Science and Numerical Simulation}},
 doi = {10.1016/j.cnsns.2011.12.017}
}

@article{Markovsky.2021,
 author = {Markovsky, Ivan and D{\"o}rfler, Florian},
 year = {2021},
 title = {{Behavioral systems theory in data-driven analysis, signal processing, and control}},
 pages = {42--64},
 volume = {52},
 issn = {1367-5788},
 journal = {{Annual Reviews in Control}},
 doi = {10.1016/j.arcontrol.2021.09.005}
}

@article{Maupong.2017,
 author = {Maupong, T. M. and Rapisarda, P.},
 year = {2017},
 title = {{Data-driven control: A behavioral approach}},
 pages = {37--43},
 volume = {101},
 issn = {0167-6911},
 journal = {{Systems {\&} Control Letters}},
 doi = {10.1016/j.sysconle.2016.04.006}
}

@article{Mueller.2009,
 author = {Mueller, Markus},
 year = {2009},
 title = {{Normal form for linear systems with respect to its vector relative degree}},
 pages = {1292--1312},
 volume = {430},
 number = {4},
 issn = {0024-3795},
 journal = {{Linear Algebra and its Applications}},
 doi = {10.1016/j.laa.2008.10.014}
}

@article{Pechlivanidou.2022,
 author = {Pechlivanidou, Georgia and Karampetakis, Nicholas},
 year = {2022},
 title = {{Zero-order hold discretization of general state space systems with input delay}},
 pages = {708--730},
 volume = {39},
 number = {2},
 issn = {1471-6887},
 journal = {{IMA Journal of Mathematical Control and Information}},
 doi = {10.1093/imamci/dnac005}
}

@inproceedings{Persis.2023,
  author = {De Persis, Claudio and Gadginmath, D. and Pasqualetti, Fabio and Tesi, Pietro},
 year = {2023},
 title = {{Data-Driven Feedback Linearization with Complete Dictionaries}},
booktitle = {Proceedings of the 62nd IEEE Conference on Decision and Control},
 doi = {10.1109/CDC49753.2023.10383720},
pages = {3037--3042}
}

@BOOK{PoldWill98,
   AUTHOR    = {Polderman, Jan Willem and Willems, Jan C.},
   YEAR      = 1998,
   TITLE     = {Introduction to Mathematical Systems Theory. A Behavioral Approach},
   PUBLISHER = {Springer},
   Address   = {New York},
   doi       = {10.1007/978-1-4757-2953-5},
}

@book{SiraRamirez.2015,
 author = {{Sira Ram{\'i}rez}, Hebertt J.},
 year = {2015},
 title = {{Sliding mode control: The delta-sigma modulation approach}},
 address = {Cham},
 publisher = {Birkh{\"a}user},
 series = {{Control Engineering Series}},
doi = {10.1007/978-3-319-17257-6},
}

@article{Waarde.2020,
author={van Waarde, Henk J. and De Persis, Claudio and Camlibel, M. Kanat and Tesi, Pietro},
  journal={IEEE Control Systems Letters}, 
  title={Willems’ Fundamental Lemma for State-Space Systems and Its Extension to Multiple Datasets}, 
  year={2020},
  volume={4},
  number={3},
  pages={602-607},
  doi={10.1109/LCSYS.2020.2986991}}

@article{Waarde.2020b,
 author = {van Waarde, Henk J. and Eising, Jaap and Trentelman, Harry and Camlibel, Kanat},
 year = {2020},
 title = {{Data Informativity: A New Perspective on Data-Driven Analysis and Control}},
 pages = {4753--4768},
 volume = {65},
number = {11},
 issn = {1558-2523},
 journal = {{IEEE Transactions on Automatic Control}},
 doi = {10.1109/TAC.2020.2966717}
}

@article{Willems.1986a,
 author = {Willems, Jan C.},
 year = {1986},
 title = {{From time series to linear system---Part I. Finite dimensional linear time invariant systems}},
 pages = {561--580},
 volume = {22},
 number = {5},
 issn = {0005-1098},
 journal = {{Automatica}},
 doi = {10.1016/0005-1098(86)90066-X}
}

@article{Willems.1986b,
 author = {Willems, Jan C.},
 year = {1986},
 title = {{From time series to linear system---Part II. Exact modelling}},
 pages = {675--694},
 volume = {22},
 number = {6},
 journal = {{Automatica}},
doi = {10.1016/0005-1098(86)90005-1},
}

@article{Willems.1987,
 author = {Willems, Jan C.},
 year = {1987},
 title = {{From time series to linear system---Part III. Approximate modelling}},
 pages = {87--115},
 volume = {23},
 number = {1},
 journal = {{Automatica}},
doi = {10.1016/0005-1098(87)90120-8},
}

@inproceedings{Willems.1997,
 author = {Willems, Jan C.},
 title = {{Fitting Data Sequences to Linear Systems}},
 pages = {405--416},
 publisher = {{Birkh{\"a}user}},
 isbn = {978-1-4612-4120-1},
 editor = {Byrnes, Christopher I. and Datta, Biswa N. and Martin, Clyde F. and Gilliam, David S.},
 booktitle = {{Systems and Control in the Twenty-First Century}},
 year = {1997},
 address = {Boston, MA},
doi = {10.1007/978-1-4612-4120-1_22},
  series       = {Systems \& Control: Foundations \& Applications},
  volume       = {22}
}

@article{Ziegler.1942,
 author = {Ziegler, J. G. and Nichols, N. B.},
 year = {1942},
 title = {{Optimum Settings for Automatic Controllers}},
 pages = {759--765},
 number = {64},
 journal = {{Transactions of the ASME}}
}

@Article{Berger.2018,
  author           = {Thomas Berger and Huy Ho{\`{a}}ng L{\^{e}} and Timo Reis},
  year             = {2018},
  journal          = {Automatica},
  title            = {Funnel control for nonlinear systems with known strict relative degree},
  doi              = {10.1016/j.automatica.2017.10.017},
  pages            = {345--357},
  volume           = {87},
}

@article{Berger.2021,
  title={Funnel control of nonlinear systems},
  author={Berger, Thomas and Ilchmann, Achim and Ryan, Eugene P.},
  journal={Math. Control Signals Syst.},
  volume={33},
  pages={151--194},
  year={2021},
    doi              = {10.1007/s00498-021-00277-z}
}

@article{Berger.2025,
    title = {Funnel control - A survey},
    author = {Thomas Berger and Achim Ilchmann and Eugene P. Ryan},
    journal = {Annual Reviews in Control},
    year = {2025},
    volume = 60,
    pages = {Article 101024},
    doi = {10.1016/j.arcontrol.2025.101024},
}

@Book{ShteEdwa14,
  author    = {Y. Shtessel and C. Edwards and L. Fridman and A. Levant},
  year      = {2014},
  title     = {Sliding Mode Control and Observation},
  publisher = {Birkh\"{a}user},
  address = {New York, NY},
  series    = {Control Engineering Series},
doi = {10.1007/978-0-8176-4893-0},
}

@INPROCEEDINGS{ByrnWill84,
   AUTHOR    = {Byrnes, Christopher I. and Willems, Jan C.},
   YEAR      = 1984,
   TITLE     = {Adaptive stabilization of multivariable linear systems},
   BOOKTITLE = {Proc. 23rd~{IEEE} Conf. Decis. Control},
   Pages     = {1574--1577},
doi = {10.1109/CDC.1984.272346},
}

@ARTICLE{KhalSabe87,
   AUTHOR    = {Khalil, Hassan and Saberi, Ali},
   YEAR      = 1987,
   TITLE     = {Adaptive stabilization of a class of nonlinear systems using high-gain feedback},
   JOURNAL   = {{IEEE} Trans. Autom. Control},
   Volume    = 32,
number = {11},
   Pages     = {1031--1035},
doi = {10.1109/TAC.1987.1104481},
}

@ARTICLE{Mare84,
   AUTHOR    = {Mareels, Iven},
   YEAR      = 1984,
   TITLE     = {A simple selftuning controller for stably invertible systems},
   JOURNAL   = {Syst. Control Lett.},
   Volume    = 4,
   Number    = 1,
   Pages     = {5--16},
doi = {10.1016/S0167-6911(84)80045-6},
}

@INCOLLECTION{Mors83,
   AUTHOR    = {Morse, A. Stephen},
   YEAR      = 1983,
   TITLE     = {Recent problems in parameter adaptive control},
   BOOKTITLE = {Outils et Mod\`{e}les Math\'{e}matiques pour l'Automatique, l'Analyse de Syst\`{e}mes et le Traitment du Signal},
   PUBLISHER = {\'Editions du Centre National de la Recherche Scientifique (CNRS)},
   Pages     = {733--740},
   Editor    = {Landau, I. D.},
   Address   = {Paris},
   Volume    = 3,
}

@ARTICLE{IlchWirt13,
   AUTHOR    = {Ilchmann, Achim and Wirth, Fabian},
   YEAR      = 2013,
   TITLE     = {On minimum phase},
   JOURNAL   = {at - Automatisierungstechnik},
   Volume    = 61,
number = 12,
   Pages     = {805--817},
doi = {10.1515/auto.2013.1002},
}

\end{document}